\documentclass[11pt]{amsart}

\usepackage{amsmath,amssymb,amsthm,amsfonts,dsfont}
\usepackage{bm,graphicx,hyperref,cancel}
\usepackage[capitalise]{cleveref}
\usepackage{mathtools}
\usepackage{multirow}
\usepackage{braket}

\newcommand{\bvec}[1]{\ensuremath{\mathbf{#1}}}

\newcommand{\ZZ}{\mathbb{Z}}

\newcommand{\RR}{\mathbb{R}}
\newcommand{\CC}{\mathbb{C}}

\newcommand{\One}{\mathds{1}}

\newcommand{\fraksl}{{\mathfrak{sl}}}

\newcommand{\vq}{{\mathbf{q}}}

\DeclareMathOperator{\spn}{span}

\renewcommand{\Re}{\operatorname{Re}}
\renewcommand{\Im}{\operatorname{Im}}
\DeclareMathOperator{\diag}{diag}

\usepackage{etoolbox}
\newtoggle{stdmargins}

\usepackage{tcolorbox}

\newcommand{\dd}{\,\ensuremath{\textrm{d}}}

\usepackage{tikz}
\usetikzlibrary{arrows,shapes,calc,positioning}

\usepackage[calc]{datetime2}
\DTMnewdatestyle{mydate}{%
}
\DTMsetdatestyle{mydate}

\usepackage{xcolor}
\definecolor{scarlet}{RGB}{255, 36, 0}
\definecolor{purp}{RGB}{160, 32, 240}

\newcommand{\idty}{\mathds{1}}

\usepackage{enumitem}

\newtheorem{assumption}{Assumption}
\newtheorem{theorem}{Theorem}
\newtheorem*{theorem*}{Theorem}
\newtheorem*{conjecture*}{Conjecture}
\newtheorem{lemma}{Lemma}
\newtheorem*{lemma*}{Lemma}
\newtheorem{proposition}{Proposition}
\newtheorem*{proposition*}{Proposition}
\newtheorem{corollary}{Corollary}

\newtheorem{remark}{Remark}

\theoremstyle{definition}
\newtheorem{definition}{Definition}

\usepackage[margin=1in]{geometry}

\crefname{proposition}{Proposition}{Propositions}

\usepackage[boxsize=1em,centertableaux]{ytableau}
\DeclareMathOperator{\Sym}{Sym}
\DeclareMathOperator{\gen}{gen}

\title[The Many-Body Ground State Manifold of the FBI Hamiltonian for MATBG]{The Many-Body Ground State Manifold of Flat Band Interacting Hamiltonian for Magic Angle Twisted Bilayer Graphene}

\author{Kevin~D. Stubbs}
\address{Department of Mathematics, University of California, Berkeley, CA 94720, USA}
\email{kevin.d.stubbs@gmail.com}

\author{Michael Ragone}
\address{Department of Mathematics, University of California, Berkeley, CA 94720, USA}
\email{micragone@berkeley.edu}

\author{Allan H. MacDonald}
\address{Department of Physics, University of Texas at Austin, Austin, Texas 78712, USA}
\email{macdpc@physics.utexas.edu}

\author{Lin Lin}
\address{Department of Mathematics, University of California, Berkeley, CA 94720, USA}
\address{Applied Mathematics and Computational Research Division, Lawrence Berkeley National Laboratory, Berkeley, CA 94720, USA}
\email{linlin@math.berkeley.edu}

\allowdisplaybreaks

\begin{document}

\maketitle

\begin{abstract}
At a magic relative twist angle, magic angle twisted bilayer graphene (MATBG) has an octet of flat bands that can host strong correlation physics when partially filled.
A key theoretical discovery in MATBG is the existence of ferromagnetic Slater determinants as exact ground states of the corresponding flat band interacting (FBI) Hamiltonian. The FBI Hamiltonian describes the behavior of electrons that interact with each other in a high-dimensional space, and is constructed from the band structure of the non-interacting Bistritzer--MacDonald model at the chiral limit.
A key property of the FBI Hamiltonian for MATBG is that it is frustration free and can be written as a sum of non-commuting terms.
In this work, we provide a complete characterization of the ground state manifold  of the FBI Hamiltonian, proving that it is precisely the linear span of such ferromagnetic Slater determinants.

\end{abstract}

\section{Introduction}

Magic angle twisted bilayer graphene (MATBG), as well as other moir\'e materials, provide an exciting new platform towards improved understanding of materials with strong electron-electron interactions \cite{AndreiMacDonald2020,AndreiEfetovJarilloHerreroEtAl2021}.
Recent experiments on MATBG have shown a wealth of strongly correlated behavior such as correlated insulator phases at integer fillings and superconducting phases away from integer fillings \cite{2018Nature,Serlin,CaoFatemiDemir2018,LuStepanovYang2019,YankowitzChenPolshyn2019}.
In a certain simplified limit, called the chiral limit, the non-interacting model for MATBG gains an additional symmetry which allows one to rigorously prove that this model has exactly  flat bands \cite{TarnopolskyKruchkovVishwanath2019, BeckerEmbreeWittstenEtAl2021, becker2023chiral,BeckerHumbertZworski2023,BeckerHumbertZworski2022a, BeckerHumbertZworski2022, WatsonLuskin2021}. 
The Bloch states  can be labeled by chirality, valley, and spin quantum numbers, each of which has two possible values, $\pm 1$, resulting in a flat band octet. Based on these flat bands, previous work has proposed flat band interacting (FBI) model, which is an interacting model for MATBG at the chiral limit~\cite{BultinckKhalafLiuEtAl2020,ChatterjeeBultinckZaletel2020,SoejimaParkerBultinckEtAl2020,XieMacDonald2020,WuSarma2020,DasLuHerzog-Arbeitman2021,BernevigSongRegnaultEtAl2021,LiuKhalafLee2021,SaitoGeRademaker2021,JiangLaiWatanabe2019,PotaszXieMacDonald2021,LiuKhalafLeeEtAl2021,FaulstichStubbsZhuEtAl2023,BeckerLinStubbs2023}.
One of the most remarkable features of the FBI Hamiltonian for MATBG is that the  Hamiltonian includes strong long range electron-electron interactions, yet there exist ferromagnetic type Slater determinant states (also called Hartree--Fock states), which are exact ground states of the  Hamiltonian \cite{BeckerLinStubbs2023,KangVafek2019,BernevigSongRegnaultEtAl2021}.
This result is surprising: when electron-electron interactions are strong one would typically expect the ground states are far from being a Slater determinant, such as in the case of Hubbard models.

The FBI Hamiltonian with 8 flat bands (also called the spinful, valleyful model) is relatively complex to study.
Within the 8 flat bands, there are two 4 dimensional subspaces related to each other by the \(C_{2z} \mathcal{T}\) symmetry, which is a combination of \(180^{\circ}\) in-plane rotation and spinless time reversal symmetry.
The FBI Hamiltonian commutes with rotations within each of these 4 dimensional subspaces so this model has a large \(\mathsf{U}(4) \times \mathsf{U}(4)\) symmetry; often referred to as the ``hidden symmetry'' of MATBG \cite{BultinckKhalafLiuEtAl2020,LianSongRegnaultEtAl2021}.
As a consequence of this symmetry, given any Slater determinant ground state, applying a $\mathsf{U}(4) \times \mathsf{U}(4)$ transformation yields another Slater determinant ground state.
Furthermore, any linear combination of these Slater determinant ground states is also a ground state.

To facilitate analysis, simplified models have also been considered in the literature:

\begin{enumerate}

\item A spinless, valleyless model (also called the single spin, single valley model) that involves only the chiral flavor.

\item A spinless, valleyful model (also called the single spin, valleyful model) that includes both chiral and valley flavors.

\end{enumerate}

 To the best of our knowledge, it remains an open question whether the ground state manifold of the FBI Hamiltonian is fully spanned by linear combinations of Slater determinant ground states, both in the full model and in the simplified models above.

\subsection{Contribution}

The main result of this paper is to provide an affirmative answer to the question above, which was conjectured in \cite{StubbsBeckerLin2024}:

\begin{theorem}[Main Result, informal version of~\cref{thm:single-valley,prop:valleyful-case,prop:valleyful-spinful-case}]
  \label{thm:main-result}
  A many-body state is a ground state of the FBI Hamiltonian for MATBG if and only if it can be written as a linear combination of Slater determinant states, each of which is a ground state of the FBI Hamiltonian.
\end{theorem}

Our results consist of three main components: \cref{thm:single-valley} addresses the ground states of the spinless, valleyless model, \cref{prop:valleyful-case} examines the spinless, valleyful model, and \cref{prop:valleyful-spinful-case} analyzes the full spinful, valleyful model. In all cases, we prove that any ground state must be expressible as a linear combination of Slater determinant states, each of which is itself a ground state of the FBI Hamiltonian. The dimension of the ground state manifold differs in the three settings and we provide an explicit dimension count in all these cases.

It is known that the Slater determinant ground states are \emph{uniformly half-filled}, meaning that the number of electrons is uniform at each momentum point $\bvec{k}$ in the Brillouin zone. Consequently, any linear combination of these Slater determinant states must also be uniformly half-filled. Therefore our first step is to establish that any state in the ground state manifold must be uniformly half-filled in all three settings. A key insight is that any ground state must vanish under the application of chiral transfer operators; this imposes strong constraints for a many-body state to be a ground state of the FBI Hamiltonian for MATBG (\cref{prop:equiv-characterization}). The proof of~\cref{prop:equiv-characterization} makes use of the properties of the Jacobi-$\theta$ functions, which appear in other contexts when studying lowest Landau level physics~\cite{LedwithTarnpolskyKhalaf2020,MeraOzawa2024}.

Following this result, the characterization of the ground state manifold is straightforward in the spinless, valleyless case since there are only two possible Slater determinant ground states. However, the complexity increases significantly when a continuous ``hidden symmetry'' is present: $\mathsf{U}(2) \times \mathsf{U}(2)$ in the spinless, valleyful case and $\mathsf{U}(4) \times \mathsf{U}(4)$ in the spinful, valleyful case. To fully characterize the ground state manifold in these cases, we employ the highest weight theorem (\cref{thm:theorem-of-highest-wt}) from representation theory. Ultimately, we show that the ground state manifold is fully determined by a special irreducible representation of the $\mathsf{U}(2) \times \mathsf{U}(2)$ or $\mathsf{U}(4) \times \mathsf{U}(4)$ group acting on a tensor product space. This irreducible representation can be generated from any vector, such as a Slater determinant ground state. 
The dimension of the ground state manifold then follows directly from the hook length formula~\cite{FultonHarris2004}.

\subsection{Discussion}

Our results significantly extend the previous findings in \cite{BeckerLinStubbs2023,StubbsBeckerLin2024}, which are the only rigorous results on interacting electrons in MATBG to our knowledge, and focus solely on determining whether non-ferromagnetic Slater determinant ground states can exist. Notably, the techniques employed in \cite{BeckerLinStubbs2023,StubbsBeckerLin2024} differ entirely from this work. The proof in \cite{BeckerLinStubbs2023} assumes the absence of translation symmetry breaking and relies on an algebraic condition that must be satisfied by all nontrivial orthogonal projectors. While this condition can be computationally challenging to verify in general, it has potential applicability beyond MATBG. In contrast, \cite{StubbsBeckerLin2024} introduces a more computationally tractable condition specifically for MATBG, leveraging a set of Sylvester equations to explicitly rule out solutions that break translation symmetry. Neither of these approaches is applicable in this work for characterizing the many-body ground state manifold.
It is also worth noting that our theoretical results align with previous numerical studies, which have shown that the ground states of MATBG are well described by Slater determinant states well beyond this idealized setting~\cite{FaulstichStubbsZhuEtAl2023,SoejimaParkerBultinckEtAl2020,ChatterjeeBultinckZaletel2020,KwanWagnerBultinck2021,BultinckKhalafLiuEtAl2020}.

In the spinless, valleyless case, where the only degree of freedom is a chirality pseudospin, we prove that the system is essentially an Ising ferromagnet along the chiral axis.  Although one might expect, on heuristic grounds, that an interaction term with chirality-dependent form factors would select the fully polarized states, it is highly nontrivial to rule out the existence of any other degenerate many-body states. When valley and spin degrees of freedom are included, the situation becomes even richer. We prove that these degrees of freedom do not introduce any additional ``magnetic'' type anisotropy beyond that already present in the chirality sector. The study of the many-body magnetic anisotropy beyond the chiral model is an interesting direction for future research.

Our proof of \cref{prop:equiv-characterization} encounters subtle difficulties at the thermodynamic limit. Specifically, to prove \cref{prop:equiv-characterization}, we show that a certain linear transformation is invertible. However, with our current techniques, this argument breaks down in the thermodynamic limit even in the spinless, valleyless model, suggesting the possibility of new emergent behavior that is not captured by the present theory. We note that the vanishing neutral excitation gap is not solely due to the presence of Goldstone modes, which correspond to zero-energy excitations arising from continuous symmetry \cite{BernevigLianCowsik2022}.

Another important open question is understanding the energy gap between the ground state manifold and the remaining states, both in neutral excitations and in few-particle excitations. Recent studies have suggested the existence of a family of skyrmionic states, whose pairing could give rise to superconducting behavior \cite{ChatterjeeBultinckZaletel2020,KwanWagnerBultinck2021}. While neutral excitations are crucial for assessing the stability of Slater determinant ground states, few-particle excitations can provide deeper insight into the superconductivity observed away from integer fillings.

\subsection{Outline}
The remainder of this paper is organized as follows.
In \cref{sec:non-inter-model}, we recall some basic facts about the non-interacting Bistritzer-MacDonald model for MATBG at the chiral limit.
Additionally, in \cref{sec:valley-spin-degrees}, we discuss the chiral, valley, and spin degrees of freedom.
While these degrees of freedom do not meaningfully change the properties of the non-interacting model, they play an important role in the interacting model.
In \cref{sec:inter-model-chir}, we use the notations introduced in~\cref{sec:non-inter-model} to define the flat-band interacting (FBI) Hamiltonian for chiral MATBG introduced in previous work \cite{BeckerLinStubbs2023}.
We also introduce the \textit{chiral transfer operators} which allow us to consider models with valley and spin on the same footing as those without valley or spin.
Consequently, most of our results will be stated in terms of the chiral transfer operators.

After a brief overview describing the nature of the many-body ground states in \cref{sec:nature-many-body}, we prove the first main technical result, \cref{prop:equiv-characterization}, in~\cref{sec:equiv-characterization}.
Then in \cref{sec:repr-theory-ground}, we begin by discussing the implications of~\cref{prop:equiv-characterization}.
In particular, in \cref{sec:occup-vect-repr}, introduce the \textit{occupation vector representation} of many-body state which is used to directly apply the representation theory to many-body states.
The occupation vector representation, together with tools from the representation theory In~\cref{sec:tools-from-repr}, are then used in \cref{sec:valleyful-case,sec:vall-spinf-case} to prove the main result for the valleyful case and the valleyful and spinful case respectively.

This paper also contains three appendixes.
In~\cref{sec:lemma-linear-indep,sec:proof-gs-prop-2a}, we prove two key lemmas which are used as part of~\cref{sec:equiv-characterization}. Additional background materials on representation theory are provided in \cref{sec:tools-from-repr-appendix}.

\subsection*{Acknowledgments}
This work was supported by the Simons Targeted Grants in Mathematics and Physical Sciences on Moir\'e Materials Magic (K.D.S., A.H.M., L.L.). L.L. is a Simons Investigator in Mathematics. We thank Simon Becker and Eslam Khalaf for helpful discussions.

\section{Non-Interacting Model}
\label{sec:non-inter-model}

In this section, we introduce a few necessary facts about the non-interacting Bistritzer-MacDonald model for MATBG at the chiral limit, referred to as  chiral MATBG for short. For a more complete discussion we refer the readers to \cite{BeckerEmbreeWittsten2022} and \cite[Section 2]{BeckerLinStubbs2023}.
This model is described by a differential operator $H(\alpha)$, for some \(\alpha \in \mathbb{R}\), densely defined on $H^{1}(\RR^{2}; \CC^{4})$ with the following form:
\begin{equation}
  \label{eq:h-bm}
  H(\alpha) =
  \begin{bmatrix}
    0 & D(\alpha)^{\dagger} \\
    D(\alpha) & 0 
  \end{bmatrix}
  \quad
  \text{where}
  \quad
  D(\alpha) =
  \begin{bmatrix}
    i (\partial_{x_{1}} + i \partial_{x_{2}}) & \alpha e^{i x_{2}} U(\bvec{r}) \\[.5ex]
    \alpha e^{- i x_{2}} U(-\bvec{r}) & i (\partial_{x_{1}} + i \partial_{x_{2}}) + i
  \end{bmatrix}
\end{equation}
This Hamiltonian is associated with a lattice $\Gamma \subseteq \mathbb{R}^2$, called the moir{\'e} lattice, with generating vectors $\bvec{a}_{1} = \begin{bmatrix} 1 & 0 \end{bmatrix}^{\top}$ and $\bvec{a}_{2}  = \begin{bmatrix} -1/2 & \sqrt{3}/2\end{bmatrix}^{\top}$.
The parameter $\alpha$ is inversely proportional to the twist angle and \(U(\bvec{r})\) is a real analytic potential satisfying the following properties (where $\omega := e^{2 \pi i / 3}$ and $R_{3}$ denotes clockwise rotation by $\frac{2 \pi}{3}$)
\begin{equation}
  \label{eq:u-bm-properties}
  \begin{split}
    U(\bvec{r} + m \bvec{a}_{1} + n \bvec{a}_{2}) & = \omega^{m + n} U(\bvec{r}) \\
    U(R_{3} \bvec{r}) & = \omega U(\bvec{r}) \\
    U(x_{1}, x_{2}) & = \overline{U(x_1, -x_2)}
  \end{split}
\end{equation}
Using the properties of $U$, it can be verified that $H(\alpha)$ is periodic with respect to the lattice $\Gamma$.
Let $\Gamma^{*}$ denote the dual lattice to $\Gamma$ and let $\Omega^{*} := \RR^{2} / \Gamma^{*}$ denote the unit cell of the dual lattice.

Since $H(\alpha)$ is periodic with respect to $\Gamma$, we can define the Bloch-Floquet transform of $H_{\bvec{k}}(\alpha) := e^{-i \bvec{k} \cdot \bvec{r}} H(\alpha) e^{i \bvec{k} \cdot \bvec{r}}$ which is densely defined on $H^{1}(\RR^{2}; \CC^{4})$ and has compact resolvent for all $\bvec{k} \in \Omega^{*}$.
By Bloch's theorem, for each $\bvec{k} \in \Omega^{*}$ there exists a discrete set $\mathcal{I} \subseteq \ZZ$ and a set of functions $\{ u_{n \bvec{k}}(\bvec{r}) : \bvec{k} \in \Omega^{*}, n \in \mathcal{I} \}$ so that 
\begin{equation}
  \label{eq:h-bm-bloch}
  \begin{split}
   & \hspace{2em} H_{\bvec{k}}(\alpha) u_{n \bvec{k}}(\bvec{r}) = \epsilon_{n \bvec{k}} u_{n \bvec{k}}(\bvec{r}) \\
   &\hspace{2em} u_{n \bvec{k}}(\bvec{r} + \bvec{a}) = u_{n \bvec{k}}(\bvec{r}) \qquad \forall \bvec{a} \in \Gamma^{*} \\
   & \int_{\RR^{2}} \braket{u_{m \bvec{k}}(\bvec{r}), u_{n \bvec{k}}(\bvec{r})}_{\CC^{4}} \dd{\bvec{r}} = \delta_{mn}
  \end{split}
\end{equation}
where $\braket{\cdot, \cdot}_{\CC^{4}}$ denotes the inner product between vectors in $\CC^{4}$ and for each $\bvec{k} \in \Omega^{*}$ the Bloch eigenvalues $\epsilon_{n \bvec{k}}$ satisfy $\epsilon_{(n-1) \mathbf{k}} \leq \epsilon_{n \mathbf{k}} \leq \epsilon_{(n+1) \mathbf{k}}$ for all $n \in \mathcal{I}$ (i.e. they are ordered in non-decreasing order).
A key fact about about the Hamiltonian defined in~\cref{eq:h-bm} is that when $U$ satisfies the properties~\cref{eq:u-bm-properties}, then the Hamiltonian $H(\alpha)$ can exhibit exactly flat bands:
\begin{theorem}[\cite{BeckerEmbreeWittsten2022,TarnopolskyKruchkovVishwanath2019,WatsonLuskin2021}]
  Let $H(\alpha)$ be as in~\cref{eq:h-bm} where $U$ is a real analytic function satisfying the symmetry conditions~\cref{eq:u-bm-properties} and let $\epsilon_{n \bvec{k}}$ be the corresponding Bloch eigenvalues for $H(\alpha)$.
  There exists a non-empty discrete set $\mathcal{A} \subseteq \CC$ so that $\alpha \in \mathcal{A}$ implies the set
  \begin{equation}
    \mathcal{N} = \left\{ n \in \mathcal{I} : \epsilon_{n \bvec{k}} = 0, ~\forall \bvec{k} \in \Omega^{*} \right\}
  \end{equation}
  is non-empty.
\end{theorem}
As shown in \cite[Theorem 3]{BeckerHumbertZworski2023}, Hamiltonians of the form $H(\alpha)$ can have either two or four flat bands depending on the choice of $U(\mathbf{r})$.
In this work, we assume that $H(\alpha)$ has two flat bands which is consistent with the number of bands found in experiment \cite{2018Nature}:
\begin{assumption}
  \label{assume:bands}
  We assume that the potential $U$ and parameter $\alpha$ are chosen so that $H(\alpha)$ has two flat bands (i.e. $\# | \mathcal{N}|  = 2$).
\end{assumption}

\subsection{Chiral, Valley, and Spin Degrees of Freedom}
\label{sec:valley-spin-degrees}
The Hamiltonian in~\cref{eq:h-bm} satisfies 
  \begin{equation}
    \begin{bmatrix}
      I_{2 \times 2} & \\
      & - I_{2 \times 2}
    \end{bmatrix}
    H(\alpha)
    \begin{bmatrix}
      I_{2 \times 2} & \\
      & - I_{2 \times 2}
    \end{bmatrix}
    =
    - H(\alpha).
  \end{equation}
  Thus we can choose the eigenstates corresponding to  zero eigenvalues of \(H(\alpha)\)  so that they are also eigenstates of \(\diag{(1,1,-1,-1)}\).
  In particular, the eigenfunctions of \(H(\alpha)\) can be written as \(\begin{bmatrix} w(\bvec{r}) & 0 \end{bmatrix}^{\top}\) and \(\begin{bmatrix} 0 & v(\bvec{r}) \end{bmatrix}^{\top}\) for some functions \(w(\bvec{r}), v(\bvec{r})\) with values in \(\mathbb{C}^2\).
  This partitions the zero energy subspace of \(H(\alpha)\) into two classes referred to as the two ``chiral'' sectors of MATBG.
  By~\cref{assume:bands}, \(H(\alpha)\) has two flat bands so we can subdivide this two dimensional subspace into two 1 dimensional subspaces using chiral symmetry.
  Furthermore, \(H(\alpha)\) commutes with an additional symmetry \(C_{2z} \mathcal{T}\), which a combination of \(180^{\circ}\) in-plane rotation and spinless time reversal symmetry. Algebraically, the \(C_{2z} \mathcal{T}\) symmetry acts on $\CC^4$ valued functions as follows:
  \begin{equation}
      C_{2z} \mathcal{T} 
      \begin{bmatrix}
      w(\bvec{r}) \\[.5ex]
      v(\bvec{r})
      \end{bmatrix}
      =
      \begin{bmatrix}
      \overline{v(-\bvec{r})} \\[.5ex]
      \overline{w(-\bvec{r})} 
      \end{bmatrix}.
  \end{equation}
  In particular,  the \(C_{2z} \mathcal{T}\) symmetry maps the two chiral sectors onto each another.

A physically complete description of MATBG has two additional degrees of freedom which are not included in the definition of $H(\alpha)$: valley and spin.
To understand the valley degree of freedom, recall that \cref{eq:h-bm} is derived by performing a low energy expansion near one of the Dirac points of the monolayer graphene and setting some of the hopping coefficients to zero~\cite{TarnopolskyKruchkovVishwanath2019}.
However, monolayer graphene has two Dirac points so we could have equally well performed an expansion about the other Dirac point.
We refer to the expansion about the first Dirac point as the first valley and the second Dirac point as the second valley.
Since the two valleys are related by time reversal symmetry, the corresponding Hamiltonian which expands about both valleys takes the form:
\begin{equation}
  H^{(v)}(\alpha) =
  \begin{bmatrix}
    H(\alpha) & \\
    & \overline{H(\alpha)}
  \end{bmatrix}.
\end{equation}
Note that the spectral properties of $H^{(v)}(\alpha)$ are completely determined by $H(\alpha)$ and that under~\cref{assume:bands}, $H^{(v)}(\alpha)$ has four flat bands which are further subdivided into two 2 dimensional chiral subspaces related by \(C_{2z} \mathcal{T}\).

To understand the spin degree of freedom, we recall that electrons are spin-1/2 particles and have two spin states $\uparrow$ or $\downarrow$.
Experiments have shown that spin-orbit coupling in graphene is extremely small so we may neglect the dependence on spin \cite{HuertasHernandoGuineaBrataas2006,AndreiMacDonald2020}.
Hence, the non-interacting Hamiltonian including both valley and spin takes the form:
\begin{equation}
  H^{(v,s)}(\alpha) =
  \begin{bmatrix}
    H(\alpha) & \\
    & \overline{H(\alpha)}
  \end{bmatrix} \otimes I_{2 \times 2}
\end{equation}
where the additional tensor product with identity is due to the inclusion of spin.
Similar to $H^{(v)}(\alpha)$, the spectral properties $H^{(v,s)}(\alpha)$ are completely determined by $H(\alpha)$ and under~\cref{assume:bands}, $H^{(v,s)}(\alpha)$ has eight flat bands which are further subdivided into two 4 dimensional chiral subspaces which are related by \(C_{2z} \mathcal{T}\).

In the interacting model, valley and spin degrees of freedom significantly enhance the ground state degeneracy.
In the following sections, we will analyze three important cases: the single valley case, the valleyful case, and the valleyful and spinful case.
These cases correspond to $H(\alpha)$, $H^{(v)}(\alpha)$, $H^{(v,s)}(\alpha)$ respectively. 

\section{Flat band interacting model}
\label{sec:inter-model-chir}

The flat band interacting model (FBI) Hamiltonian models the chiral MATBG with electron-electron interactions, which can be defined using second quantization in momentum space.
In this section, we first introduce the basic setup of second quantized Hamiltonians momentum space and then specialize to chiral MATBG in~\cref{sec:form-factor-chiral,sec:form-factor-chiral-valley-spin}.
For a more complete discussion, we refer the readers to \cite[Section 3]{BeckerLinStubbs2023}.

Suppose that $H$ is a Hamiltonian with a non-empty set of flat bands $\mathcal{N}$ and flat band Bloch eigenfunctions $\{ u_{n \bvec{k}}(\bvec{r}) : n \in \mathcal{N}, \bvec{k} \in \Omega^{*} \}$ where $u_{n \bvec{k}} : \Omega \to \CC^{L}$ for some positive integer $L$. 
For each $u_{n \bvec{k}}$ in this set, we associate the creation and annihilation operators $\hat{f}_{n \bvec{k}}^{\dagger}$, $\hat{f}_{n \bvec{k}}$ which creates or annihilates the state\footnote{The additional factor of $e^{i \bvec{k} \cdot \bvec{r}}$ is so that the created or annihilated state is an eigenfunction of the original Hamiltonian, $H$, and not the Bloch transformed Hamiltonian $e^{- i \bvec{k} \cdot \bvec{r}} H e^{i \bvec{k} \cdot \bvec{r}}$.} $e^{i \bvec{k} \cdot \bvec{r}} u_{n \bvec{k}}(\bvec{r})$.
These operators satisfy the boundary conditions
\begin{equation}
  \hat{f}_{n (\bvec{k} + \bvec{G})}^{\dagger} = \hat{f}_{n \bvec{k}}^{\dagger}, \qquad \hat{f}_{n (\bvec{k} + \bvec{G})} = \hat{f}_{n \bvec{k}},
\end{equation}
as well as the canonical anticommutation relations:
\begin{equation}
  \{ \hat{f}_{n \bvec{k}}, \hat{f}_{m \bvec{k}'} \} = 0 \qquad \{ \hat{f}_{n \bvec{k}}^{\dagger}, \hat{f}_{m \bvec{k}'} \} = 0 \qquad \{ \hat{f}_{n \bvec{k}}^{\dagger}, \hat{f}_{m \bvec{k}'} \} = \delta_{mn}\, \delta_{\bvec{k} - \bvec{k}' \in \Gamma^{*}}.
\end{equation}

The FBI model is defined in terms of the form factor, $\Lambda_{\bvec{k}}(\bvec{q} + \bvec{G})$, which is a $\# | \mathcal{N} | \times \# | \mathcal{N} |$ matrix defined as follows:
\begin{equation}
  [\Lambda_{\bvec{k}}(\bvec{q} + \bvec{G})]_{mn}
  = \frac{1}{|\Omega|} \int_{\RR^{2}} e^{i \bvec{G} \cdot \bvec{r}} \braket{ u_{m \bvec{k}}(\bvec{r}), u_{n (\bvec{k} + \bvec{q})}(\bvec{r}) }_{\CC^{L}} \dd{\bvec{r}} 
\end{equation}
where $\bvec{k}, \bvec{q} \in \Omega^{*}$, $\bvec{G} \in \Gamma^{*}$, and $\braket{\cdot, \cdot}_{\CC^{L}}$ denotes the inner product of vectors in $\CC^{L}$.

Our next step towards defining an FBI is to pick a finite set of momenta $\mathcal{K} \subset \Omega^{*}$.
In particular, we choose basis vectors $\bvec{b}_{1}, \bvec{b}_{2}$ for the dual lattice $\Gamma^{*}$, fix integers $n_{\bvec{k}_{x}}, n_{\bvec{k}_{y}} > 0$, and define $\mathcal{K}$ as:
\begin{equation}
  \label{eq:mp-grid}
  \mathcal{K} := \left\{ \frac{i}{n_{k_{x}}} \bvec{b}_1 + \frac{j}{n_{k_{y}}} \bvec{b}_2 : i \in \{0, 1, \cdots, n_{k_{x}} - 1\}, j \in \{0, 1, \cdots, n_{k_{y}} - 1 \} \right\} \subseteq \Omega^*.
\end{equation}
We also define $N_{\bvec{k}} := \# |\mathcal{K}| = n_{\bvec{k}_{x}} n_{\bvec{k}_{y}}$, the number of momentum points.
Note that the set of points $\mathcal{K} + \Gamma^{*}$ form a uniform sampling of $\RR^{2}$ which is closed under inversion (i.e. $\bvec{k} \in \mathcal{K} + \Gamma^{*}$ implies $- \bvec{k} \in \mathcal{K} + \Gamma^{*}$.
For our main result we will require that both $n_{\bvec{k}_{x}}$ and $n_{\bvec{k}_{y}}$ are sufficiently large and finite (see~\cref{assume:grid-spacing} for the precise assumption).

Given these definitions, the corresponding FBI Hamiltonian is
\begin{align}
  \hat{H}_{FBI}
  & = \frac{1}{N_{\bvec{k}} |\Omega|} \sum_{\bvec{q}' \in \mathcal{K} + \Gamma^{*}}^{} \hat{V}(\bvec{q}') \widehat{\rho}(\bvec{q}') \widehat{\rho}(-\bvec{q}'), \label{eq:h-fbi} \\
  \widehat{\rho}(\bvec{q}')
  & = \sum_{\bvec{k} \in \mathcal{K}}^{} \sum_{m,n \in \mathcal{N}}^{} [\Lambda_{\bvec{k}}(\bvec{q}')]_{mn} \left( \hat{f}_{m \bvec{k}}^{\dagger} \hat{f}_{n (\bvec{k} + \bvec{q}')} - \frac{1}{2} \delta_{\bvec{q}' \in \Gamma^{*}} \delta_{mn} \right) \label{eq:rho-q}
\end{align}
where $\hat{V}(\bvec{q}')$ is any radial, continuous function so that $\hat{V}(\bvec{q}') > 0$ for all $\bvec{q}' \in \RR^{2}$.
For MATBG, a common choice is the double gate-screened Coulomb interaction 
\begin{equation}
  \hat{V}(\vq') = \frac{2\pi}{\epsilon} \frac{\tanh(|\bvec{q}'| d/2)}{|\bvec{q}'|}.
\end{equation}
Here $\epsilon, d > 0$ parameterize the strength and length of the screened Coulomb interaction, respectively (see e.g., \cite[Appendix C]{BernevigSongRegnaultEtAl2021}).

Before proceeding, we recall the number operator
$\hat{N} := \sum_{\bvec{k} \in \mathcal{K}} \sum_{m \in \mathcal{N}} \hat{f}_{m \bvec{k}}^{\dagger} \hat{f}_{m \bvec{k}}$.
It can be easily checked that $[\hat{H}_{FBI}, \hat{N} ] = 0$ so any eigenstate of $\hat{H}_{FBI}$ may be written as a linear combination of eigenstates of $\hat{N}$.
Given a choice of momentum grid $\mathcal{K}$, we say a many-body state is \textit{half-filled}, if it satisfies $\hat{N} \ket{\Phi} = \frac{1}{2} (\# | \mathcal{N} |) N_{\bvec{k}} \ket{\Phi}$ since the number of electrons in $\ket{\Phi}$ is exactly half the number of the possible creation operators.

\subsection{Form Factor of the Spinless, Valleyless Model}
\label{sec:form-factor-chiral}
While the FBI Hamiltonian is defined with an orthogonal basis $\{ u_{n \bvec{k}}(\bvec{r}) : n \in \mathcal{N}, \bvec{k} \in \mathcal{K} \}$, the choice of this orthogonal basis is arbitrary so we could also define this FBI Hamiltonian using a different orthogonal basis.
This observation amounts to rotating each creation and annihilation operator by a family of unitaries $\{ U(\bvec{k}) \}_{\bvec{k} \in \mathcal{K}} \subset \mathsf{U}(\# | \mathcal{N} |)$ as follows:
\begin{equation}
  \hat{f}_{m \bvec{k}}^{\dagger} \mapsto \sum_{}^{} \hat{f}_{n \bvec{k}}^{\dagger} [ U(\bvec{k}) ]_{n m} \qquad 
  \hat{f}_{m \bvec{k}} \mapsto \sum_{}^{} \hat{f}_{n \bvec{k}} [ \overline{U(\bvec{k})} ]_{n m}.
\end{equation}
Under this basis change the form factor transforms as $\Lambda_{\bvec{k}}(\bvec{q}') \mapsto U(\bvec{k}) \Lambda_{\bvec{k}}(\bvec{q}') U(\bvec{k} + \bvec{q})$ and one can verify that $\hat{H}_{FBI}$ remains unchanged.
In the physics context, this freedom in the basis choice is known as a ``gauge freedom'' and a choice of rotation is known as a ``choice of gauge''.

As proven in \cite[Lemma 4.4]{BeckerLinStubbs2023}, using the symmetries of the Bistritzer-MacDonald Hamiltonian, there exists a choice of gauge so that the form factor is diagonal for all $\bvec{k}, \bvec{q}'$ and satisfies a symmetry relation:
\begin{proposition}
  \label{prop:form-factor-gauge-fixing}
  Let $H(\alpha)$ be as in~\cref{eq:h-bm} and suppose~\cref{assume:bands} holds.
  There exists a choice of gauge so that for all $\bvec{k} \in \Omega^{*}$ and all $\bvec{q}' \in \RR^{2}$ the form factor can be written as:
  \begin{equation}
    \label{eq:form-factor}
    \Lambda_{\bvec{k}}(\bvec{q}') =
    \begin{bmatrix}
      a_{\bvec{k}}(\bvec{q}') & \\
                              & \overline{a_{\bvec{k}}(\bvec{q}')}
    \end{bmatrix}
    \quad \text{where} \quad a_{\bvec{k}}(\bvec{q}') \in \CC.
  \end{equation}
  Additionally, in this choice of gauge for all $\bvec{G} \in \Gamma^{*}$
  \begin{equation}
    \label{eq:form-factor-particle-hole}
    \overline{a_{\bvec{k}}(\bvec{G})} = a_{-\bvec{k}}(\bvec{G}).
  \end{equation}
\end{proposition}
Under the choice of gauge from \cref{prop:form-factor-gauge-fixing}, we label the two flat bands as $\mathcal{N} = \{ +, - \}$ which we refer to the $+$ band and the $-$ band.
In this choice of gauge, $\widehat{\rho}(\bvec{q}')$ (\cref{eq:rho-q}) can be written in a particularly simple form:
\begin{equation}
  \widehat{\rho}(\bvec{q}') = \sum_{\bvec{k} \in \mathcal{K}}^{} a_{\bvec{k}}(\bvec{q}') \left( \hat{f}_{+ \bvec{k}}^{\dagger}\hat{f}_{+ (\bvec{k} + \bvec{q})} - \frac{1}{2} \delta_{\bvec{q} \in \Gamma^{*}} \right) + \overline{a_{\bvec{k}}(\bvec{q}')} \left( \hat{f}_{- \bvec{k}}^{\dagger}\hat{f}_{- (\bvec{k} + \bvec{q})} - \frac{1}{2} \delta_{\bvec{q} \in \Gamma^{*}} \right). 
\end{equation}
The choice of gauge in \cref{prop:form-factor-gauge-fixing} is also referred to as the sublattice gauge in the literature because the two Bloch eigenfunctions \(u_{+ \bvec{k}}(\bvec{r})\) and \(u_{- \bvec{k}}(\bvec{r})\) are supported on the A and B sublattices respectively.

\subsection{Form Factor for with Valley and Spin Degrees of Freedom}
\label{sec:form-factor-chiral-valley-spin}
Following the setup~\cite{BultinckKhalafLiuEtAl2020}, the form factor also can be written in diagonal form when we include the valley and spin degrees of freedom.
In particular, since the two valleys are related by time reversal symmetry, when we include the valley degree of freedom the form factor takes the form
\begin{equation}
  \label{eq:form-factor-valley}
  \Lambda_{\bvec{k}}^{(v)}(\bvec{q}') =
  \begin{bmatrix}
    a_{\bvec{k}}(\bvec{q}') &&& \\
                            & \overline{a_{\bvec{k}}(\bvec{q}')} && \\
                            && \overline{a_{\bvec{k}}(\bvec{q}')} \\
                            &&& a_{\bvec{k}}(\bvec{q}') \\
  \end{bmatrix}
\end{equation}
where the first term on the diagonal corresponds to the $+$ band in the $K$ valley, the second term corresponds to the $-$ band in the $K$ valley, the third $(+, K')$, and the fourth $(-, K')$.

Identifying $K = +1$ and $K' = -1$, we can generate all many body states applying the following creation operators on the vacuum state:
\begin{equation}
  \Big\{ \hat{f}_{(m, \tau), \bvec{k}}^{\dagger} : m \in \{ +, - \},~ \tau \in \{ +, - \},~ \bvec{k} \in \mathcal{K} \Big\}.
\end{equation}

When we include both the valley and spin degrees of freedom, the form factor can be written
\begin{equation}
  \label{eq:form-factor-valley-spin}
  \Lambda_{\bvec{k}}^{(v,s)}(\bvec{q}') =
  \begin{bmatrix}
    a_{\bvec{k}}(\bvec{q}') &&& \\
                            & \overline{a_{\bvec{k}}(\bvec{q}')} && \\
                            && \overline{a_{\bvec{k}}(\bvec{q}')} \\
                            &&& a_{\bvec{k}}(\bvec{q}') \\
  \end{bmatrix} \otimes I_{2 \times 2}.
\end{equation}
The corresponding set of generating creation operators is given by
\begin{equation}
  \Big\{ \hat{f}_{(m, \tau, s), \bvec{k}}^{\dagger} : m \in \{ +, - \},~ \tau \in \{ +, - \},~s \in \{ \uparrow, \downarrow\},~ \bvec{k} \in \mathcal{K} \Big\}
\end{equation}
where $m$ indexes the bands, $\tau$ indexes the valley, and $s$ indexes spin.

In all cases, the form factor $\Lambda_{\bvec{k}}(\bvec{q}')$ is diagonal whose non-zero entries are either $a_{\bvec{k}}(\bvec{q}')$ or $\overline{a_{\bvec{k}}(\bvec{q}')}$.
We refer to the terms multiplied by $a_{\bvec{k}}(\bvec{q}')$ as the ``$+$ chiral sector'' and the terms multiplied by $\overline{a_{\bvec{k}}(\bvec{q}')}$ as the ``$-$ chiral sector''.
Note that when including the valley degree of freedom, the creation operators for the $+$ chiral sector have $m \tau = +1$ and the creation operators for the $-$ chiral sector have $m \tau = -1$.

To unify the FBI model in three cases, we define the \emph{chiral transfer operator} $\hat{C}_{\pm, \bvec{k}, \bvec{k}'}$ for any $\bvec{k}, \bvec{k}' \in \Omega^{*}$ as follows:
  \begin{equation}
    \label{eq:chiral-transfer-operators}
    \begin{split}
      \hat{C}_{\pm, \bvec{k}, \bvec{k}'} & = \hat{f}_{\pm, \bvec{k}}^{\dagger} \hat{f}_{\pm, \bvec{k}'}, \\[1ex]
      \hat{C}_{\pm, \bvec{k}, \bvec{k}'}^{(v)} & = \sum_{m \tau = \pm}^{} \hat{f}_{(m, \tau), \bvec{k}}^{\dagger} \hat{f}_{(m, \tau), \bvec{k}'}, \\
      \hat{C}_{\pm, \bvec{k}, \bvec{k}'}^{(v,s)} & = \sum_{s \in \{ \uparrow, \downarrow \}}^{} \sum_{m\tau = \pm}^{} \hat{f}_{(m, \tau, s), \bvec{k}}^{\dagger} \hat{f}_{(m, \tau, s), \bvec{k}'}.
    \end{split}
  \end{equation}
  The special case $\bvec{k} = \bvec{k}'$ is particularly important so we define the special notation $\hat{n}_{\pm, \bvec{k}} := \hat{C}_{\pm, \bvec{k}, \bvec{k}}$ which we refer to as the \textit{chiral number operators}.
  The operators $\hat{n}_{\pm, \bvec{k}}^{(v)}$ and $\hat{n}_{\pm, \bvec{k}}^{(v,s)}$ are defined analogously.

For reasons which will be made clear in~\cref{coro:uniform-filling}, we define the integer $N_{occ}$ where $N_{occ} = 1$ for the spinless, valleyless model, $N_{occ} = 2$ for the spinless, valleyful model, and $N_{occ} = 4$ for the full model.
With these definitions and suppressing the superscripts on the chiral transfer operators, $\widehat{\rho}(\bvec{q}')$ in all cases simplifies to
\begin{equation}
  \label{eq:rho-q-tbg}
  \widehat{\rho}(\bvec{q}')
  = \sum_{\bvec{k} \in \mathcal{K}}^{} a_{\bvec{k}}(\bvec{q}') \left( \hat{C}_{+, \bvec{k}, (\bvec{k} + \bvec{q}')} - \frac{1}{2} N_{occ}\, \delta_{\bvec{q}' \in \Gamma^{*}} \right) + \overline{a_{\bvec{k}}(\bvec{q}')} \left( \hat{C}_{-, \bvec{k}, (\bvec{k} + \bvec{q}')} - \frac{1}{2} N_{occ}\, \delta_{\bvec{q}' \in \Gamma^{*}} \right).
\end{equation}
This unified form of $\widehat{\rho}(\bvec{q}')$ allows us to analyze the ground states of $\hat{H}_{FBI}$ in all three cases simultaneously.

In summary, the FBI Hamiltonian for the chiral MATBG in \cref{eq:h-fbi} with $\widehat{\rho}(\bvec{q}')$ defined in \cref{eq:rho-q-tbg} is called

\begin{enumerate}

\item Spinless, valleyless if $\hat{C}_{\pm, \bvec{k}, \bvec{k}'} = \hat{f}_{\pm, \bvec{k}}^{\dagger} \hat{f}_{\pm, \bvec{k}'}$ and $N_{occ} = 1$;

\item Spinless, valleyful if $\hat{C}_{\pm, \bvec{k}, \bvec{k}'} = \sum_{m \tau = \pm}^{} \hat{f}_{(m, \tau), \bvec{k}}^{\dagger} \hat{f}_{(m, \tau), \bvec{k}'}$  and $N_{occ} = 2$;

\item Spinful, valleyful if $\hat{C}_{\pm, \bvec{k}, \bvec{k}'} = \sum_{s \in \{ \uparrow, \downarrow \}}^{} \sum_{m\tau = \pm}^{} \hat{f}_{(m, \tau, s), \bvec{k}}^{\dagger} \hat{f}_{(m, \tau, s), \bvec{k}'}$   and $N_{occ} = 4$.

\end{enumerate}

\section{Properties of the Many-Body Ground States}
\label{sec:nature-many-body}
Before proving the main result in~\cref{sec:equiv-characterization,sec:repr-theory-ground}, let us discuss a few properties of the many-body ground states of $\hat{H}_{FBI}$ shared in all of the three cases.
The Hamiltonian \(\hat{H}_{FBI}\)  \eqref{eq:h-fbi} has the following three properties:
\begin{enumerate}
\item \(\hat{H}_{FBI}\) can be written as a sum of  terms $\hat{H}_{FBI}=\sum_{\bvec{q}'} H(\vq')$. The terms $\{H(\vq')\}$ do not commute with each other and cannot be simultaneously diagonalized.
\item Each term $H(\vq')$ is a positive semidefinite operator. Therefore \(\hat{H}_{FBI}\) is also positive semidefinite.
\item There exist many-body states, \(\ket{\Phi}\), such that $H(\vq') \ket{\Phi} =0$ for all $\vq'$. This implies  \(\hat{H}_{FBI} \ket{\Phi} = 0\).
\end{enumerate}
A Hamiltonian that satisfies properties (2) and (3) is called \emph{frustration-free}. Combined with (1), this means that \(\hat{H}_{FBI}\) is a non-commuting, frustration-free Hamiltonian.

The structure of the FBI Hamiltonian \eqref{eq:h-fbi} further implies that \(\ket{\Phi}\) is a ground state if and only if \(\widehat{\rho}(\bvec{q}') \ket{\Phi} = 0\) for all \(\bvec{q}' \in \mathcal{K} + \Gamma^{*}\) since
\begin{equation}
  \braket{\Phi | \hat{H}_{FBI} | \Phi } = \sum_{\bvec{q}'}^{} \hat{V}(\bvec{q}') \| \widehat{\rho}(\bvec{q}') \ket{\Phi} \|^{2}
\end{equation}
and \(V(\bvec{q}') > 0\) for all \(\bvec{q}' \in \mathbb{R}^{2}\).
Frustration-freeness plays a central role in the characterization of the ground states.

As a first consequence of frustration-freeness, we recall the per-$\bvec{k}$ number operator, $\hat{N}_{\bvec{k}}$, and the total number operator, $\hat{N}$, which can be defined in terms of the chiral number operators, $\hat{n}_{\pm, \bvec{k}}$ as follows
\begin{equation}
    \hat{N}_{\bvec{k}} = \hat{n}_{+,\bvec{k}} + \hat{n}_{-, \bvec{k}}, \quad
    \hat{N} = \sum_{\bvec{k} \in \mathcal{K}}^{} \hat{N}_{\bvec{k}}.
\end{equation}
For $\bvec{q}' = \bvec{0}$, it's easily verified that
\begin{equation}
  \widehat{\rho}(\bvec{0}) = \left( \hat{N} - N_{occ} N_{\bvec{k}} \right).
\end{equation}
Therefore, the frustration-freeness of \(\hat{H}_{FBI}\) immediately implies the following proposition:
\begin{proposition}
  \label{prop:total-number}
  A many-body state $\ket{\Phi}$ is a ground state of $\hat{H}_{FBI}$ only if $\hat{N} \ket{\Phi} = N_{occ} N_{\bvec{k}} \ket{\Phi}$.
\end{proposition}
While~\cref{prop:total-number} holds for all frustration-free FBI Hamiltonians, the specific properties of form factors for chiral MATBG allows us to prove something significantly stronger:
\begin{proposition}
  \label{prop:uniform-chiral-filling}
  If $\ket{\Phi}$ is a many-body ground state of \(\hat{H}_{FBI}\) then $\ket{\Phi}$ may be decomposed as
  \begin{equation}
    \label{eq:mbgs-chiral-decomp}
    \ket{\Phi} = \sum_{\ell}^{} c_{\ell} \ket{\Phi^{(\ell)}},
  \end{equation}
  where the states $\ket{\Phi^{(\ell)}}$ are joint eigenfunctions of $\{ \hat{n}_{\pm, \bvec{k}} : \bvec{k} \in \mathcal{K} \}$ satisfying the following for all $\bvec{k} \in \mathcal{K}$
  \begin{equation}
    \begin{split}
      \hat{n}_{+, \bvec{k}} \ket{\Phi^{(\ell)}} & = \lambda_{+}^{(\ell)} \ket{\Phi^{(\ell)}} \\[.5ex]
      \hat{n}_{-, \bvec{k}} \ket{\Phi^{(\ell)}} & = \lambda_{-}^{(\ell)} \ket{\Phi^{(\ell)}}.
    \end{split}
  \end{equation}
\end{proposition}
The important content of~\cref{prop:uniform-chiral-filling} is not only that the ground state can be decomposed as joint eigenfunctions of $\hat{n}_{\pm, \bvec{k}}$, but that the eigenvalues $\lambda_{\pm}^{(\ell)}$ are \textit{independent} of $\bvec{k}$.
An immediate corollary of this result is that any many-body ground state $\ket{\Phi}$ is uniformly filled:
\begin{corollary}
  \label{coro:uniform-filling}
  If $\ket{\Phi}$ is a many-body ground state of \(\hat{H}_{FBI}\) then $\hat{N}_{\bvec{k}} \ket{\Phi} = N_{occ} \ket{\Phi}$ for all $\bvec{k} \in \mathcal{K}$.
Hence $N_{occ}$ is the number of electrons which are occupied at each momentum $\bvec{k} \in \mathcal{K}$.
\end{corollary}
\begin{proof}
  Suppose that $\ket{\Phi}$ is a ground state and has a decomposition as in~\cref{eq:mbgs-chiral-decomp}.
  We may assume without loss of generality that this decomposition is chosen so that each term is mutually orthogonal (i.e. $\braket{\Phi^{(\ell)}, \Phi^{(\ell')}} = \delta_{\ell, \ell'}$).
  For each $\bvec{k} \in \mathcal{K}$, we have that
\begin{equation}
  \hat{N}_{\bvec{k}} \ket{\Phi^{(\ell)}} = \left( \hat{n}_{+, \bvec{k}} + \hat{n}_{-, \bvec{k}} \right) \ket{\Phi^{(\ell)}} = (\lambda_{+}^{(\ell)} + \lambda_{-}^{(\ell)}) \ket{\Phi^{(\ell)}} .
\end{equation}
Therefore
\begin{equation}
 \hat{N} \ket{\Phi} = \left(\sum_{\bvec{k}}^{} \hat{N}_{\bvec{k}} \right) \left( \sum_{\ell}^{} c_{\ell} \ket{\Phi^{(\ell)}}\right) = \sum_{\ell}^{} c_{\ell} \left( \sum_{\bvec{k} \in \mathcal{K}}^{} (\lambda_{+}^{(\ell)} + \lambda_{-}^{(\ell)}) \ket{\Phi^{(\ell)}} \right).
\end{equation}
Since $\ket{\Phi^{(\ell)}}$ are orthogonal and $\hat{N} \ket{\Phi} = N_{occ} N_{\bvec{k}} \ket{\Phi}$, it follows that $\lambda_{+}^{(\ell)} + \lambda_{-}^{(\ell)} = N_{occ}$ for all $\ell$ proving the result.
\end{proof}
The decomposition in \cref{prop:uniform-chiral-filling} is a key part of the characterization of the many-body ground states.
Since the states $\ket{\Phi^{(\ell)}}$ can be chosen to have different $\hat{n}_{\pm, \bvec{k}}$ eigenvalues, we can divide the analysis into studying the different eigenspaces of $\hat{n}_{\pm, \bvec{k}}$ independently.

The key observation which simplifies analyzing these eigenspaces is that both $\hat{C}_{\pm, \bvec{k}, \bvec{k}'}$ commute with a $\mathsf{U}(N_{occ})$
action within the $+$ chiral sector and a $\mathsf{U}(N_{occ})$ action within the $-$ chiral sector.
Explicitly, in the valleyful case, for any $U \in \mathsf{U}(N_{occ})$ if we rotate the creation operators by $U$, that is we define 
\begin{equation}
  \hat{g}_{(m, \tau), \bvec{k}}^{\dagger} = \sum_{m' \tau' = +}^{} \hat{f}_{(m', \tau'), \bvec{k}}^{\dagger} [ U ]_{(m', \tau'),(m, \tau)}
\end{equation}
then we have
\begin{equation}
  \begin{split}
    \hat{C}_{+, \bvec{k}, \bvec{k}'}
    & = \sum_{m \tau = +}^{} \hat{f}_{(m, \tau), \bvec{k}}^{\dagger} \hat{f}_{(m,\tau), \bvec{k}'} \\
    & = \sum_{m \tau = +}^{} \left( \sum_{m' \tau' = +}^{} \hat{g}_{(m', \tau'), \bvec{k}}^{\dagger} [\overline{U}]_{(m', \tau'), (m, \tau)}  \right) \left( \sum_{m'' \tau'' = +}^{} \hat{g}_{(m'', \tau''), \bvec{k}} [U]_{(m'', \tau''), (m, \tau)}  \right) \\
    & = \sum_{m \tau = +}^{} \hat{g}_{(m, \tau), \bvec{k}}^{\dagger} \hat{g}_{(m, \tau), \bvec{k}} 
  \end{split}
\end{equation}
where in the last line we have used that when $U$ is a unitary $\sum_{\alpha}^{} [U]_{\beta \alpha} [\overline{U}]_{\gamma \alpha} = \delta_{\beta \gamma}$.
Similar calculations show that in all cases the chiral transfer operators, $\hat{C}_{\pm, \bvec{k}, \bvec{k}'}$, commute with basis transformation within each chiral sector from $\mathsf{U}(N_{occ}) \times \mathsf{U}(N_{occ})$.

With this observation, one can easily check that \(\widehat{\rho}(\bvec{q}')\) also commutes with this \(\mathsf{U}(N_{occ}) \times \mathsf{U}(N_{occ})\) action and therefore the ground state manifold is an irreducible representation of the group $\mathsf{U}(N_{occ}) \times \mathsf{U}(N_{occ})$.
This $\mathsf{U}(N_{occ}) \times \mathsf{U}(N_{occ})$ freedom in the ground state is known as the ``hidden symmetry'' of chiral MATBG in the physics literature \cite{BultinckKhalafLiuEtAl2020}.
Once we restrict to a specific eigenspace $\lambda_{+}$, we can use the representation theory of $\mathsf{U}(N_{occ}) \times \mathsf{U}(N_{occ})$ to show that that all ground states lie in the irreducible representation generated by the Hartree-Fock ground states proving~\cref{thm:main-result}.

As a first step, we will derive a necessary and sufficient condition for a many-body state $\ket{\Phi}$ to be a ground state in~\cref{sec:equiv-characterization}.
This condition will almost immediately imply~\cref{prop:uniform-chiral-filling}.
Using this condition, we will then study the representation theory of the ground state manifold in~\cref{sec:repr-theory-ground}.

\section{An Equivalent Characterization of the Ground States}
\label{sec:equiv-characterization}
The main result of this section is to prove the following equivalent condition for a state $\ket{\Phi}$ to be a ground state:
\begin{proposition}
  \label{prop:equiv-characterization}
  For the spinless, valleyless FBI model of MATBG, a many-body state $\ket{\Phi}$ satisfies $\hat{H}_{FBI} \ket{\Phi} = 0$ if and only if $\ket{\Phi}$ is half filled and for all $\bvec{k}, \bvec{k}' \in \mathcal{K}$ 
  \begin{equation}
    \hat{C}_{\pm, \bvec{k}, \bvec{k}'} \ket{\Phi} = 0 \quad \text{whenever} \quad \bvec{k} - \bvec{k}' \not\in \Gamma^{*}. \label{eq:ferromagnetism}
  \end{equation}
  where $\hat{C}_{\pm, \bvec{k}, \bvec{k}'}$ is as defined in~\cref{eq:chiral-transfer-operators}.
  The same equivalence is true for valleyful MATBG and valleyful and spinful MATBG if we replace the $\hat{C}$ operators with $\hat{C}^{(v)}$ and $\hat{C}^{(v,s)}$ respectively.
\end{proposition}

To simplify the notation in our proof, we will henceforth suppress the superscripts on the $\hat{C}$ operators; the arguments we use to prove~\cref{prop:equiv-characterization} will be independent of which case we consider.
The requirement that every ground state $\ket{\Phi}$ must lie in the kernel of $\hat{C}_{\pm, \bvec{k}, \bvec{k}'}$  for all pairs of momenta implies that the ground state must be composed to states which uniformly fill each chiral sector:
\begin{lemma}
  \label{lem:uniform-filling}
  Suppose that $\ket{\Phi}$ is a many-body state which satisfies~\cref{eq:ferromagnetism}.
  Then for all $\bvec{k}, \bvec{k}' \in \mathcal{K}$
  \begin{equation}
    \Big( \hat{n}_{\pm, \bvec{k}} - \hat{n}_{\pm, \bvec{k}'}\Big) \ket{\Phi} = 0. \label{eq:uniform-filling} 
  \end{equation}
\end{lemma}
\begin{proof}
  The conclusion of the lemma is trivial if $\bvec{k} - \bvec{k}' \in \Gamma^{*}$ so we only need to consider $\bvec{k} - \bvec{k}' \not\in \Gamma^{*}$.
  By the CAR for all $\bvec{k}, \bvec{k}' \in \mathcal{K}$ 
  \begin{equation}
    [ \hat{C}_{\pm, \bvec{k}, \bvec{k}'}, \hat{C}_{\pm, \bvec{k}', \bvec{k}} ] = \hat{C}_{\pm, \bvec{k}, \bvec{k}} - \hat{C}_{\pm, \bvec{k}', \bvec{k}'} = \hat{n}_{\pm, \mathbf{k}} - \hat{n}_{\pm, \mathbf{k}'}.
  \end{equation}
  But for any $\bvec{k} - \bvec{k}' \not\in \Gamma^{*}$
  \begin{equation}
   [ \hat{C}_{\pm, \bvec{k}, \bvec{k}'}, \hat{C}_{\pm, \bvec{k}', \bvec{k}} ] \ket{\Phi} = \bigg(\hat{C}_{\pm, \bvec{k}, \bvec{k}'} \hat{C}_{\pm, \bvec{k}', \bvec{k}}  - \hat{C}_{\pm, \bvec{k}', \bvec{k}} \hat{C}_{\pm, \bvec{k}, \bvec{k}'} \bigg) \ket{\Phi} = 0
 \end{equation}
 where the last equality is due to~\cref{eq:ferromagnetism} which proves the lemma.
\end{proof}
This lemma implies~\cref{prop:uniform-chiral-filling} since the set of operators $\{ \hat{n}_{\chi, \bvec{k}} : \chi \in \{ +, - \},~ \bvec{k} \in \mathcal{K} \}$ are self-adjoint and commuting. So any many-body state may be orthogonally decomposed into a sum of joint eigenvectors.

We begin by showing that~\cref{eq:ferromagnetism} is sufficient for $\ket{\Phi}$ to be a ground state in~\cref{sec:proof-sufficiency}.
To prove~\cref{eq:ferromagnetism} is also necessary for $\ket{\Phi}$ to be a ground state, we split the analysis into two steps: considering $\widehat{\rho}(\bvec{q}')$ where $\bvec{q}' \in \Gamma^{*}$, which we refer to as the ``on lattice terms'', and $\bvec{q}' \not\in \Gamma^{*}$, which we refer to as the ``off lattice terms''.
Due to frustration-freeness, every ground state must satisfy $\widehat{\rho}(\mathbf{q}') \ket{\Phi} = 0$.
Using the on lattice terms, we will show that every ground state must be a joint eigenvector of a family of self-adjoint commuting operators with a fixed eigenvalue.
Using this joint eigenvector property combined with the frustration-freeness of the off lattice terms, will then prove~\cref{prop:equiv-characterization}.
We consider the on lattice terms in~\cref{sec:on-lattice-terms} and the off lattice terms in~\cref{sec:off-lattice-terms}.

\subsection{Proof that~\Cref{eq:ferromagnetism} is Sufficient}
\label{sec:proof-sufficiency}
Suppose that $\ket{\Phi}$ is a many-body state which is half-filled and satisfies~\cref{eq:uniform-filling,eq:ferromagnetism}.
When $\bvec{q}' \not\in \Gamma^{*}$, the operator $\widehat{\rho}(\bvec{q}')$ applied to $\ket{\Phi}$ evaluates to
\begin{equation}
  \widehat{\rho}(\bvec{q}') \ket{\Phi} = \left( \sum_{\bvec{k} \in \mathcal{K}}^{} a_{\bvec{k}}(\bvec{q}') \hat{C}_{+, \bvec{k}, (\bvec{k} + \bvec{q}')}  + \overline{a_{\bvec{k}}(\bvec{q}')}  \hat{C}_{-, \bvec{k}, (\bvec{k} + \bvec{q}')} \right) \ket{\Phi} = 0
\end{equation}
where the final equality is because $\bvec{k} - (\bvec{k} + \bvec{q}') = -\bvec{q}' \not\in \Gamma^{*}$ and~\cref{eq:ferromagnetism}.
Therefore, to prove $\ket{\Phi}$ is a ground state we only need to show that $\widehat{\rho}(\bvec{G}) \ket{\Phi} = 0$ for all $\bvec{G} \in \Gamma^{*}$.

As noted above, the family of operators $\{ \hat{n}_{c, \bvec{k}} : c \in \{ +, - \}, \bvec{k} \in \mathcal{K} \}$ are all self-adjoint and mutually commuting, we may decompose the state $\ket{\Phi}$ into a sum of joint eigenvectors of this family as follows:
\begin{equation}
  \begin{split}
    \ket{\Phi} & = \sum_{\ell}^{} c_{\ell} \ket{\Psi^{(\ell)}}, \quad \text{where} \\
               & \hat{n}_{\pm, \bvec{k}} \ket{\Psi^{(\ell)}} = \lambda_{\pm}^{(\ell)} \ket{\Psi^{(\ell)}}.
  \end{split}
\end{equation}
Following the reasoning in the proof of~\cref{coro:uniform-filling}, for all $\ell$, $\lambda_{+}^{(\ell)} +  \lambda_{-}^{(\ell)} = N_{occ}$.

For each $\mathbf{q}' \in \Gamma^{*}$, we may apply $\widehat{\rho}(\mathbf{q}')$ to each $\ket{\Psi^{(\ell)}}$ and we have that
\begin{equation}
  \begin{split}
    \widehat{\rho}(\bvec{q}') \ket{\Psi^{(\ell)}}
    & = \left[ \sum_{\bvec{k} \in \mathcal{K}}^{} a_{\bvec{k}}(\bvec{q}') \left( \hat{n}_{+, \bvec{k}} - \frac{1}{2} N_{occ} \right) + \overline{a_{\bvec{k}}(\bvec{q}')} \left( \hat{n}_{-, \bvec{k}} - \frac{1}{2} N_{occ} \right) \right] \ket{\Psi^{(\ell)}} \\[1ex]
    & = \left[ \sum_{\bvec{k} \in \mathcal{K}}^{} a_{\bvec{k}}(\bvec{q}') \left( \lambda_{+}^{(\ell)} - \frac{1}{2} N_{occ} \right) + \overline{a_{\bvec{k}}(\bvec{q}')} \left( \lambda_{-}^{(\ell)} - \frac{1}{2} N_{occ} \right) \right] \ket{\Psi^{(\ell)}} \\
    & = \left[ \sum_{\bvec{k} \in \mathcal{K}}^{} \Re{(a_{\bvec{k}}(\bvec{q}'))} \left( \lambda_{+}^{(\ell)} + \lambda_{-}^{(\ell)} - N_{occ} \right) + \Im{(a_{\bvec{k}}(\bvec{q}'))} \left( \lambda_{+}^{(\ell)} - \lambda_{-}^{(\ell)} \right) \right] \ket{\Psi^{(\ell)}}.
  \end{split}
\end{equation}
The first term vanishes since $\lambda_{+}^{(\ell)} + \lambda_{-}^{(\ell)} = N_{occ}$ and so we conclude that
\begin{equation}
    \widehat{\rho}(\bvec{q}') \ket{\Psi^{(\ell)}} = \left( \lambda_{+}^{(\ell)} - \lambda_{-}^{(\ell)} \right) \left( \sum_{\mathbf{k} \in \mathcal{K}}^{} \Im{(a_{\bvec{k}}(\bvec{q}'))} \right)  \ket{\Psi^{(\ell)}}.
\end{equation}
To show that the above expression vanishes, we recall that due to the choice of gauge $a_{-\bvec{k}}(\bvec{G}) = \overline{a_{\bvec{k}}(\bvec{G})}$ for all $\mathbf{G} \in \Gamma^{*}$ (\cref{eq:form-factor-particle-hole}).
Since the grid $\mathcal{K}$ (\cref{eq:mp-grid}) is closed under inversion, $\sum_{\bvec{k} \in \mathcal{K}}^{} \Im{(a_{\bvec{k}}(\mathbf{q}'))} = 0$ for all $\mathbf{q}' \in \Gamma^{*}$.
Hence, for all $\bvec{q}' \in \Gamma^{*}$, $\widehat{\rho}(\bvec{q}') \ket{\Phi} = 0$ and so $\ket{\Phi}$ is a ground state.

From the above calculation, we have shown that any half-filled many-body state $\ket{\Phi}$ which satisfies \cref{eq:uniform-filling} is a ground state.
In the next section, we begin our proof that these conditions are also necessary.

\subsection{On Lattice Terms, $\bvec{q}' \in \Gamma^*$}
\label{sec:on-lattice-terms}
For $\bvec{q}' = \bvec{G} \in \Gamma^{*}$, we can use a property of the gauge choice (\cref{eq:form-factor-particle-hole}) to simplify~\cref{eq:rho-q-tbg}.
In particular, we have
\begin{equation}
  \begin{split}
    \widehat{\rho}(\bvec{G})
    & = \sum_{\bvec{k} \in \mathcal{K}}^{} a_{\bvec{k}}(\bvec{G}) \left( \hat{n}_{+, \bvec{k}} - \frac{1}{2} N_{occ} \right) + \overline{a_{\bvec{k}}(\bvec{G})} \left( \hat{n}_{-, \bvec{k}} - \frac{1}{2} N_{occ} \right) \\
    & = \sum_{\bvec{k} \in \mathcal{K}}^{} a_{\bvec{k}}(\bvec{G}) \left( \hat{n}_{+, \bvec{k}} - \frac{1}{2} N_{occ} \right) + a_{-\bvec{k}}(\bvec{G}) \left( \hat{n}_{-, \bvec{k}} - \frac{1}{2} N_{occ} \right) \\
    & = \sum_{\bvec{k} \in \mathcal{K}}^{} a_{\bvec{k}}(\bvec{G}) \left( \hat{n}_{+, \bvec{k}} + \hat{n}_{-, (-\bvec{k})} - N_{occ} \right).
  \end{split}
\end{equation}
where to get the third line, we have performed the change of variables $\bvec{k} \mapsto - \bvec{k}$.

For a state $\ket{\Phi}$ to be annihilated by $\widehat{\rho}(\bvec{G})$, the following condition is sufficient:
\begin{equation}
  \label{eq:gs-prop-1}
  \left( \hat{n}_{+, \bvec{k}} + \hat{n}_{-, (-\bvec{k})} - N_{occ} \right) \ket{\Phi} = 0 \qquad \forall \bvec{k} \in \mathcal{K}.
\end{equation}
The main goal in this section is to show that~\cref{eq:gs-prop-1} is in fact \textit{necessary} to be a ground state.
\begin{proposition}
  \label{prop:gs-prop-1}
  A many-body state satisfies $\widehat{\rho}(\bvec{G}) \ket{\Phi} = 0$ for all $\bvec{G} \in \Gamma^{*}$ if and only if it is a joint eigenvector of the family of operators $\{ \hat{n}_{+, \bvec{k}} + \hat{n}_{-, (-\bvec{k})} : \bvec{k} \in \mathcal{K} \}$ with eigenvalue $N_{occ}$.
\end{proposition}
Before proving~\cref{prop:gs-prop-1}, we remark that the operators $\hat{n}_{+, \bvec{k}} + \hat{n}_{-, (-\bvec{k}), (-\bvec{k})}$ will appear many times during our proof so we define the shorthand
\begin{equation}
  \label{eq:ccal-def}
  \hat{\mathcal{N}}_{\bvec{k}} := \hat{n}_{+, \bvec{k}} + \hat{n}_{-, (-\bvec{k})}.
\end{equation}
\begin{proof}[Proof of~\cref{prop:gs-prop-1}]
Any $\ket{\Phi}$ satisfying~\cref{eq:gs-prop-1} implies $\widehat{\rho}(\bvec{G}) \ket{\Phi} = 0$ for all $\bvec{G} \in \Gamma^{*}$. So we only need to prove the converse.
  If $\ket{\Phi}$ is a ground state then for all $\bvec{G} \in \Gamma^{*}$:
  \begin{equation}
    \label{eq:rho-g-zero}
    \widehat{\rho}(\bvec{G}) \ket{\Phi} = \sum_{\bvec{k} \in \mathcal{K}}^{} a_{\bvec{k}}(\bvec{G}) \left(  \hat{\mathcal{N}}_{\bvec{k}} - N_{occ} \right) \ket{\Phi} = 0.
  \end{equation}
  The key observation is the following: Suppose we fix some complete basis for all many-body states, $\{ \ket{\Psi^{(\nu)}} \}_{\nu}$, where $\nu$ runs over an arbitrary indexing set.
  If we take the inner product with $\ket{\Psi^{(\nu)}}$ on both sides of~\cref{eq:rho-g-zero} then it must be that for all $\nu$:
  \begin{equation}
    \label{eq:rho-g-nu-zero}
    \sum_{\bvec{k} \in \mathcal{K}}^{} a_{\bvec{k}}(\bvec{G}) \braket{\Psi^{(\nu)}, \left( \hat{\mathcal{N}}_{\bvec{k}} - N_{occ} \right) \Phi} = 0.
  \end{equation}
  However, the sum over $\bvec{k} \in \mathcal{K}$ can now be viewed as an inner product between two vectors whose entries are indexed by $\bvec{k}$.
  Since~\cref{eq:rho-g-nu-zero} must be true for all $\bvec{G} \in \Gamma^{*}$, if we can show the infinite matrix
  \begin{equation}
    \label{eq:akq-matrix-form}
    \begin{bmatrix}
      a_{\bvec{k}_{1}}(\bvec{G}_{1}) & a_{\bvec{k}_{2}}(\bvec{G}_{1}) & \cdots & a_{\bvec{k}_{N_{\bvec{k}}}}(\bvec{G}_{1}) \\
      a_{\bvec{k}_{1}}(\bvec{G}_{2}) & a_{\bvec{k}_{2}}(\bvec{G}_{2}) & \cdots & a_{\bvec{k}_{N_{\bvec{k}}}}(\bvec{G}_{2}) \\
      \vdots & \vdots & \ddots & \vdots
    \end{bmatrix}
  \end{equation}
  is full rank, then we can conclude that for all $\nu$
  \begin{equation}
    \braket{\Psi^{(\nu)}, \left(  \hat{\mathcal{N}}_{\bvec{k}} - N_{occ} \right) \Phi} = 0 \qquad \forall \bvec{k} \in \mathcal{K}
  \end{equation}
  which implies~\cref{prop:gs-prop-1} since $\{\ket{\Psi^{(\nu)}}\}_{\nu}$ forms a complete basis set.
  We defer the proof of the following technical lemma, which implies~\cref{prop:gs-prop-1}, to~\cref{sec:lemma-linear-indep}:
  \begin{lemma}
    \label{lem:linear-indep}
    Let $a_{\bvec{k}}(\bvec{G})$ be the entries of the form factor as in~\cref{prop:form-factor-gauge-fixing}.
    For any finite set $\mathcal{K} \subseteq \Omega^{*}$ and any $\{ c_{\bvec{k}} \}_{\bvec{k} \in \mathcal{K}} \subseteq \CC$
    \begin{equation}
      \sum_{\bvec{k} \in \mathcal{K}}^{} a_{\bvec{k}}(\bvec{G}) c_{\bvec{k}} = 0\quad \forall \bvec{G} \in \Gamma^{*} \quad \text{if and only if} \quad c_{\bvec{k}} = 0 \quad \forall \bvec{k} \in \mathcal{K}.
    \end{equation}
  \end{lemma}
\end{proof}

\subsection{Off Lattice Terms, $\bvec{q}' \not\in \Gamma^*$}
\label{sec:off-lattice-terms}
The main goal in this section is to prove the following proposition which combined with~\cref{prop:total-number} immediately implies the equivalent characterization (\cref{prop:equiv-characterization}:
\begin{proposition}
  \label{prop:gs-prop-2}
  Let $\ket{\Phi}$ be any state which is a joint eigenstate of $\{ \hat{\mathcal{N}}_{\bvec{k}} \}_{\bvec{k} \in \mathcal{K}}$ and satisfies $\widehat{\rho}(\bvec{q}') \ket{\Phi} = 0$ for all $\bvec{q}' \in (\mathcal{K} + \Gamma^{*}) \setminus \Gamma^{*}$.
  For all $\bvec{k}, \bvec{k}' \in \mathcal{K}$ with $\bvec{k} - \bvec{k}' \not\in \Gamma^{*}$, $\ket{\Phi}$ also satisfies
  \begin{equation}
    \label{eq:gs-prop-2}
    \hat{C}_{\pm, \bvec{k}, \bvec{k}'} \ket{\Phi} = 0.
  \end{equation}
\end{proposition}

Towards proving~\cref{prop:gs-prop-2}, let us recall that in the previous section we have shown that every ground state $\ket{\Phi}$ must be a joint eigenvector of the family $\{ \hat{\mathcal{N}}_{\bvec{k}} \}_{\bvec{k} \in \mathcal{K}}$ which are all self-adjoint and mutually commuting.
To make use of this property, let us recall the following basic fact from linear algebra:
\begin{lemma}
  \label{lem:joint-evector}
  Suppose that $v$ is a joint eigenvector
of two matrices $A$ and $B$, then \([A, B] v = 0\).
\end{lemma}
As a consequence of this lemma, if \(\ket{\Phi}\) is a ground state (and so \(\widehat{\rho}(\bvec{q}') \ket{\Phi} = 0\)) then it must be that for all $\bvec{k} \in \mathcal{K}$
\begin{equation}
  \Big[ \hat{\mathcal{N}}_{\bvec{k}},~ \widehat{\rho}(\bvec{q}') \Big] \ket{\Phi} = 0.
\end{equation}
Since $\hat{\mathcal{N}}_{\bvec{k}}$ only involves two momenta $\bvec{k}$ and $-\bvec{k}$, for any $\bvec{q}' \in (\mathcal{K} + \Gamma^{*}) \setminus \Gamma^{*}$ one can check that the commutator \([\hat{\mathcal{N}}_{\bvec{k}},\widehat{\rho}(\bvec{q}')]\) only involves \(4\) terms.
We can further reduce this to \(2\) terms by considering the double commutator \([ \hat{\mathcal{N}}_{\bvec{k} + \bvec{q}'}, [\hat{\mathcal{N}}_{\bvec{k}},\widehat{\rho}(\bvec{q}')]]\).
In particular, a straightforward calculation (see \cref{sec:proof-gs-prop-2a}) shows that
\begin{equation}
  \label{eq:double-commutator}
    \Bigg[ \hat{\mathcal{N}}_{\bvec{k} + \bvec{q}'},~ \Big[ \hat{\mathcal{N}}_{\bvec{k}},~ \widehat{\rho}(\bvec{q}') \Big]  \Bigg] = a_{\bvec{k}}(\bvec{q}') \hat{C}_{+, \bvec{k}, (\bvec{k} + \bvec{q}')} - \overline{a_{-\bvec{k} - \bvec{q}'}(\bvec{q}')} \hat{C}_{-, (-\bvec{k} - \bvec{q}'), (-\bvec{k})}
\end{equation}
and due to~\cref{lem:joint-evector} every ground state must be annihilated by the above operator.

It turns out we can further simplify the above condition, by noticing that the outcome of applying these two chiral transfer operators (\(\hat{C}_{+, \bvec{k}, (\bvec{k} + \bvec{q}')}\) and \(\hat{C}_{-, (-\bvec{k} - \bvec{q}'), (-\bvec{k})}\)) must be orthogonal and so any ground state must be annihilated by these two terms separately.
\begin{lemma}
  \label{lem:double-commutator}
  If $\ket{\Phi}$ is a many-body ground state which is a joint eigenvector of $\{ \hat{N}_{\bvec{k}} \}_{\bvec{k} \in \mathcal{K}}$ and $\widehat{\rho}(\bvec{q}') \ket{\Phi} = 0$ for all $\bvec{q}' \in (\mathcal{K} + \Gamma^{*}) \setminus \Gamma^{*}$, then for all $\bvec{k} \in \mathcal{K}$ and all $\bvec{q}' \in (\mathcal{K} + \Gamma^{*}) \setminus \Gamma^{*}$
  \begin{equation}
    \label{eq:lem-double-commutator}
    \begin{split}
      a_{\bvec{k}}(\bvec{q}') \hat{C}_{+, \bvec{k}, (\bvec{k} + \bvec{q}')} \ket{\Phi} & = 0 \\[1ex]
      \overline{a_{-\bvec{k} - \bvec{q}'}(\bvec{q}')} \hat{C}_{-, (-\bvec{k} - \bvec{q}'), (-\bvec{k})} \ket{\Phi} & = 0
    \end{split}
  \end{equation}
\end{lemma}
Note that \cref{lem:double-commutator} does not yet imply~\cref{prop:equiv-characterization} because \(a_{\bvec{k}}(\bvec{q}')\) or \(a_{-\bvec{k} - \bvec{q}'}(\bvec{q}')\) can be zero (and hence~\cref{eq:lem-double-commutator} may be vacuously true).
\begin{proof}
  Using the CAR we can verify that
  \begin{equation}
    \begin{split}
      [ \hat{\mathcal{N}}_{\bvec{k}},\, \hat{C}_{+, \bvec{k}, (\bvec{k} + \bvec{q}')} ] & = \hat{C}_{+, \bvec{k}, (\bvec{k} + \bvec{q}')} \\
      [ \hat{\mathcal{N}}_{\bvec{k}},\, \hat{C}_{-, (-\bvec{k} - \bvec{q}'), (-\bvec{k})} ] & = - \hat{C}_{-, (-\bvec{k} - \bvec{q}'), (-\bvec{k})}. 
    \end{split}
  \end{equation}
  Since $\ket{\Phi}$ is an eigenvector of $\hat{\mathcal{N}}_{\bvec{k}}$ with eigenvalue $N_{occ}$, we have that
  \begin{equation}
    \begin{split}
      \hat{\mathcal{N}}_{\bvec{k}} \hat{C}_{+, \bvec{k}, (\bvec{k} + \bvec{q}')} \ket{\Phi}
      & = \hat{C}_{+, \bvec{k}, (\bvec{k} + \bvec{q}')} \hat{\mathcal{N}}_{\bvec{k}} \ket{\Phi} + \Big[ \hat{\mathcal{N}}_{\bvec{k}}, \hat{C}_{+, \bvec{k}, (\bvec{k} + \bvec{q}')} \Big] \ket{\Phi} \\
      & = (N_{occ} + 1) \hat{C}_{+, \bvec{k}, (\bvec{k} + \bvec{q}')}  \ket{\Phi}.
    \end{split}
  \end{equation}
  So  $\hat{C}_{+, \bvec{k}, (\bvec{k} + \bvec{q}')}  \ket{\Phi}$ is an eigenvector of $\hat{\mathcal{N}}_{\bvec{k}}$ with eigenvalue $(N_{occ} + 1)$.
  
  A similar calculation shows that
  \begin{equation}
    \hat{\mathcal{N}}_{\bvec{k}} \hat{C}_{-, (-\bvec{k} - \bvec{q}'), (-\bvec{k})} \ket{\Phi} = (N_{occ} - 1) \hat{C}_{-, (-\bvec{k} - \bvec{q}'), (-\bvec{k})} \ket{\Phi}.
  \end{equation}
  But $\hat{\mathcal{N}}_{\bvec{k}}$ is a self-adjoint operator so it must be that $\hat{C}_{+, \bvec{k}, (\bvec{k} + \bvec{q}')}  \ket{\Phi}$ and $\hat{C}_{-, (-\bvec{k} - \bvec{q}'), (-\bvec{k})} \ket{\Phi}$ are orthogonal.
  Since these two states are orthogonal, \cref{eq:double-commutator} implies
  \begin{equation}
    \label{eq:double-commutator-2}
    \begin{split}
     & a_{\bvec{k}}(\bvec{q}') \hat{C}_{+, \bvec{k}, (\bvec{k} + \bvec{q}')} \ket{\Phi} = 0, \\[1ex]
      & \overline{a_{-\bvec{k} - \bvec{q}'}(\bvec{q}')} \hat{C}_{-, (-\bvec{k} - \bvec{q}'), (-\bvec{k})} \ket{\Phi} = 0
    \end{split}
  \end{equation}
  which proves the statement.
\end{proof}
While generally the value of \(a_{\bvec{k}}(\bvec{q}')\) can be difficult to determine, for any \(\bvec{k} \in \mathcal{K}\), we have that \(a_{\bvec{k}}(\bvec{0}) = 1\) since by definition
\begin{equation}
  a_{\bvec{k}}(\bvec{0}) = \int_{\Omega}^{} \| u_{1,\bvec{k}}(\bvec{r}) \|^{2} \dd{\bvec{r}} = 1.
\end{equation}
Using the Lipschitz property of the map $\bvec{q}' \to |a_{\bvec{k}}(\bvec{q}')|^{2}$, we have the following proposition
\begin{proposition}[{\cite[Lemma 4.1]{BeckerLinStubbs2023}}]
  \label{prop:non-degeneracy}
  There exists a constant $q_{c}>0$ so that $|a_{\bvec{k}}(\bvec{q}')| > 0$ for all $\bvec{k} \in \Omega^{*}$, $| \bvec{q}' | < q_{c}$.
\end{proposition}
To make use of this fact, we make the following assumption on the grid $\mathcal{K}$.
\begin{assumption}
  \label{assume:grid-spacing}
  We assume that $n_{\bvec{k}_{x}}$ and $n_{\bvec{k}_{y}}$ are chosen so that the spacing between two neighboring points in $\mathcal{K}$ are separated by a distance less than $q_{c}$.
\end{assumption}
We now have all the tools needed to prove~\cref{prop:gs-prop-2}.
\begin{proof}[Proof of~\cref{prop:gs-prop-2}]
  To prove this proposition, we make use of the following identity
  \begin{equation}
    \Big[ \hat{C}_{\pm, \bvec{k}, (\bvec{k} + \bvec{q}_{1})},~ \hat{C}_{\pm, (\bvec{k} + \bvec{q}_{1}), (\bvec{k} + \bvec{q}_{1} + \bvec{q}_{2})} \Big] = \hat{C}_{\pm, \bvec{k}, (\bvec{k} + \bvec{q}_{1} + \bvec{q}_{2})}-\delta_{\bvec{q}_{1} + \bvec{q}_{2} \in \Gamma^{*}} \hat{C}_{\pm, (\bvec{k} + \bvec{q}_{1}), (\bvec{k} + \bvec{q}_{1})},
  \end{equation}
  Therefore, for any $\bvec{q}_{1} + \bvec{q}_{2} \not\in \Gamma^{*}$ we have the following implication:
  \begin{equation}
    \label{eq:ferromagnetic-implication}
    \begin{array}{r}
      \hat{C}_{\pm, \bvec{k}, (\bvec{k} + \bvec{q}_{1})} \ket{\Phi} = 0 \\[1ex]
      \hat{C}_{\pm, (\bvec{k} + \bvec{q}_{1}), (\bvec{k} + \bvec{q}_{1} + \bvec{q}_{2})} \ket{\Phi} = 0
    \end{array}
    \quad
    \Longrightarrow
    \quad
    \hat{C}_{\pm, \bvec{k}, (\bvec{k} + \bvec{q}_{1} + \bvec{q}_{2})} \ket{\Phi} = 0.
  \end{equation}
  As for the proof, we begin by fixing an arbitrary $\bvec{k} \in \mathcal{K}$ and let $\bvec{e}_{1} = \frac{\bvec{b}_{1}}{n_{k_{x}}}$, $\bvec{e}_{2} = \frac{\bvec{b}_{2}}{n_{k_{y}}}$ where we recall $\bvec{b}_{1}$, $\bvec{b}_{2}$ are generating vectors for $\Gamma^{*}$.
  We will show that for any $\bvec{k}' \in \mathcal{K}$ where $\bvec{k}' \neq \bvec{k}$ that $\hat{C}_{\pm, \bvec{k}, \bvec{k}'} \ket{\Phi} = 0$.
  
  Because two neighboring points in \(\mathcal{K}\) differ by \(\pm \bvec{e}_{1}\) or \(\pm \bvec{e}_{2}\) (modulo the lattice)~\cref{assume:grid-spacing} and~\cref{eq:ferromagnetic-implication} gives us the implication
  \begin{equation}
    \begin{split}
      \hat{C}_{\pm, \bvec{k}, (\bvec{k} + \bvec{e}_{1})} \ket{\Phi} & = 0 \\
      \hat{C}_{\pm, (\bvec{k} + \bvec{e}_{1}), (\bvec{k} + 2 \bvec{e}_{1})} \ket{\Phi} & = 0
    \end{split}
    \quad
    \Longrightarrow
    \quad
    \hat{C}_{\pm, \bvec{k}, (\bvec{k} + 2 \bvec{e}_{1})} \ket{\Phi} = 0.
  \end{equation}
  so long as \(2 \bvec{e}_{1} \not\in \Gamma^{*}\).
  
  Once we know that every ground state must satisfy \(\hat{C}_{\pm, \bvec{k}, (\bvec{k} + 2 \bvec{e}_{1})} \ket{\Phi} = 0\), we also have the implication
  \begin{equation}
    \begin{split}
      \hat{C}_{\pm, \bvec{k}, (\bvec{k} + 2 \bvec{e}_{1})} \ket{\Phi} & = 0 \\
      \hat{C}_{\pm, (\bvec{k} + 2 \bvec{e}_{1}), (\bvec{k} + 2 \bvec{e}_{1} + \bvec{e}_{2})} \ket{\Phi} & = 0
    \end{split}
    \quad
    \Longrightarrow
    \quad
    \hat{C}_{\pm, \bvec{k}, (\bvec{k} + 2 \bvec{e}_{1} + \bvec{e}_{2})} \ket{\Phi} = 0.
  \end{equation}
  so long as \(2 \bvec{e}_{1} + \bvec{e}_{2} \not\in \Gamma^{*}\).
  Since any $\bvec{k}' \in \mathcal{K}$ may be written as $\bvec{k}' = m \bvec{e}_{1} + n \bvec{e}_{2}$ for some $m, n \in \mathbb{Z}$, we can repeat this argument to conclude that \(\hat{C}_{\pm, \bvec{k}, \bvec{k}'} \ket{\Phi} = 0\) so long as $\bvec{k} - \bvec{k}' \not\in \Gamma^{*}$. 
\end{proof}

\section{Ground State Manifold of the Spinless, Valleyless Model and Occupation Vector Representations}
\label{sec:repr-theory-ground}
From the arguments in the previous section we have proven~\cref{prop:equiv-characterization} which fully characterizes all ground states of MATBG.
It is worth observing that this characterization only depends on the $\hat{C}$ operators and hence is \textit{independent} of the underlying material details.
As outlined in~\cref{sec:nature-many-body}, \cref{lem:uniform-filling} implies that we can decompose any many-body ground state
\begin{equation}
  \label{eq:particle-number-decomp}
  \begin{split}
    \ket{\Phi} & = \sum_{\lambda_{+} = 0}^{N_{occ}}  c_{\lambda_{+}}  \ket{\Phi_{\lambda_{+}}},~\text{where} \\
                 \hat{n}_{+, \bvec{k}} & \ket{\Phi_{\lambda_{+}}} = \lambda_{+}\ket{\Phi_{\lambda_{+}}} \\
    \hat{n}_{-, \bvec{k}} & \ket{\Phi_{\lambda_{+}}} = (N_{occ} - \lambda_{+}) \ket{\Phi_{\lambda_{+}}}
  \end{split}
\end{equation}
Here, the sum over $\lambda_{+}$ runs over $\{ 0, \cdots, N_{occ} \}$ since these are all eigenvalues of $\hat{n}_{+, \bvec{k}}$.
For the single valley case,~\cref{eq:particle-number-decomp} is enough to completely fix the ground state
\begin{theorem}
  \label{thm:single-valley}
  For the  spinless, valleyless model of MATBG, a many-body state $\ket{\Phi} $ is a ground state if and only if it can be written as
  \begin{equation}
    \ket{\Phi} = \alpha \ket{\tilde{\Phi}_{1}} + \beta \ket{\tilde{\Phi}_{0}} \qquad \alpha, \beta \in \CC
  \end{equation}
  where
  \begin{equation}
    \label{eq:single-valley-fsd-def}
    \ket{\tilde{\Phi}_{1}} = \left( \prod_{\bvec{k} \in \mathcal{K}}^{} \hat{f}_{+, \bvec{k}}^{\dagger} \right) \ket{\mathrm{vac}} \qquad 
    \ket{\tilde{\Phi}_{0}} = \left( \prod_{\bvec{k} \in \mathcal{K}}^{} \hat{f}_{-, \bvec{k}}^{\dagger} \right) \ket{\mathrm{vac}}
  \end{equation}
\end{theorem}
\begin{proof}
  For the single valley model, $N_{occ} = 1$ so the sum in~\cref{eq:particle-number-decomp} only has two terms, $\ket{\Phi_{1}}$ and $\ket{\Phi_{0}}$, which must satisfy
  \begin{equation}
    \left\{
    \begin{array}{rl}
      \hat{n}_{+, \bvec{k}} \ket{\Phi_{1}} & = \ket{\Phi_{1}} \\
      \hat{n}_{-, \bvec{k}} \ket{\Phi_{1}} & = 0
    \end{array}
    \right.
    \qquad
    \left\{
    \begin{array}{rl}
      \hat{n}_{+, \bvec{k}} \ket{\Phi_{0}} & = 0 \\
      \hat{n}_{-, \bvec{k}} \ket{\Phi_{0}} & = \ket{\Phi_{0}}
    \end{array}
    \right.
  \end{equation}
  The only state which is a $+1$ eigenstate of $\hat{n}_{+, \bvec{k}}$ for all $\bvec{k} \in \mathcal{K}$ is $\left( \prod_{\bvec{k} \in \mathcal{K}}^{} \hat{f}_{+, \bvec{k}}^{\dagger} \right) \ket{\mathrm{vac}}$ and the only state $+1$ eigenstate of $\hat{n}_{-, \bvec{k}}$ for all $\bvec{k} \in \mathcal{K}$ is $\left( \prod_{\bvec{k} \in \mathcal{K}}^{} \hat{f}_{-, \bvec{k}}^{\dagger} \right) \ket{\mathrm{vac}}$ which proves the claim.
\end{proof}
The argument used in the proof of~\cref{thm:single-valley} generally fixes the form of $\ket{\Phi_{\lambda_{+}}}$ when $\lambda_{+} = 0$ and $\lambda_{+} = N_{occ}$ since there is only one state which has $\lambda_{+} = 0$ for all $\bvec{k}$ and $\lambda_{+} = N_{occ}$ for all $\bvec{k}$.
Since in the single valley case, the only possible values of $\lambda_{+}$ are $0$ or $N_{occ}$ so this completely fixes the ground state.

To handle the other cases (i.e. $\lambda_{+} \not\in \{ 0, N_{occ} \}$), we will use the fact that the chiral transfer operators commute with a $\mathsf{U}(N_{occ}) \times \mathsf{U}(N_{occ})$ basis transformation as noted in \cref{sec:nature-many-body}.
To make a formal argument, we introduce a bijection between the states $\ket{\Phi_{\lambda_{+}}}$ from the decomposition in~\cref{eq:particle-number-decomp} into a subspace of $\bigwedge^{N_{occ} N_{\bvec{k}}} \CC^{2 N_{occ} N_{\bvec{k}}}$; 
we refer to this mapping as the \textit{occupation vector} representation of a many-body state.

\subsection{Occupation Vector Representation of a Many-Body State}
\label{sec:occup-vect-repr}
Throughout this work, we have described both the many-body operators and many-body states using the language of second quantization (i.e. in terms of the creation and annihilation operators).
Second quantization is an attractive language for studying many particle systems, because it expresses many-body operators in a way which is agnostic to the number of particles in the system.
Because of the decomposition in~\cref{eq:particle-number-decomp} however, once we know $\lambda_{+}$, we know precisely how many electrons are in each chiral sector at each $\bvec{k} \in \mathcal{K}$ making second quantization unnecessary.

Let us begin by considering the valleyful model with $\lambda_{+} = 1$ and $N_{\bvec{k}} = 1$.
Every state from this space may be written as a linear combination of states of the form
\begin{equation}
  \label{eq:lambda-1-nk-1}
  \left( \alpha_{1} \hat{f}_{(+,+), \bvec{k}_{1}}^{\dagger} + \alpha_{2} \hat{f}_{(-,-), \bvec{k}_{1}}^{\dagger} \right) \left( \beta_{1} \hat{f}_{(+,-), \bvec{k}_{1}}^{\dagger} + \beta_{2} \hat{f}_{(-,+), \bvec{k}_{1}}^{\dagger} \right) \ket{\mathrm{vac}}.
\end{equation}
We may identify this state with the vector
\begin{equation}
  \begin{bmatrix}
    \alpha_{1} \\
    \alpha_{2}
  \end{bmatrix}
  \otimes
  \begin{bmatrix}
    \beta_{1} \\
    \beta_{2}
  \end{bmatrix} \in \CC^{2} \otimes \CC^{2}.
\end{equation}
Importantly, under this identification, for any $(U_{1}, U_{2}) \in \mathsf{U}(2) \times \mathsf{U}(2)$  performing a change of basis for the creation operators in~\cref{eq:lambda-1-nk-1} results in the mapping:
\begin{equation}
  \begin{bmatrix}
    \alpha_{1} \\
    \alpha_{2}
  \end{bmatrix}
  \otimes
  \begin{bmatrix}
    \beta_{1} \\
    \beta_{2}
  \end{bmatrix}
  \mapsto
  \left(
  U_{1}
  \begin{bmatrix}
    \alpha_{1} \\
    \alpha_{2}
  \end{bmatrix}
  \right)
  \otimes
  \left(
  U_{2}
  \begin{bmatrix}
    \beta_{1} \\
    \beta_{2}
  \end{bmatrix}
  \right),
\end{equation}
so the image of this map is a faithful representation $\mathsf{U}(2) \times \mathsf{U}(2)$.
We refer to $\begin{bmatrix} \alpha_{1} & \alpha_{2}\end{bmatrix}^{\top} \otimes \begin{bmatrix} \beta_{1} & \beta_{2} \end{bmatrix}^{\top}$ as the \textit{occupation vector} corresponding to~\cref{eq:lambda-1-nk-1} since $\begin{bmatrix} \alpha_{1} & \alpha_{2}\end{bmatrix}^{\top}$ describes the occupation in the $+$ chiral sector at momentum $\bvec{k}_{1}$ (and similarly for $-$ chiral sector).

For $N_{\bvec{k}} = 2$ and $\lambda_{+} = 1$ in the valleyful model, any state may be written as a linear combination of states of the form
\begin{equation}
  \begin{split}
    & \left( \alpha_{1} \hat{f}_{(+,+), \bvec{k}_{1}}^{\dagger} + \alpha_{2} \hat{f}_{(-,-), \bvec{k}_{1}}^{\dagger} \right) \left( \beta_{1} \hat{f}_{(+,+), \bvec{k}_{2}}^{\dagger} + \beta_{2} \hat{f}_{(-,-), \bvec{k}_{2}}^{\dagger} \right) \\[1ex]
    & \hspace{2em} \left( \gamma_{2} \hat{f}_{(+,-), \bvec{k}_{1}}^{\dagger} + \gamma_{2} \hat{f}_{(-,+), \bvec{k}_{1}}^{\dagger} \right) \left( \delta_{2} \hat{f}_{(+,-), \bvec{k}_{2}}^{\dagger} + \delta_{2} \hat{f}_{(-,+), \bvec{k}_{2}}^{\dagger} \right)
    \ket{\mathrm{vac}}.
  \end{split}
\end{equation}
which can be identified with
\begin{equation}
  \left(
  \begin{bmatrix}
    \alpha_{1} \\
    \alpha_{2}
  \end{bmatrix}
  \otimes
  \begin{bmatrix}
    \beta_{1} \\
    \beta_{2}
  \end{bmatrix}
  \right)
  \otimes
  \left(
  \begin{bmatrix}
    \gamma_{1} \\
    \gamma_{2}
  \end{bmatrix}
  \otimes
  \begin{bmatrix}
    \delta_{1} \\
    \delta_{2}
  \end{bmatrix}
\right)
\in ( \CC^{2} \otimes \CC^{2} ) \otimes (\CC^{2} \otimes \CC^{2}).
\end{equation}
This mapping is well defined once we fix an order on the set of momenta and the creation operators in each of the chiral sector.
A basis change by $\mathsf{U}(N_{occ}) \times \mathsf{U}(N_{occ})$ transforms the occupation vector as
\begin{equation}
  \left(
  \begin{bmatrix}
    \alpha_{1} \\
    \alpha_{2}
  \end{bmatrix}
  \otimes
  \begin{bmatrix}
    \beta_{1} \\
    \beta_{2}
  \end{bmatrix}
  \right)
  \otimes
  \left(
  \begin{bmatrix}
    \gamma_{1} \\
    \gamma_{2}
  \end{bmatrix}
  \otimes
  \begin{bmatrix}
    \delta_{1} \\
    \delta_{2}
  \end{bmatrix}
\right)
\mapsto
U_{1}^{\otimes 2}
\left(
  \begin{bmatrix}
    \alpha_{1} \\
    \alpha_{2}
  \end{bmatrix}
  \otimes
  \begin{bmatrix}
    \beta_{1} \\
    \beta_{2}
  \end{bmatrix}
\right)
\otimes
U_{2}^{\otimes 2}
\left(
  \begin{bmatrix}
    \gamma_{1} \\
    \gamma_{2}
  \end{bmatrix}
  \otimes
  \begin{bmatrix}
    \delta_{1} \\
    \delta_{2}
  \end{bmatrix}
\right)
\end{equation}
which is clearly a faithful representation.
For arbitrary $N_{\bvec{k}}$ in the valleyful model, we can embed the set of all states with $\lambda_{+} = 1$ into $(\CC^{2})^{\otimes N_{\bvec{k}}} \otimes (\CC^{2})^{\otimes N_{\bvec{k}}}$ in the same way.

For the valleyful and spinful model, the mapping to an occupation vector is nearly identical to the valleyful case.
For example, when $N_{\bvec{k}} = 1$ and $\lambda_{+} = 1$ we can embed each vector from this space into $\CC^{4} \otimes (\CC^{4} \wedge \CC^{4} \wedge \CC^{4})$ and for $\lambda_{+} = 2$ we embed into $(\CC^{4} \wedge \CC^{4}) \otimes (\CC^{4} \wedge \CC^{4})$.
Note that where there are more than one electron per $\bvec{k}$, we map occupation vector of the many-body state into a wedge product; either \(\mathbb{C}^{4} \wedge \mathbb{C}^{4}\) or \(\mathbb{C}^{4} \wedge \mathbb{C}^{4} \wedge \mathbb{C}^{4}\) depending on whether there are two or three electrons at each momenta.
This choice of embedding space ensures that the embedded space remains a faithful representation of $\mathsf{U}(N_{occ}) \times \mathsf{U}(N_{occ})$.
A summary of the embedding in each case is given in~\cref{fig:occupation-vector}.

With this mapping to occupation vectors, we can now directly apply tools from representation theory to the prove the main result for the valleyful case (\cref{sec:valleyful-case}) and the valleyful and spinful case (\cref{sec:vall-spinf-case}).

\begin{figure}[h]
  \centering
  \renewcommand{\arraystretch}{1.4}
  \begin{tabular}[t]{ll}
    \multicolumn{2}{c}{\textbf{Valleyful Case}} \\
    $\lambda_{+}$ & Embedded space at each $\bvec{k}$ \\
    \hline
    0 & N/A \\
    1 & $\CC^{2} \otimes \CC^{2}$ \\
    2 & N/A \\
  \end{tabular}
  \qquad
  \renewcommand{\arraystretch}{1.4}
  \begin{tabular}[t]{ll}
    \multicolumn{2}{c}{\textbf{Valleyful and Spinful Case}} \\
    $\lambda_{+}$ & Embedded space at each $\bvec{k}$ \\
    \hline
    0 & N/A \\
    1 & $\CC^{4} \otimes (\CC^{4} \wedge \CC^{4} \wedge \CC^{4})$ \\
    2 & $(\CC^{4} \wedge \CC^{4}) \otimes (\CC^{4} \wedge \CC^{4})$ \\
    3 & $(\CC^{4} \wedge \CC^{4} \wedge \CC^{4}) \otimes \CC^{4}$ \\
    4 & N/A \\
  \end{tabular}
  \caption{Occupation vector embedding based on the value of $\lambda_{+}$}
  \label{fig:occupation-vector}
\end{figure}

\section{Tools from Representation Theory}
\label{sec:tools-from-repr}

Due to the hidden symmetry, the ground states of the valleyful models (both spinless and spinful) exhibit significantly greater complexity than those of the spinless, valleyless model. Tools from representation theory play an important role in analyzing these states. At a high level, we examine the action of the $\mathsf{U}(N_{occ})$ group on some tensor product space $(\mathbb{C}^d)^{\otimes N}$. By the Schur--Weyl duality, this space decomposes into irreducible representations (irreps) indexed by Young diagrams. We will argue that most of these irreps do not contribute to the ground state manifold.
The key results supporting this argument are the Theorem of the Highest Weight (\cref{sec:theor-high-weight}), which enables us to identify generating vectors for the irreps, and the Littlewood--Richardson rules (\cref{sec:young-tabl-littl}), which determine the irreps relevant to the ground state manifold. To illustrate these methods, we analyze a simple example of a valleyful model with two sites in \cref{sec:example-two-sites}.  

We first state without proof a few important results from the representation theory of \(\mathfrak{su}(d)\) which will be used as part of the arguments~\cref{sec:valleyful-case,sec:vall-spinf-case}.
For a more in depth discussion, we refer the reader to~\cref{sec:tools-from-repr-appendix} as well as \cite[Chapter 3]{Ragone2024}.
For our calculations, we will represent the standard basis of $\CC^{d}$, denoted $\{ e_{j} \}_{j = 1}^{d}$, with bra-ket notation $\ket{j} := e_{j}$.
When we have a tensor product space $\CC^{d} \otimes \CC^{d}$, we will also suppress the implicit tensor product symbol $e_{j} \otimes e_{k} = \ket{j} \otimes \ket{k} = \ket{j} \ket{k}$.

\subsection{Theorem of Highest Weight}
\label{sec:theor-high-weight}
We begin by stating a special case of the Theorem of Highest Weight. Here, we are working with tensor powers $(\CC^d)^{\otimes N}$ of the defining representation of the Lie group $\mathsf{SU}(d)$, which are given by the map $U\mapsto U^{\otimes N}$. This then induces a representation of the Lie algebra $\mathfrak{su}(d)$ on $(\CC^d)^{\otimes N}$ given by the map 
\[
X\mapsto X\otimes (\idty_{d\times d})^{\otimes N-1} + (\idty_{d\times d}) \otimes X\otimes (\idty_{d\times d})^{\otimes N-2} + \dots + (\idty_{d\times d})^{\otimes N-1}\otimes X.
\]
\begin{theorem}[Theorem of Highest Weight for \(\mathfrak{su}(d)\) and \((\mathbb{C}^{d})^{\otimes N}\)]
  \label{thm:theorem-of-highest-wt}
  For any \(N \in \mathbb{Z}_{+}\), consider the tensor representation \((\mathbb{C}^{d})^{\otimes N}\) of \(\mathfrak{su}(d)\) and let \(V \) be an irreducible subrepresentation.
  There exists a highest weight vector \(v \in (\mathbb{C}^{d})^{\otimes N}\) such that
  \begin{equation}
    V = \gen_{\mathfrak{su}(d)}\{ v \}.
  \end{equation}
\end{theorem}
\begin{remark}\label{rem:reps of U, SU, su}
    Since $\mathsf{SU}(d)$ is simply connected, the representations of $\mathsf{SU}(d)$ and of $\mathfrak{su}(d)$ are in one-to-one correspondence ~\cite[Theorem 5.6]{Hall2015}. Further, any irreducible representation of $\mathsf{U}(d)$ remains irreducible upon restricting to the subgroup $\mathsf{SU}(d)$.
\end{remark}
This is a special instance of the more general definition of a highest weight vector, which may be found in~\cref{sec:tools-from-repr-appendix}). For our purposes we equivalently characterize highest weight vectors in terms of the operators \(H_{j}\) and \(E_{j}\) defined as follows:
\begin{equation}
  \begin{array}{lll}
    H_{j} := \ket{j}\bra{j} && j \in \{ 1, \cdots, d \} \\[1ex]
    E_{j} := \ket{j}\bra{j + 1} && j \in \{ 1, \cdots, d-1 \}
  \end{array}
\end{equation}
Highest weight vectors may then be identified using the following proposition.
\begin{proposition}
  \label{thm:highest-wt-vector}
  A vector \(v \in (\mathbb{C}^{d})^{\otimes N}\) is a \emph{highest weight vector} for \(\mathfrak{su}(d)\) if it lies in the joint kernel of the family of operators
\begin{equation}
   \sum_{\ell = 0}^{N - 1} \left( \idty_{d \times d} \right)^{\otimes \ell} \otimes E_j \otimes \left( \idty_{d \times d} \right)^{\otimes (N - \ell - 1)} \qquad \forall j \in \{ 1, \cdots, d-1 \}.
 \end{equation}
 Additionally, \(v\) is a joint eigenvector of the family of operators
 \begin{equation}
   \sum_{\ell = 0}^{N - 1} \left( \idty_{d \times d} \right)^{\otimes \ell} \otimes H_j \otimes \left( \idty_{d \times d} \right)^{\otimes (N - \ell - 1)} \qquad \forall j \in \{ 1, \cdots, d \}.
 \end{equation}
 The tuple of eigenvalues \( ( \mu_{j} )_{j=1}^{d}\) corresponding to this family is called the \emph{weight} of \(v\).
\end{proposition}
We note that the weights \( ( \mu_{j} )_{j}\) of \(v\) allow us to distinguish highest weight vectors from different irreps:
\begin{lemma}
  Suppose that \(v\) and \(v'\) are highest weight vectors with weights \(( \mu_{j} )_{j}\) and \(( \mu_{j}' )_{j}\) respectively.
  If there exists a \(j_{*}\) so that \(\mu_{j_{*}} \neq \mu_{j_{*}}'\), then \(v \perp v'\).
  Additionally, the irreps generated by \(v\) and \(v'\) are orthogonal.
\end{lemma}
\begin{proof}
  The orthogonality of \(v\) and \(v'\) is immediate since \(\{ H_{j} \}_{j}\) is a family of self-adjoint commuting operators and hence eigenvectors with different eigenvalues are orthogonal.
  Since the representation $U^{\otimes N}$ is unitary, the irreps generated by \(v\) and \(v'\) are necessarily orthogonal ~\cite[Proposition 4.27]{Hall2015}.
\end{proof}
The Theorem of Highest Weight also holds when we consider \(V \subseteq (\mathbb{C}^{d} \wedge \mathbb{C}^{d})^{\otimes N}\) instead of \((\mathbb{C}^{d})^{\otimes N}\) with some minor modifications:
\begin{proposition}
  \Cref{thm:theorem-of-highest-wt} holds when \(\mathbb{C}^{d}\) is replaced with \(\mathbb{C}^{d} \wedge \mathbb{C}^{d}\).
  \Cref{thm:highest-wt-vector} also holds when \(\mathbb{C}^{d}\) is replaced with \(\mathbb{C}^{d} \wedge \mathbb{C}^{d}\) if the operators \(H_{j}\) are \(E_{j}\) are replaced with
  \begin{equation}
    \begin{array}{lll}
      H_{j} := \ket{j}\bra{j} \otimes \One_{d \times d} +\One_{d \times d} \otimes \ket{j}\bra{j}, & & \forall j \in \{ 1, \cdots, d \}. \\[1ex]
      E_{j} := \ket{j}\bra{j + 1} \otimes \One_{d \times d} +\One_{d \times d} \otimes \ket{j}\bra{j + 1}, & & \forall j \in \{ 1, \cdots, d \}.
    \end{array}
  \end{equation}
\end{proposition}
Now that we can characterize generating vectors for different irreps, we can turn our attention to understanding which irreps are relevant for the ground state manifold.

\subsection{Young Tableaux and the Littlewood-Richardson Rules}
\label{sec:young-tabl-littl}
For our discussion of Young tableaux and the Littlewood-Richardson rules (which we use to classify the irreps) we follow \cite[Lecture 6]{FultonHarris2004}. The irreducible subrepresentations of $\mathsf{SU}(d)$\footnote{It should be noted that the discussion in Fulton and Harris \cite{FultonHarris2004} considers representations of $GL(d)$--however, these particular irreps of $GL(d)$ remain irreducible after restricting to $\mathsf{SU}(d)$.} in $(\CC^d)^{\otimes N}$ are labeled by partitions $\lambda$ of $N$, where $\lambda = (\lambda_1,\lambda_2,\dots , \lambda_r)$ with $\lambda_1\geq \lambda_2\geq \dots \geq \lambda_r \geq 0$ and $N = \lambda_1 + \dots +\lambda_r$. We write $\mathbb{S}_\lambda(\CC^d)$ to denote the irrep corresponding to a partition $\lambda$. Partitions in turn are in one-to-one correspondence with Young diagrams. Here are a few possible diagrams for $N=4$:
\[
    (1^4)=(1,1,1,1)\leftrightarrow \ydiagram{1,1,1,1}, \qquad (2,1,1) \leftrightarrow \ydiagram{2,1,1}, \qquad (4) \leftrightarrow \ydiagram{4}
\] In this manner Young diagrams may be used to label certain irreducible representations of $\mathsf{SU}(d)$. For example, the partition $\lambda = (1^k)$ corresponds to the exterior power representation $\mathbb{S}_\lambda(\CC^d) = \bigwedge^k \CC^d$, while the partition $\mu = (r)$ corresponds to the $r^{th}$ symmetric power representation $\mathbb{S}_\mu(\CC^d)$. Note that if $\lambda$ a partition such that $\lambda_{d+1}\neq 0$, then $\mathbb{S}_\lambda(\CC^d)$ is zero.

Now, given two irreducible representations $\mathbb{S}_\lambda(\CC^d)$ and $ \mathbb{S}_\mu(\CC^d)$, it is natural to wonder what we can say about the irreps of the tensor product $\mathbb{S}_\lambda(\CC^d) \otimes \mathbb{S}_\mu(\CC^d)$.
This can be mapped to a combinatorial problem whose solution is given by the Littlewood-Richardson rules.
For our purposes we need only a special case of the Littlewood-Richardson rule, which may be found in \cite[Equation 6.9]{FultonHarris2004}.

\begin{theorem}[Littlewood-Richardson Rule, special case]
  \label{thm:littlewood-richardson_maintext}
Let $\lambda$ be a partition of $N$ and $m\leq d$ a positive integer. Then we have the following decomposition into irreducibles:
\[
  \mathbb{S}_\lambda(\CC^d)\otimes \mathbb{S}_{(1^m)}(\CC^d) \cong \bigoplus_{\pi} \mathbb{S}_\pi(\CC^d), 
\]
where the sum is over partitions $\pi$ of $N+m$ whose Young diagram is obtained from that of $\lambda$ by adding $m$ boxes, with no two in the same row. In particular this decomposition is multiplicity-free.
\end{theorem}
For example, the tensor product $S_{(2,2)}(\CC^{d}) \otimes S_{(1^2)}(\CC^{d})$ splits into a direct sum of irreducibles:
\begin{equation}
  \ydiagram{2,2} \otimes \ydiagram{1,1}
  =
  \ydiagram{3,3} \oplus \ydiagram{3,2,1} \oplus \ydiagram{2,2,1,1}.
\end{equation}
Note that the Young diagrams
\begin{equation}
  \ydiagram{4,2} \qquad \qquad \ydiagram{2,2,2}
\end{equation}
are not part of this decomposition because they require adding two boxes in the same row.
We also note that given a Young diagram, we may readily compute the dimension of the corresponding irrep via the Hook Length Formula \cite[Theorem 6.3]{FultonHarris2004}:
\begin{theorem}[Hook Length Formula for $\mathsf{SU}(d)$]
  \label{thm:hook-formula}
  Let $\lambda = (\lambda_1,\dots,\lambda_N)$ be a partition of $N$ and let \(\mathbb{S}_{\lambda}(\mathbb{C}^{d})\) denote an irrep of \(\mathsf{SU}(d)\). Then
  \[
    \dim \mathbb{S}_\lambda (\CC^d) = \prod_{(i,j) \in \lambda} \frac{(d + j - i)}{h(i,j)}. 
  \]
  where \(h(i, j)\) is the hook length of the box in the \(i\)-th row and \(j\)-th column of the Young diagram \(\lambda\) and \((i,j) \in \lambda\) denotes the position of a box in the Young diagram.
\end{theorem}

Finally, we may readily interpret the irreps $\mathbb{S}_\lambda(\CC^d)$ labeled by Young diagrams $\lambda$ from the perspective of weights \cite[Proposition 15.15]{FultonHarris2004}:

\begin{proposition}
  The representation $\mathbb{S}_\lambda(\CC^d)$ is the irreducible representation of $\mathfrak{sl}(d,\CC)$ with highest weight $\mu = (\lambda_1,\lambda_2, \dots,\lambda_d)$.
\end{proposition}

\subsection{A Two Site Example with \(\mathsf{SU}(2)\)}
\label{sec:example-two-sites}
Let us now put these ideas into action to decompose the two-site tensor representation $\CC^2\otimes \CC^2$ into irreps of $\mathsf{SU}(2)$.

First, consider the defining representation $(\Pi,\CC^2)$ of $\mathsf{SU}(2)$, which is given by the map $\Pi(U)=U$ for all $U\in \mathsf{SU}(2)$. This descends to the defining representation of the Lie algebra $\mathfrak{su}(2)$, which uniquely extends to the defining representation $(\pi,\CC^2)$ of $\fraksl(2,\CC)$, which is given by $\pi(X)=X$ for all $X\in\fraksl(2,\CC)$. This has a highest weight vector $\ket{1}$:
\[
    \pi(E_{1}) \ket{1} = E_{1} \ket{1} = 0.
\] In fact, the highest weight space is 1-dimensional since $E_{1}\ket{2}\neq 0$. Thus by~\cref{thm:highest-wt-vector}, the defining representation is an irrep of highest weight with weight $(1,0)$.
Let us denote this irrep by $\mathbb{S}_{(1)}(\CC^2)$, or diagrammatically by the box
\[
  \mathbb{S}_{(1)}(\CC^2) = \ydiagram{1}.
\] 

We now demonstrate the 2-site tensor representation. At the group level, this is the representation $\Pi:\mathsf{SU}(2)\to GL(\CC^2\otimes \CC^2)$ given by $\Pi(U)=U\otimes U$. This passes to the Lie algebra representation $\pi:\fraksl(2,\CC)\to \mathfrak{gl}(\CC^2\otimes \CC^2)$ given by $\pi(X) = X\otimes \idty + \idty \otimes X$. To decompose this into irreps, we look for highest weight vectors by computing the kernel of $\pi(E_{1}) = E_{1}\otimes \idty + \idty\otimes E_{1}$. We see that this kernel is two dimensional and spanned by the vectors $\ket{1}\ket{1}$ and $\ket{1} \ket{2} - \ket{2} \ket{1}$. We may quickly calculate the weights of these two highest weight vectors:
\begin{equation}\begin{split}
  \Big( \One \otimes H_1 + H_{1} \otimes \One \Big) \ket{1}\ket{1} & = 2 \ket{1} \ket{1} \\
  \Big( \One \otimes H_2 + H_{2} \otimes \One \Big) \ket{1}\ket{1} & = 0 \ket{1} \ket{1}
\end{split}\end{equation}
\begin{equation}\begin{split}
  \Big( \One \otimes H_1 + H_{1} \otimes \One \Big) (\ket{1}\ket{2} - \ket{2} \ket{1}) & = 1 ( \ket{1} \ket{2} - \ket{2} \ket{1} ) \\
  \Big( \One \otimes H_2 + H_{2} \otimes \One \Big) (\ket{1}\ket{2} - \ket{2} \ket{1}) & =  1 ( \ket{1} \ket{2} - \ket{2} \ket{1}) .
\end{split}\end{equation}
So $\ket{1}\ket{1}$ is a highest weight vector of weight $(2,0)$ and $\ket{1}\ket{2}-\ket{2}\ket{1}$ is a highest weight vector of weight $(1,1)$.
Thus, by~\cref{thm:highest-wt-vector}, they cyclically generate orthogonal irreps.
We denote
\begin{equation}
    \mathbb{S}_{(2)}(\CC^2) := \gen_{\fraksl(2,\CC)} \{ \ket{1} \ket{1} \}, \qquad \mathbb{S}_{(1,1)}(\CC^2) := \gen_{\fraksl(2,\CC)} \{ \ket{1} \ket{2} - \ket{2} \ket{1} \}.
\end{equation}
In particular, note that we have completely decomposed $\CC^2\otimes \CC^2$ into a direct sum of irreps:\footnote{To connect to the language of spin in physics, the irrep $\mathbb{S}_{(1,1)}(\CC^2)$ is often called a singlet, while the irrep $\mathbb{S}_{(2)}(\CC^2)$ is often called a triplet.}
\[
    \CC^2 \otimes \CC^2 \cong \mathbb{S}_{(2)}(\CC^2) \oplus \mathbb{S}_{(1,1)}(\CC^2) . 
\] We will conclude with the diagrammatic version of the above isomorphism:
\begin{equation} \label{eqn:sl2 2 tensor ydiagram}
    \ydiagram{1}\otimes \ydiagram{1} \cong \ydiagram{2} \oplus \ydiagram{1,1}.
\end{equation}

\subsection{Additional Notation}
\label{sec:additional-notation}
Before proceeding, we introduce various notations which will be used in the arguments in~\cref{sec:valleyful-case,sec:vall-spinf-case}.
As introduced in the previous section, we will use the bra-ket notation \(\ket{i} \in \mathbb{C}^{d}\) to denote the \(i^{\mathrm{th}}\) standard basis vector.
When we consider the anti-symmetric space \(\mathbb{C}^{d} \wedge \mathbb{C}^{d}\), we use the mild abuse of notation $\ket{i \wedge j}$ to denote $\ket{i} \ket{j} - \ket{j} \ket{i} \in \mathbb{C}^{d} \wedge \mathbb{C}^{d}$ and take a similar convention $\ket{i \wedge j \wedge k}$ and \(\mathbb{C}^{d} \wedge \mathbb{C}^{d} \wedge \mathbb{C}^{d}\).
We also define the \textit{symmetrizer} with respect to an orthogonal basis.
\begin{definition}[Symmetrizer]
  Given set of orthogonal basis vectors $\{ \ket{j} \}_{j=1}^{n}$ on a tensor product space $\bigotimes^{N} \CC^{d}$, we define the symmetrized state
  \begin{equation}
    \Sym\big[ \ket{j_{1}} \ket{j_{2}} \cdots \ket{j_{N}} \big] = \frac{1}{N!} \sum_{\sigma \in S_{N}}^{} \ket{\sigma(j_{1})} \ket{\sigma(j_{2})} \cdots \ket{\sigma(j_{N})} ,
  \end{equation} where $S_N$ denotes the symmetric group on $N$ letters. 
\end{definition}
As a concrete example of this notation,
\begin{equation}
  \Sym\big[ \ket{1}\ket{1}\ket{2} \big] = \frac{1}{3} \Big( \ket{1} \ket{1} \ket{2} + \ket{1} \ket{2} \ket{1} + \ket{2} \ket{1} \ket{1} \Big).
\end{equation}
It will also be useful to define the set cyclically generated by a representation $(\rho,V)$:
\begin{definition}
  Let $(\rho,V)$ be a representation of a Lie group $G$ (or Lie algebra $\mathfrak{g}$). For any vector $v\in V$, we let
  \begin{equation}
    \gen_{G} \{ v \} := \spn\{ \rho(g) v : g \in G \}
  \end{equation} denote the subrepresentation cyclically generated by $v$.
\end{definition}

Finally, we introduce some notations to assist describing the many-body ground states.
In what follows, we will often work within either the $+$ chiral sector or $-$ chiral sector.
These different sectors have either 2 or 4 creation operators so we can identify them with the integers $\{ 1^{\pm}, 2^{\pm} \}$ or $\{ 1^{\pm}, 2^{\pm}, 3^{\pm}, 4^{\pm} \}$ depending on the case.
\begin{definition}
  \label{rem:basis-notation}
  For the valleyful model, we make the identifications
  \begin{equation}
    \begin{array}{c}
      (+, +) \leftrightarrow 1^{+} \qquad (-, -) \leftrightarrow 2^{+} \\[1ex]
      (+, -) \leftrightarrow 1^{-} \qquad (-, +) \leftrightarrow 2^{-}.
    \end{array}
  \end{equation}
  With this identification, we have, for example \(\hat{f}_{(+, +), \bvec{k}}^{\dagger} = \hat{f}_{1^{+}, \bvec{k}}^{\dagger}\).
  For the valleyful and spinful mode, we make the identifications
  \begin{equation}
    \begin{array}{c}
      (+, +, \uparrow) \leftrightarrow 1^{+} \qquad (+, +, \downarrow) \leftrightarrow 2^{+} \qquad (-, -, \uparrow) \leftrightarrow 3^{+} \qquad (-, -, \downarrow) \leftrightarrow 4^{+} \\[1ex]
      (+, -, \uparrow) \leftrightarrow 1^{-} \qquad (+, -, \downarrow) \leftrightarrow 2^{-} \qquad (-, +, \uparrow) \leftrightarrow 3^{-} \qquad (-, +, \downarrow) \leftrightarrow 4^{-}.
    \end{array}
  \end{equation}
\end{definition}
We will also be interested in defining notation for basis rotations on each chiral sector:
\begin{definition}[Basis Rotation]
  For the valleyful model and any $U \in \mathsf{U}(2)$, we define
  \begin{equation}
    U \circ \hat{f}_{(+,+), \bvec{k}}^{\dagger}  = \sum_{m=1}^{2} \hat{f}_{m^{+}, \bvec{k}}^{\dagger} [U]_{m, (+,+)} 
  \end{equation}
  which rotates the creation operator $\hat{f}_{(+,+), \bvec{k}}$ by $U$ within the $+$ chiral sector; $U \circ \hat{f}_{(+,-), \bvec{k}}$ is defined analogously.

  For the valleyful and spinful model and any $U \in \mathsf{U}(4)$, we define
  \begin{equation}
    \begin{split}
      U \circ \hat{f}_{(+,+,\uparrow), \bvec{k}}^{\dagger} & = \sum_{m=1}^{4} \hat{f}_{m^{+}, \bvec{k}}^{\dagger} [U]_{m, (+,+,\uparrow)}  \\
      U \circ \hat{f}_{(+,+,\uparrow), \bvec{k}}^{\dagger} \hat{f}_{(+,+,\downarrow), \bvec{k}}^{\dagger} & = \sum_{m=1}^{4} \sum_{n=1}^{4}  \Big( \hat{f}_{m^{+}, \bvec{k}}^{\dagger} [U]_{m, (+, +, \uparrow)} \Big) \Big(  \hat{f}_{n^{+}, \bvec{k}}^{\dagger} [U]_{n, (+, +, \downarrow)}  \Big)
    \end{split}
  \end{equation}
  which rotates each of the creation operators at momentum $\bvec{k}$ in the $+$ chiral sector by $U$.
  The rotation
  \begin{equation}
    U \circ \hat{f}_{(+,+,\uparrow), \bvec{k}}^{\dagger} \hat{f}_{(+,+,\downarrow), \bvec{k}}^{\dagger} \hat{f}_{(-,-,\uparrow), \bvec{k}}
  \end{equation}
  is defined analogously.
  As well as the corresponding rotations on the $-$ chiral sector.
\end{definition}

\section{Ground State Manifold of the Spinless, Valleyful Model}
\label{sec:valleyful-case}
The main goal of this section is to use the representation theory of $\mathsf{U}(2) \times \mathsf{U}(2)$ to prove the following result: 
\begin{theorem}
  \label{prop:valleyful-case}
  A many-body state $\ket{\Phi}$ is a ground state of valleyful MATBG if and only if it can be written as
  \begin{equation}
    \label{eq:valleyful-particle-number}
    \ket{\Phi} = \alpha \ket{\tilde{\Phi}_{0}} + \beta \ket{\tilde{\Phi}_{1}} + \gamma \ket{\tilde{\Phi}_{2}} \quad \alpha, \beta, \gamma \in \CC
  \end{equation}
  where
  \begin{equation}
    \ket{\tilde{\Phi}_{2}} = \prod_{\bvec{k} \in \mathcal{K}}^{} \hat{f}_{(+,+), \bvec{k}}^{\dagger} \hat{f}_{(-, -), \bvec{k}}^{\dagger} \ket{\mathrm{vac}}, \qquad
    \ket{\tilde{\Phi}_{0}} = \prod_{\bvec{k} \in \mathcal{K}}^{} \hat{f}_{(+,-), \bvec{k}}^{\dagger} \hat{f}_{(-, +), \bvec{k}}^{\dagger} \ket{\mathrm{vac}} 
  \end{equation}
  and 
  \begin{equation}
    \ket{\tilde{\Phi}_{1}} \in \spn\left\{\prod_{\bvec{k} \in \mathcal{K}}^{} \left( U_{1} \circ \hat{f}_{(+, +), \bvec{k}}^{\dagger} \right) \left( U_{2} \circ \hat{f}_{(+, -), \bvec{k}}^{\dagger} \right) \ket{\mathrm{vac}} : U_{1}, U_{2} \in \mathsf{U}(2) \right\}.
  \end{equation}
  Alternatively, the state $\ket{\tilde{\Phi}_{1}}$ may be written as a linear combination of $(N_{\bvec{k}} + 1)^{2}$ linearly independent states so the total dimension of the ground space is $(N_{\bvec{k}} + 1)^{2} + 2$.
\end{theorem}
As discussed in \cref{sec:repr-theory-ground}, the decomposition in~\cref{eq:valleyful-particle-number} is implied~\cref{lem:uniform-filling}.
Since the forms of $\ket{\tilde{\Phi}_{0}}$ and $\ket{\tilde{\Phi}_{2}}$ are fixed, the remaining part of the proof is to characterize the intersection
\begin{equation}
  \label{eq:valleyful-intersection}
  \Big\{ \ket{\Phi} : \hat{n}_{+, \bvec{k}} \ket{\Phi} = \hat{n}_{-, \bvec{k}} \ket{\Phi} = \ket{\Phi},~\forall \bvec{k} \in \mathcal{K} \Big\} \bigcap \Big\{ \ket{\Phi} : \hat{C}_{\pm, \bvec{k}, \bvec{k}'} \ket{\Phi} = 0~\forall \bvec{k}, \bvec{k}' \in \mathcal{K}, \bvec{k} \neq \bvec{k}' \Big\}
\end{equation}
as this is the set of states with $\lambda_{+} = \lambda_{-} = 1$ which are also ground states of $\hat{H}_{FBI}$.

\begin{proof}[Proof of~\cref{prop:valleyful-case}]
Using the occupation vector representation (\cref{sec:occup-vect-repr}) we can bijectively map the set of states with $\lambda_{+} = \lambda_{-} = 1$ to $(\CC^{2})^{\otimes N_{\bvec{k}}} \otimes (\CC^{2})^{\otimes N_{\bvec{k}}}$ where the first $(\CC^{2})^{\otimes N_{\bvec{k}}}$ corresponds to the $+$ chiral sector and the second $(\CC^{2})^{\otimes N_{\bvec{k}}}$ corresponds to the $-$ chiral sector.
To characterize the ground state, we will construct a decreasing sequence of $(\CC^{2})^{\otimes N_{\bvec{k}}}= V^{(1)} \supset V^{(2)} \supset \cdots \supset V^{(N_{\bvec{k}})}$ so that a many-body state $\ket{\Phi}$ with $\lambda_{+} = \lambda_{-} = 1$ satisfies $\hat{C}_{+, \bvec{k}, \bvec{k}'} \ket{\Phi} = 0$ for all $\bvec{k} \neq \bvec{k}'$ only if the occupation vector for $\ket{\Phi}$ lies in $V^{(N_{\bvec{k}})} \otimes (\CC^{2})^{\otimes N_{\bvec{k}}}$.
Due to to the symmetry between the two chiral sectors, it follows that a many-body state is a ground state if and only if the occupation vector lies in $V^{(N_{\bvec{k}})} \otimes V^{(N_{\bvec{k}})}$.
Note that since the group representations on the \(+\) and \(-\) chiral sectors commute, we may consider the representation theory of the full space by considering each sector as independent representations of $\mathsf{U}(2)$. For our purposes we may then further restrict to $\mathsf{SU}(2)$ per Remark~\ref{rem:reps of U, SU, su}.

To construct the first non-trivial vector space, $V^{(2)}$, we decompose the tensor product $\CC^{2} \otimes \CC^{2}$ into a direct sum of irreps $\mathbb{S}_{(2)}(\mathbb{C}^{2}) \oplus \mathbb{S}_{(1,1)}(\mathbb{C}^{2})$ of $\mathsf{SU}(2)$. This case is explicitly worked in~\cref{sec:example-two-sites}, and a special case of the Littlewood-Richardson rule (\cref{thm:littlewood-richardson_maintext}).
\begin{equation}
  \begin{array}{rlccrl}
    V^{(1)} & = \CC^{2} \otimes \CC^{2} \otimes (\CC^{2})^{\otimes (N_{\bvec{k}} - 2)} & & & V^{(1)} &  = \ydiagram{1} \otimes \ydiagram{1} \otimes (\CC^{2})^{\otimes (N_{\bvec{k}} - 2)} \\[1ex]
     & = ( \mathbb{S}_{(2)}(\mathbb{C}^{2}) \oplus \mathbb{S}_{(1,1)}(\mathbb{C}^{2}) ) \otimes (\CC^{2})^{\otimes (N_{\bvec{k}} - 2)} & & & & = \left(~\ydiagram{2} \oplus \ydiagram{1,1} ~\right) \otimes (\CC^{2})^{\otimes (N_{\bvec{k}} - 2)}.
  \end{array}
\end{equation}
Using the Theorem of Highest Weight, one can check that the irreps $\mathbb{S}_{(2)}$ and $\mathbb{S}_{(1,1)}$ are cyclically generated by:
\begin{equation}
  \begin{split}
    \mathbb{S}_{(2)}(\mathbb{C}^{2}) & := \gen_{\mathsf{U}(2)}\big\{ \ket{1^{+}}^{\otimes 2} \big\} \\[1ex]
    \mathbb{S}_{(1,1)}(\mathbb{C}^{2}) & := \gen_{\mathsf{U}(2)}\big\{ \ket{1^{+}} \ket{2^{+}} - \ket{2^{+}} \ket{1^{+}} \big\}
  \end{split}
\end{equation}
We can now transform these occupation vectors back to many body states.
With this transformation (and recalling the notation introduced in \cref{rem:basis-notation}) the generating vectors map to 
\begin{equation}
  \begin{split}
   & \hat{f}_{1^{+}, \bvec{k}_{1}}^{\dagger} \hat{f}_{1^{+}, \bvec{k}_{2}}^{\dagger} \ket{\Phi_{\rm rest}} \\[1ex]
   & \left( \hat{f}_{1^{+}, \bvec{k}_{1}}^{\dagger} \hat{f}_{2^{+}, \bvec{k}_{2}}^{\dagger} - \hat{f}_{2^{+}, \bvec{k}_{1}}^{\dagger} \hat{f}_{1^{+}, \bvec{k}_{2}}^{\dagger}  \right) \ket{\Phi_{\rm rest}}
  \end{split}
\end{equation}
where $\ket{\Phi_{\rm rest}}$ is an arbitrary state on the remaining $\mathcal{K} \setminus \{ \bvec{k}_{1}, \bvec{k}_{2} \}$ momenta in the \(+\) chiral sector.
These two vectors (when restricted to the first two momenta) likewise generate two irreps of $\mathsf{U}(2)$.
We can now calculate that
\begin{equation}
  \begin{split}
   & \hat{C}_{+, \bvec{k}_{1}, \bvec{k}_{2}} \left( \hat{f}_{1^{+}, \bvec{k}_{1}}^{\dagger} \hat{f}_{1^{+}, \bvec{k}_{2}}^{\dagger} \right) \ket{\Phi_{\rm rest}}  = \hat{f}_{1^{+}, \bvec{k}_{1}}^{\dagger} \hat{f}_{1^{+},\bvec{k}_{1}}^{\dagger} \ket{\Phi_{\rm rest}} = 0 \\[1ex]
   & \hat{C}_{+, \bvec{k}_{1}, \bvec{k}_{2}} \left( \hat{f}_{1^{+}, \bvec{k}_{1}}^{\dagger} \hat{f}_{2^{+}, \bvec{k}_{2}}^{\dagger} - \hat{f}_{2^{+}, \bvec{k}_{1}}^{\dagger} \hat{f}_{1^{+}, \bvec{k}_{2}}^{\dagger}  \right) \ket{\Phi_{\rm rest}} = 2 \hat{f}_{1^{+}, \bvec{k}_{1}}^{\dagger} \hat{f}_{2^{+}, \bvec{k}_{1}}^{\dagger} \ket{\Phi_{\rm rest}} \neq 0
  \end{split}
\end{equation}
Therefore, states from $\mathbb{S}_{(2)}(\mathbb{C}^{2}) \otimes (\CC^{2})^{\otimes (N_{\bvec{k}} - 2)}$ may lie in the ground state manifold however states from $\mathbb{S}_{(1,1)} \otimes (\CC^{2})^{\otimes (N_{\bvec{k}} - 2)}$ do not.
Therefore, we define $V^{(2)} := \mathbb{S}_{(2)}(\mathbb{C}^{2}) \otimes (\CC^{2})^{\otimes (N_{\bvec{k}} - 2)}$ and proceed with the induction.

To construct $V^{(3)}$ from $V^{(2)}$, we once again use the Littlewood-Richardson rule
\begin{equation}
  \begin{array}{rlccrl}
    V^{(2)} & = \mathbb{S}_{(2)}(\mathbb{C}^{2}) \otimes \CC^{2} \otimes (\CC^{2})^{\otimes (N_{\bvec{k}} - 3)} & & & V^{(2)} &  = \ydiagram{2} \otimes \ydiagram{1} \otimes (\CC^{2})^{\otimes (N_{\bvec{k}} - 3)} \\[1ex]
     & = ( \mathbb{S}_{(3)}(\mathbb{C}^{2}) \oplus \mathbb{S}_{(2,1)}(\mathbb{C}^{2}) ) \otimes (\CC^{2})^{\otimes (N_{\bvec{k}} - 3)} & & & & = \left(~\ydiagram{3} \oplus \ydiagram{2,1} ~\right) \otimes (\CC^{2})^{\otimes (N_{\bvec{k}} - 3)}
  \end{array}
\end{equation}
By the Theorem of Highest Weight, we can verify these spaces have generating vectors
\begin{equation}
  \begin{split}
    \mathbb{S}_{(3)}(\mathbb{C}^{2}) & := \gen_{\mathsf{U}(2)}\big\{ \ket{1^{+}}^{\otimes 3} \big\} \\[1ex]
    \mathbb{S}_{(2,1)}(\mathbb{C}^{2}) & := \gen_{\mathsf{U}(2)}\big\{ \Sym\Big[\ket{1^{+}} \ket{2^{+}} \Big] \ket{1^{+}} - \ket{1^{+}}^{\otimes 2} \ket{2^{+}}  \big\}.
  \end{split}
\end{equation}
The corresponding many-body states for these generating vectors are
\begin{equation}
  \begin{split}
   & \hat{f}_{1^{+}, \bvec{k}_{1}}^{\dagger} \hat{f}_{1^{+}, \bvec{k}_{2}}^{\dagger} \hat{f}_{1^{+}, \bvec{k}_{3}} \ket{\Phi_{\rm rest}} \\[1ex]
   & \left( \frac{1}{2} \hat{f}_{1^{+}, \bvec{k}_{1}}^{\dagger} \hat{f}_{2^{+}, \bvec{k}_{2}}^{\dagger} \hat{f}_{1^{+}, \bvec{k}_{3}}^{\dagger} + \frac{1}{2} \hat{f}_{2^{+}, \bvec{k}_{1}}^{\dagger} \hat{f}_{1^{+}, \bvec{k}_{2}}^{\dagger} \hat{f}_{1^{+}, \bvec{k}_{3}}^{\dagger} - \hat{f}_{1^{+}, \bvec{k}_{1}}^{\dagger} \hat{f}_{1^{+}, \bvec{k}_{2}}^{\dagger} \hat{f}_{2^{+}, \bvec{k}_{3}}^{\dagger}  \right) \ket{\Phi_{\rm rest}}
  \end{split}
\end{equation}
where $\ket{\Phi_{\rm rest}}$ is now an arbitrary state on the momenta $\mathcal{K} \setminus \{ \bvec{k}_{1}, \bvec{k}_{2}, \bvec{k}_{3} \}$ in the \(+\) chiral sector.
It is easily seen the generator for $\mathbb{S}_{(3)}(\mathbb{C}^{2})$ is annihilated by $\hat{C}_{+, \bvec{k}_{2}, \bvec{k}_{3}}$.
For the generator for $\mathbb{S}_{(2,1)}$ we instead have
\begin{equation}
  \label{eq:a2-valleyful-check}
  \begin{split}
    & \hat{C}_{+, \bvec{k}_{2}, \bvec{k}_{3}}\left( \frac{1}{2} \hat{f}_{1^{+}, \bvec{k}_{1}}^{\dagger} \hat{f}_{2^{+}, \bvec{k}_{2}}^{\dagger} \hat{f}_{1^{+}, \bvec{k}_{3}}^{\dagger} + \frac{1}{2} \hat{f}_{2^{+}, \bvec{k}_{1}}^{\dagger} \hat{f}_{1^{+}, \bvec{k}_{2}}^{\dagger} \hat{f}_{1^{+}, \bvec{k}_{3}}^{\dagger} - \hat{f}_{1^{+}, \bvec{k}_{1}}^{\dagger} \hat{f}_{1^{+}, \bvec{k}_{2}}^{\dagger} \hat{f}_{2^{+}, \bvec{k}_{3}}^{\dagger}  \right) \ket{\Phi_{\rm rest}} \\
    & \hspace{2em} = -\frac{3}{2} \hat{f}_{1^{+},\bvec{k}_{1}}^{\dagger} \hat{f}_{1^{+},\bvec{k}_{2}}^{\dagger} \hat{f}_{2^{+}, \bvec{k}_{2}}^{\dagger} \ket{\Phi_{\rm rest}} \neq 0
  \end{split}
\end{equation}
So states from $\mathbb{S}_{(3)}(\mathbb{C}^{2}) \otimes (\CC^{2})^{\otimes (N_{\bvec{k}} - 3)}$ may lie in the ground space however states from $\mathbb{S}_{(2,1)} \otimes (\CC^{2})^{\otimes (N_{\bvec{k}} - 3)}$ do not.
Continuing the induction, we define $V^{(3)} := \mathbb{S}_{(3)}(\mathbb{C}^{2}) \otimes (\CC^{2})^{\otimes (N_{\bvec{k}} - 3)}$.
At the $n^{\rm th}$ stage of the induction, we perform the decomposition:
\begin{equation}
  \begin{split}
    V^{(n)}
    & = ( \mathbb{S}_{(n)} \otimes \CC^{2} ) \otimes (\CC^{2})^{\otimes (N_{\bvec{k}} - (n + 1))} \\
    & = ( \mathbb{S}_{(n+1)} \oplus \mathbb{S}_{(n,1)} ) \otimes (\CC^{2})^{\otimes (N_{\bvec{k}} - (n + 1))} 
  \end{split}
\end{equation}
Here $\mathbb{S}_{(n+1)}$ and $\mathbb{S}_{(n,1)}$ have corresponding Young diagrams 
\begin{equation}
  \mathbb{S}_{(n+1)} = \overbrace{\ydiagram{7}}^{n+1} \qquad \mathbb{S}_{(n,1)} = \overbrace{\ydiagram{6,1}}^{n}.
\end{equation}
As before, we can use the theorem of highest weight to verify that the following vectors generate the two irreps 
\begin{equation}
  \begin{split}
    \mathbb{S}_{(n+1)}(\mathbb{C}^{2}) & := \gen_{\mathsf{U}(2)}\big\{ \ket{1^{+}}^{\otimes (n+1)} \big\} \\[1ex]
    \mathbb{S}_{(n,1)}(\mathbb{C}^{2}) & := \gen_{\mathsf{U}(2)}\big\{ \Sym\Big[\ket{1^{+}}^{\otimes (n-1)} \ket{2^{+}}\Big] \ket{1^{+}} - \ket{1^{+}}^{\otimes n}\ket{2^{+}}  \big\}
  \end{split}
\end{equation}
and repeating a similar calculation to~\cref{eq:a2-valleyful-check} confirms states with occupation vectors from $\mathbb{S}_{(n,1)}(\mathbb{C}^{2}) \otimes (\CC^{2})^{\otimes (N_{\bvec{k}} - (n + 1))}$ cannot be ground states for all $n$.

At the final stage of the induction, we conclude the only possible subspace for the ground states is
\begin{equation}
 \mathbb{S}_{(n)}(\mathbb{C}^{2}) \otimes (\CC^{2})^{\otimes N_{\bvec{k}}} =
 \gen_{\mathsf{U}(2)}\big\{ \ket{1^{+}}^{\otimes N_{\bvec{k}}} \big\} \otimes (\CC^{2})^{\otimes N_{\bvec{k}}}
\end{equation}
which corresponds to
\begin{equation}
  \spn\Big\{\prod_{\bvec{k} \in \mathcal{K}}^{} \left( U_{1} \circ \hat{f}_{(+,+), \bvec{k}}^{\dagger} \right) \ket{\Phi_{\rm rest}} : U_{1} \in \mathsf{U}(2) \Big\}.
\end{equation}
Applying the same analysis to the $-$ chiral sector with the hook formula (\cref{thm:hook-formula}) completes the proof.
\end{proof}

\section{Representation Theory for the Spinful and Valleyful Model}
\label{sec:vall-spinf-case}
The main goal of this section is to use the representation theory of $\mathsf{U}(4) \times \mathsf{U}(4)$ to prove the following result:
\begin{theorem}
  \label{prop:valleyful-spinful-case}
  A many-body state $\ket{\Phi}$ is a ground state of spinful and valleyful MATBG if and only if it can be written as
  \begin{equation}
    \label{eq:valleyful-spinful-particle-number}
    \ket{\Phi} = \alpha \ket{\tilde{\Phi}_{0}} + \beta \ket{\tilde{\Phi}_{1}} + \gamma \ket{\tilde{\Phi}_{2}} + \delta \ket{\tilde{\Phi}_{3}} + \theta \ket{\tilde{\Phi}_{4}} \quad \alpha, \beta, \gamma, \delta, \theta \in \CC
  \end{equation}
  where
  \begin{equation}
    \ket{\tilde{\Phi}_{4}} =  \prod_{\bvec{k} \in \mathcal{K}}^{} \prod_{s \in \{ \uparrow, \downarrow \}}^{}  \hat{f}_{(+,+,s), \bvec{k}}^{\dagger} \hat{f}_{(-, -,s), \bvec{k}}^{\dagger} \ket{\mathrm{vac}}, \qquad
    \ket{\tilde{\Phi}_{0}} = \prod_{\bvec{k} \in \mathcal{K}}^{} \prod_{s \in \{ \uparrow, \downarrow\}}^{}  \hat{f}_{(+,-,s), \bvec{k}}^{\dagger} \hat{f}_{(-, +,s), \bvec{k}}^{\dagger} \ket{\mathrm{vac}}
  \end{equation}
  and
  \begin{equation}
    \begin{split}
      \ket{\tilde{\Phi}_{1}} & \in \spn\left\{ \prod_{\bvec{k} \in \mathcal{K}}^{} \left( U_{1} \circ \hat{f}_{(+,+, \uparrow), \bvec{k}}^{\dagger} \hat{f}_{(+,+, \downarrow), \bvec{k}}^{\dagger}  \hat{f}_{(-, -, \uparrow), \bvec{k}}^{\dagger}  \right) \left( U_{2} \circ \hat{f}_{(+, -, \uparrow), \bvec{k}}^{\dagger} \right) \ket{\mathrm{vac}} : U_{1}, U_{2} \in \mathsf{U}(4) \right\} \\
      \ket{\tilde{\Phi}_{2}} & \in \spn\left\{\prod_{\bvec{k} \in \mathcal{K}}^{} \left( U_{1} \circ \hat{f}_{(+,+, \uparrow), \bvec{k}}^{\dagger} \hat{f}_{(+,+,\downarrow), \bvec{k}}^{\dagger} \right) \left( U_{2} \circ \hat{f}_{(+, -, \uparrow), \bvec{k}}^{\dagger} \hat{f}_{(+, -, \downarrow), \bvec{k}}^{\dagger} \right) \ket{\mathrm{vac}} : U_{1}, U_{2} \in \mathsf{U}(4) \right\} \\
      \ket{\tilde{\Phi}_{3}} & \in \spn\left\{\prod_{\bvec{k} \in \mathcal{K}}^{} \left( U_{1} \circ \hat{f}_{(+,+, \uparrow), \bvec{k}}^{\dagger} \right) \left( U_{2} \circ \hat{f}_{(+, -, \uparrow), \bvec{k}}^{\dagger} \hat{f}_{(+, -, \downarrow), \bvec{k}}^{\dagger} \hat{f}_{(-,+, \uparrow), \bvec{k}}^{\dagger}  \right) \ket{\mathrm{vac}} : U_{1}, U_{2} \in \mathsf{U}(4) \right\}
    \end{split}
  \end{equation}
  Alternatively, the states $\{ \ket{\tilde{\Phi}_{i}} : i \in \{0, \cdots, 4 \} \}$ lie in a vector spaces of the following dimensions:
  \begin{center}
    \renewcommand{\arraystretch}{1.4}
    \begin{tabular}[t]{ll}
      State & Vector Space Dimension \\
      $\ket{\tilde{\Phi}_{0}}$ & 1 \\
      $\ket{\tilde{\Phi}_{1}}$ & $\frac{1}{36} (N_{\bvec{k}} + 1)^{2} (N_{\bvec{k}} + 2)^{2} (N_{\bvec{k}} + 3)^{2}$ \\
      $\ket{\tilde{\Phi}_{2}}$ & $\frac{1}{144} (N_{\bvec{k}} + 1)^{2} (N_{\bvec{k}} + 2)^{4} (N_{\bvec{k}} + 3)^{2}$ \\
      $\ket{\tilde{\Phi}_{3}}$ &  $\frac{1}{36} (N_{\bvec{k}} + 1)^{2} (N_{\bvec{k}} + 2)^{2} (N_{\bvec{k}} + 3)^{2}$ \\ 
      $\ket{\tilde{\Phi}_{4}}$ & 1 \\
    \end{tabular}
  \end{center}
  and the total dimension of the ground space is the sum of these dimensions.
\end{theorem}

\begin{proof}
  The general strategy used to prove~\cref{prop:valleyful-spinful-case} will be the same as~\cref{prop:valleyful-case} in that we will use the decomposition implied by~\cref{lem:uniform-filling} to characterize the non-trivial $+$ chiral fillings $\lambda_{+} = 1, 2, 3$.
  The case $\lambda_{+} = 1$ will follow using the same argument as used for~\cref{prop:valleyful-case} with minimal changes since $\CC^{2}$ isomorphically embeds into a subspace of $\CC^{4}$, and $\CC^{4} \wedge \CC^{4} \wedge \CC^{4} \cong \CC^{4}$.
  While the case $\lambda_{+} = 3$ is the same as $\lambda_{+} = 1$ due to the symmetry between the two chiral sectors, $\lambda_{+} = 2$ requires a new argument entirely. 

  \underline{$\lambda_{+} = 1$ and $\lambda_{+} = 3$}:
  The argument in this case is nearly identical to the one presented in~\cref{sec:valleyful-case} with a few minor changes.
  We begin by considering $\lambda_{+} = 1$ and focusing on the $+$ chiral sector only.
  We construct a decreasing sequence of subspaces $(\CC^{4})^{\otimes N_{\bvec{k}}} = V^{(1)} \supset V^{(2)} \supset \cdots \supset V^{(N_{\bvec{k}})}$ so that a state $\ket{\Phi}$ is a ground state only if its occupation vector lies in $V^{(N_{\bvec{k}})} \otimes (\CC^{4})^{\otimes N_{\bvec{k}}}$.

  We can decompose the occupation vector space on the $+$ chiral sector as follows:
  \begin{equation}
    \begin{array}{rlccrl}
      V^{(1)} & = \CC^{4} \otimes \CC^{4} \otimes (\CC^{4})^{\otimes (N_{\bvec{k}} - 2)} & & & V^{(1)} &  = \ydiagram{1} \otimes \ydiagram{1} \otimes (\CC^{4})^{\otimes (N_{\bvec{k}} - 2)} \\[1ex]
              & = ( \mathbb{S}_{(2)}(\mathbb{C}^{4}) \oplus \mathbb{S}_{(1,1)}(\mathbb{C}^{4}) ) \otimes (\CC^{4})^{\otimes (N_{\bvec{k}} - 2)} & & & & = \left(~\ydiagram{2} \oplus \ydiagram{1,1} ~\right) \otimes (\CC^{4})^{\otimes (N_{\bvec{k}} - 2)}.
    \end{array}
  \end{equation}
  which is precisely the same decomposition as used in~\cref{sec:valleyful-case} with $\CC^{2}$ replaced with $\CC^{4}$.
  In fact, using the Theorem of Highest Weight, $\mathbb{S}_{(2)}(\mathbb{C}^{4})$ and $\mathbb{S}_{(1,1)}(\mathbb{C}^{4})$ are cyclically generated by 
  \begin{equation}
    \begin{split}
      \mathbb{S}_{(2)}(\mathbb{C}^{4}) & := \gen_{\mathsf{U}(4)}\big\{ \ket{1^{+}}^{\otimes 2} \big\}, \\[1ex]
      \mathbb{S}_{(1,1)}(\mathbb{C}^{4}) & := \gen_{\mathsf{U}(4)}\big\{ \ket{1^{+}} \ket{2^{+}} - \ket{2^{+}} \ket{1^{+}} \big\}.
    \end{split}
  \end{equation}
  which is the same as the valleyful case with $\mathsf{U}(2)$ replaced with $\mathsf{U}(4)$ and now $\ket{1}, \ket{2} \in \CC^{4}$ instead of $\CC^{2}$.
  Repeating the same reasoning as used in the valleyful case where we replace $\mathsf{U}(2)$ with $\mathsf{U}(4)$ and $\CC^{2}$ with $\CC^{4}$ where appropriate we can conclude that the occupation vector for every ground state must lie in the vector space $V^{(N_{\bvec{k}})} \otimes (\CC^{4})^{\otimes N_{\bvec{k}}}$ where
  \begin{equation}
    V^{(N_{\bvec{k}})} = \gen_{\mathsf{U}(4)}\big\{ \ket{1^{+}}^{\otimes N_{\bvec{k}}} \big\}.
  \end{equation}

  To fix the form of the ground state on the $-$ chiral sector we introduce a particle hole transformation \cite{Zirnbauer2021}, $\mathcal{P}_{-}$, which is a bijective map which maps between $\CC^{4}$ and $\CC^{4} \wedge \CC^{4} \wedge \CC^{4}$.
  The operator $\mathcal{P}_{-}$ is a linear map which acts on a basis for $\CC^{4}$ as follows:
  \begin{equation}
    \begin{split}
      \mathcal{P}_{-} \ket{1^{-}} & = \ket{2^{-} \wedge 3^{-} \wedge 4^{-}} \qquad
                                    \mathcal{P}_{-} \ket{2^{-}} = \ket{1^{-} \wedge 3^{-} \wedge 4^{-}} \\
      \mathcal{P}_{-} \ket{3^{-}} & = \ket{1^{-} \wedge 2^{-} \wedge 4^{-}} \qquad
                                    \mathcal{P}_{-} \ket{4^{-}} = \ket{1^{-} \wedge 2^{-} \wedge 3^{-}}.
    \end{split}
  \end{equation}
  To maintain the isomorphism between the occupation vector space and the many-body space, we can extend the definition of $\mathcal{P}_{-}$ to the many-body space.
  In particular, we can define $\mathcal{P}_{-}$ to act as
  \begin{equation}
    \begin{split}
      & \mathcal{P}_{-} \hat{f}_{m^{-}, \bvec{k}}^{\dagger} \mathcal{P}_{-}^{-1} = (-1)^{m + 1} \hat{f}_{m^{-}, \bvec{k}} \\
      & \mathcal{P}_{-} \hat{f}_{m^{-}, \bvec{k}} \mathcal{P}_{-}^{-1} = (-1)^{m + 1} \hat{f}_{m^{-}, \bvec{k}}^{\dagger} \\
      & \mathcal{P}_{-} \ket{\mathrm{vac}} = \prod_{\bvec{k} \in \mathcal{K}} \hat{f}_{1^{-}, \bvec{k}}^{\dagger} \hat{f}_{2^{-}, \bvec{k}}^{\dagger} \hat{f}_{3^{-}, \bvec{k}}^{\dagger} \hat{f}_{4^{-}, \bvec{k}}^{\dagger} \ket{\mathrm{vac}}.
    \end{split}
  \end{equation}
  where we have associated the $-$ chiral states $(+, -, \uparrow)$, $(+, -, \downarrow)$, etc. with $\{ 1^{-}, 2^{-}, 3^{-}, 4^{-} \}$ per \cref{rem:basis-notation}.
  One can verify these definitions maintain the isomorphism between the two spaces.

  Importantly, the operator $\hat{C}_{-, \bvec{k}, \bvec{k}'}$ anticommutes with the particle hole symmetry operator:
  \begin{equation}
    \mathcal{P}_{-} \hat{C}_{-, \bvec{k}, \bvec{k}'} \mathcal{P}_{-}^{-1} = - \hat{C}_{-, \bvec{k}, \bvec{k}'} \quad \forall \bvec{k}, \bvec{k}' \in \mathcal{K}.
  \end{equation}
  So $\hat{C}_{-, \bvec{k}, \bvec{k}'} \ket{\Phi_{\rm rest}} = 0$ if and only if $\hat{C}_{-, \bvec{k}, \bvec{k}'} \mathcal{P}_{-} \ket{\Phi_{\rm rest}} = 0$.
  Therefore, following the reasoning used for the $+$ chiral sector, a many-body state is a ground state if and only if its occupation vector lies in $V^{(N_{\bvec{k}})} \otimes (\mathcal{P}_{-} V^{(N_{\bvec{k}})}$) which proves the claim.

  \underline{$\lambda_{+} = 2$}:
  Similar to the proof for $\lambda_{+} = 1$, we use the Littlewood-Richardson rules to decompose the occupation vector space for the $+$ chiral sector $(\CC^{4} \wedge \CC^{4})^{\otimes N_{\bvec{k}}}$ inductively.

  For $V^{(1)}$ we have the decomposition
  \begin{equation}
    \begin{split}
      V^{(1)} & = (\CC^{4} \wedge \CC^{4}) \otimes (\CC^{4} \wedge \CC^{4}) \otimes (\CC^{4} \wedge \CC^{4})^{\otimes (N_{\bvec{k}} - 2)} \\[1ex]
              & = ( \mathbb{S}_{(2,2)}(\mathbb{C}^{4}) \oplus \mathbb{S}_{(2,1,1)}(\mathbb{C}^{4}) \oplus \mathbb{S}_{(1,1,1,1)}(\mathbb{C}^{4}) ) \otimes (\CC^{4} \wedge \CC^{4})^{\otimes (N_{\bvec{k}} - 2)}
    \end{split}
  \end{equation}
  with corresponding Young diagrams
  \begin{equation}
    \begin{split}
      V^{(1)} &  = \ydiagram{1,1} \otimes \ydiagram{1,1} \otimes (\CC^{4} \wedge \CC^{4})^{\otimes (N_{\bvec{k}} - 2)} \\[1ex]
              & = \left(~\ydiagram{2,2} \oplus \ydiagram{2,1,1} \oplus \ydiagram{1,1,1,1} ~\right) \otimes (\CC^{4} \wedge \CC^{4})^{\otimes (N_{\bvec{k}} - 2)}.
    \end{split}
  \end{equation}
  Using the Young symmetrizers and the Theorem of Highest Weight we can verify the following are generating vectors for the irreps:
  \begin{equation}
    \begin{split}
      \mathbb{S}_{(2,2)}(\mathbb{C}^{4}) & = \gen_{\mathsf{U}(4)}\big\{ \ket{1^{+} \wedge 2^{+}}^{\otimes 2} \big\} \\
      \mathbb{S}_{(2,1,1)}(\mathbb{C}^{4}) & = \gen_{\mathsf{U}(4)}\Big\{ \ket{1^{+} \wedge 2^{+}} \ket{1^{+} \wedge 3^{+}} + \ket{1^{+} \wedge 3^{+}} \ket{1^{+} \wedge 2^{+}} \Big\} \\
      \mathbb{S}_{(1,1,1,1)}(\mathbb{C}^{4}) &= \gen_{\mathsf{U}(4)} \Big\{
                                               \sum_{\sigma \in S_{4}}^{} (-1)^{\sigma} \ket{\sigma(1)^{+} \wedge \sigma(2)^{+}} \ket{\sigma(3)^{+} \wedge \sigma(4)^{+}} \Big\}.
    \end{split}
  \end{equation}
  Explicitly, up to a scaling constant, the generator for $\mathbb{S}_{(1,1,1,1)}(\mathbb{C}^{4})$ is a sum of the following six terms:
  \begin{equation}
    \label{eq:lambda-2-generator-occupation}
    \begin{split}
      & \ket{1^{+} \wedge 2^{+}} \ket{3^{+} \wedge 4^{+}} 
        + \ket{2^{+} \wedge 3^{+}} \ket{1^{+} \wedge 4^{+}}
        + \ket{3^{+} \wedge 4^{+}} \ket{1^{+} \wedge 2^{+}} \\
      & - \ket{1^{+} \wedge 3^{+}} \ket{2^{+} \wedge 4^{+}}
        - \ket{2^{+} \wedge 4^{+}} \ket{1^{+} \wedge 3^{+}}
        + \ket{1^{+} \wedge 4^{+}} \ket{2^{+} \wedge 3^{+}}. 
    \end{split}
  \end{equation}
  Mapping these equations back to Fock space, we see that the corresponding many-body states are
  \begin{equation}
    \label{eq:lambda-2-generators-1}
    \begin{split}
      & \hat{f}_{1^{+}, \bvec{k}_{1}}^{\dagger} \hat{f}_{2^{+}, \bvec{k}_{1}}^{\dagger} \hat{f}_{1^{+}, \bvec{k}_{2}}^{\dagger} \hat{f}_{2^{+}, \bvec{k}_{2}}^{\dagger} \ket{\Phi_{\rm rest}} \\
      & \left( \hat{f}_{1^{+}, \bvec{k}_{1}}^{\dagger} \hat{f}_{2^{+}, \bvec{k}_{1}}^{\dagger} \hat{f}_{1^{+}, \bvec{k}_{2}}^{\dagger} \hat{f}_{3^{+}, \bvec{k}_{2}}^{\dagger} - \hat{f}_{1^{+}, \bvec{k}_{1}}^{\dagger} \hat{f}_{3^{+}, \bvec{k}_{1}}^{\dagger} \hat{f}_{1^{+}, \bvec{k}_{2}}^{\dagger} \hat{f}_{2^{+}, \bvec{k}_{2}}^{\dagger}  \right) \ket{\Phi_{\rm rest}} \\
      & \left( \sum_{\sigma \in S_{4}}^{} (-1)^{\sigma}  \hat{f}_{\sigma(1)^{+}, \bvec{k}_{1}}^{\dagger} \hat{f}_{\sigma(2)^{+}, \bvec{k}_{1}}^{\dagger} \hat{f}_{\sigma(3)^{+}, \bvec{k}_{2}}^{\dagger} \hat{f}_{\sigma(4)^{+}, \bvec{k}_{2}}^{\dagger} \right) \ket{\Phi_{\rm rest}}.
    \end{split}
  \end{equation}
  For the generator for $\mathbb{S}_{(2,2)}$, one can easily verify that
  \begin{equation}
    \hat{C}_{+, \bvec{k}_{1}, \bvec{k}_{2}} \left( \hat{f}_{1^{+}, \bvec{k}_{1}}^{\dagger} \hat{f}_{2^{+}, \bvec{k}_{1}}^{\dagger} \hat{f}_{1^{+}, \bvec{k}_{2}}^{\dagger} \hat{f}_{2^{+}, \bvec{k}_{2}}^{\dagger} \right) \ket{\Phi_{\rm rest}} = 0.
  \end{equation}
  Hence the occupation vector for a ground state may lie within $\mathbb{S}_{(2,2)}(\mathbb{C}^{4}) \otimes (\CC^{4} \wedge \CC^{4})^{\otimes (N_{\bvec{k}} - 2)}$.
  For $\mathbb{S}_{(2,1,1)}$ we have that
  \begin{equation}
    \begin{split}
      \hat{C}_{+, \bvec{k}_{1}, \bvec{k}_{2}} \left( \hat{f}_{1^{+}, \bvec{k}_{1}}^{\dagger} \hat{f}_{2^{+}, \bvec{k}_{1}}^{\dagger} \hat{f}_{1^{+}, \bvec{k}_{2}}^{\dagger} \hat{f}_{3^{+}, \bvec{k}_{2}}^{\dagger} - \hat{f}_{1^{+}, \bvec{k}_{1}}^{\dagger} \hat{f}_{3^{+}, \bvec{k}_{1}}^{\dagger} \hat{f}_{1^{+}, \bvec{k}_{2}}^{\dagger} \hat{f}_{2^{+}, \bvec{k}_{2}}^{\dagger}  \right) & \ket{\Phi_{\rm rest}} \\
      = -2 \left( \hat{f}_{1^{+}, \bvec{k}_{1}}^{\dagger} \hat{f}_{2^{+}, \bvec{k}_{1}}^{\dagger} \hat{f}_{3^{+}, \bvec{k}_{1}}^{\dagger}  \hat{f}_{1^{+}, \bvec{k}_{2}}^{\dagger}  \right) & \ket{\Phi_{\rm rest}} \neq 0
    \end{split}
  \end{equation}
  and hence the occupation vector for a ground state cannot lie in $\mathbb{S}_{(2,1,1)}(\mathbb{C}^{4}) \otimes (\CC^{4} \wedge \CC^{4})^{\otimes (N_{\bvec{k}} - 2)}$.

  To rule out the occupation vector of a ground state being an element of $\mathbb{S}_{(1,1,1,1)}(\mathbb{C}^{4}) \otimes (\CC^{4} \wedge \CC^{4})^{\otimes (N_{\bvec{k}} - 2)}$ we can explicitly calculate all of 24 terms after applying $\hat{C}_{+, \bvec{k}_{1}, \bvec{k}_{2}}$ to \cref{eq:lambda-2-generators-1} and show that the result does not vanish, however, there is an easier way.
  Looking at \cref{eq:lambda-2-generator-occupation}, we see that there are only three terms which have $\ket{4^{+}}$ at momenta $\bvec{k}_{2}$.
  Since $\hat{C}_{+, \bvec{k}_{1}, \bvec{k}_{2}}$ moves electrons from $\bvec{k}_{2}$ to $\bvec{k}_{1}$ after the applying $\hat{C}_{+, \bvec{k}_{1}, \bvec{k}_{2}}$, these three terms are the only terms which have components where $\ket{4^{+}}$ is occupied at $\bvec{k}_{2}$
  Using this observation, we can easily check that
  \begin{equation}
    \begin{split}
      & \hat{C}_{+, \bvec{k}_{1}, \bvec{k}_{2}} \left( \sum_{\sigma \in S_{4}}^{} (-1)^{\sigma}  \hat{f}_{\sigma(1)^{+}, \bvec{k}_{1}}^{\dagger} \hat{f}_{\sigma(2)^{+}, \bvec{k}_{1}}^{\dagger} \hat{f}_{\sigma(3)^{+}, \bvec{k}_{2}}^{\dagger} \hat{f}_{\sigma(4)^{+}, \bvec{k}_{2}}^{\dagger} \right) \ket{\Phi_{\rm rest}} \\
      & = -3 \left( \hat{f}_{1^{+}, \bvec{k}_{1}}^{\dagger} \hat{f}_{2^{+}, \bvec{k}_{1}}^{\dagger} \hat{f}_{3^{+}, \bvec{k}_{1}}^{\dagger} \hat{f}_{4^{+}, \bvec{k}_{2}}^{\dagger} \right) \ket{\Phi_{\rm rest}} + \ket{\Phi_{\text{no $\ket{4^{+}}$ at $\bvec{k}_{2}$}}} \ket{\Phi_{\rm rest}}.
    \end{split}
  \end{equation}
  By definition $\ket{\Phi_{\text{no $\ket{4^{+}}$ at $\bvec{k}_{2}$}}} \ket{\Phi_{\rm rest}}$ and $\hat{f}_{1^{+}, \bvec{k}_{1}}^{\dagger} \hat{f}_{2^{+}, \bvec{k}_{1}}^{\dagger} \hat{f}_{3^{+}, \bvec{k}_{1}}^{\dagger} \hat{f}_{4^{+}, \bvec{k}_{2}}^{\dagger} \ket{\Phi_{\rm rest}}$ are orthogonal so this implies that no ground state may lie in $\mathbb{S}_{(1,1,1,1)}(\mathbb{C}^{4}) \otimes (\CC^{4} \wedge \CC^{4})^{\otimes (N_{\bvec{k}} - 2)}$.
  These calculations imply that only \(\mathbb{S}_{(2,2)}\) remains a potential part of the ground state manifold.

  Similar to the proof in~\cref{sec:valleyful-case}, we continue adding more momentum points $\bvec{k}_{2}, \bvec{k}_{3}, \mathbf{k}_{4}, \cdots$ and at each stage only the rectangular Young diagram remains a potential part of the ground state manifold.
  By the Littlewood-Richardson rules, at the $n^{\rm th}$ stage of this process, we decompose the vector space in Young diagrams as $\mathbb{S}_{(n +1, n+1)}(\mathbb{C}^{4}) \oplus \mathbb{S}_{(n + 1, n, 1)}(\mathbb{C}^{4}) \oplus \mathbb{S}_{(n,n,1,1)}(\mathbb{C}^{4})$:
  \begin{equation}
    \overbrace{\ydiagram{7,7}}^{n+1} \oplus \overbrace{\ydiagram{7,6,1}}^{n+1} \oplus \overbrace{\ydiagram{6,6,1,1}}^{n}
  \end{equation}
  By the Theorem of Highest Weight, one can verify the following are generating vectors for the irreps:
  \begin{equation}
    \begin{split}
      \mathbb{S}_{(n+1, n+1)} & = \gen_{\mathsf{U}(4)}\big\{ \ket{1^{+} \wedge 2^{+}}^{\otimes (n+1)} \big\} \\
      \mathbb{S}_{(n+1,n,1)} & = \gen_{\mathsf{U}(4)}\Big\{ \Sym\Big[ \ket{1^{+} \wedge 2^{+}}^{\otimes (n-1)} \ket{1^{+} \wedge 3^{+}} \Big] \ket{1^{+} \wedge 2^{+}} - \ket{1^{+} \wedge 3^{+}} \ket{1^{+} \wedge 2^{+}} \Big\} \\
      \mathbb{S}_{(n,n,1,1)} &= \gen_{\mathsf{U}(4)} \Big\{ \sum_{\sigma \in S_{4}}^{} (-1)^{\sigma} \Sym\Big[ (\ket{1^{+} \wedge 2^{+}})^{\otimes (n - 1)}  \ket{\sigma(1)^{+} \wedge \sigma(2)^{+}} \Big] \ket{\sigma(3)^{+} \wedge \sigma(4)^{+}} \Big\}.
    \end{split}
  \end{equation}
  These generating vectors are constructed by appending on $(n - 1)$ copies of $\ket{1^{+} \wedge 2^{+}}$ to the beginning of the occupation vector for $\mathbb{S}_{(2,2)}(\mathbb{C}^{4})$, $\mathbb{S}_{(2,1,1)}(\mathbb{C}^{4})$, $\mathbb{S}_{(1,1,1,1)}(\mathbb{C}^{4})$ and symmetrizing with respect to the first $n$ tensor indices.
  Following a similar calculation for the $\mathbb{S}_{(2,2)}(\mathbb{C}^{4}) \oplus \mathbb{S}_{(2,1,1)}(\mathbb{C}^{4}) \oplus \mathbb{S}_{(1,1,1,1)}(\mathbb{C}^{4})$ case, we can conclude that only $\mathbb{S}_{(n+1,n+1)}(\mathbb{C}^{4})$ vanishes when acted upon by $\hat{C}_{+, \bvec{k}_{n}, \bvec{k}_{n + 1}}$.
  This implies that every ground state with \(\lambda_{+} = 2\) may be written as $\ket{\tilde{\Phi}_{2}}$.
  Applying the Hook Length formula (\cref{thm:hook-formula}) completes the proof.
\end{proof}

\bibliographystyle{plain}
\bibliography{bibliography}

\begin{thebibliography}{10}

\bibitem{AndreiEfetovJarilloHerreroEtAl2021}
Eva~Y. Andrei, Dmitri~K. Efetov, Pablo Jarillo-Herrero, Allan~H. MacDonald,
  Kin~Fai Mak, T.~Senthil, Emanuel Tutuc, Ali Yazdani, and Andrea~F. Young.
\newblock The marvels of moir\'e materials.
\newblock {\em Nat. Rev. Mater.}, 6(3):201--206, 2021.

\bibitem{AndreiMacDonald2020}
Eva~Y. Andrei and Allan~H. MacDonald.
\newblock Graphene bilayers with a twist.
\newblock {\em Nat. Mater.}, 19(12):1265--1275, 2020.

\bibitem{BeckerEmbreeWittstenEtAl2021}
Simon Becker, Mark Embree, Jens Wittsten, and Maciej Zworski.
\newblock Spectral characterization of magic angles in twisted bilayer
  graphene.
\newblock {\em Phys. Rev. B}, 103(16):165113, 2021.

\bibitem{BeckerEmbreeWittsten2022}
Simon Becker, Mark Embree, Jens Wittsten, and Maciej Zworski.
\newblock Mathematics of magic angles in a model of twisted bilayer graphene.
\newblock {\em Probability and Mathematical Physics}, 3(1):69--103, 2022.

\bibitem{becker2023chiral}
Simon Becker, Tristan Humbert, Jens Wittsten, and Mengxuan Yang.
\newblock Chiral limit of twisted trilayer graphene.
\newblock {\em arXiv:2308.10859}, 2023.

\bibitem{BeckerHumbertZworski2023}
Simon Becker, Tristan Humbert, and Maciej Zworski.
\newblock Degenerate flat bands in twisted bilayer graphene.
\newblock {\em arXiv:2306.02909}, 2023.

\bibitem{BeckerHumbertZworski2022}
Simon {Becker}, Tristan {Humbert}, and Maciej {Zworski}.
\newblock Integrability in the chiral model of magic angles.
\newblock {\em Communications in Mathematical Physics}, 403(2):1153--1169,
  2023.

\bibitem{BeckerHumbertZworski2022a}
Simon Becker, Tristan Humbert, and Maciej Zworski.
\newblock Fine structure of flat bands in a chiral model of magic angles.
\newblock In {\em Annales Henri Poincar{\'e}}, pages 1--31. Springer, 2024.

\bibitem{BeckerLinStubbs2023}
Simon Becker, Lin Lin, and Kevin~D. Stubbs.
\newblock Exact ground state of interacting electrons in magic angle graphene.
\newblock {\em Commun. Math. Phys. in press}, 2025.

\bibitem{BernevigLianCowsik2022}
B~Andrei Bernevig, Biao Lian, Aditya Cowsik, Fang Xie, Nicolas Regnault, and
  Zhi-Da Song.
\newblock {Twisted bilayer graphene. V. Exact analytic many-body excitations in
  Coulomb Hamiltonians: Charge gap, Goldstone modes, and absence of Cooper
  pairing}.
\newblock {\em Phys. Rev. B}, 103(20):205415, 2021.

\bibitem{BernevigSongRegnaultEtAl2021}
B.~Andrei Bernevig, Zhi-Da Song, Nicolas Regnault, and Biao Lian.
\newblock {Twisted bilayer graphene. III. Interacting Hamiltonian and exact
  symmetries}.
\newblock {\em Phys. Rev. B}, 103(20):205413, 2021.

\bibitem{BultinckKhalafLiuEtAl2020}
Nick Bultinck, Eslam Khalaf, Shang Liu, Shubhayu Chatterjee, Ashvin Vishwanath,
  and Michael~P. Zaletel.
\newblock Ground {{State}} and {{Hidden Symmetry}} of {{Magic-Angle Graphene}}
  at {{Even Integer Filling}}.
\newblock {\em Phys. Rev. X}, 10(3):031034, 2020.

\bibitem{CaoFatemiDemir2018}
Yuan Cao, Valla Fatemi, Ahmet Demir, Shiang Fang, Spencer~L. Tomarken, Jason~Y.
  Luo, Javier~D. Sanchez-Yamagishi, Kenji Watanabe, Takashi Taniguchi,
  Efthimios Kaxiras, Ray~C. Ashoori, and Pablo Jarillo-Herrero.
\newblock Correlated insulator behaviour at half-filling in magic-angle
  graphene superlattices.
\newblock {\em Nature}, 556(7699):80--84, 2018.

\bibitem{2018Nature}
Yuan {Cao}, Valla {Fatemi}, Shiang {Fang}, Kenji {Watanabe}, Takashi
  {Taniguchi}, Efthimios {Kaxiras}, and Pablo {Jarillo-Herrero}.
\newblock {Unconventional superconductivity in magic-angle graphene
  superlattices}.
\newblock {\em Nature}, 556(7699):43--50, 2018.

\bibitem{ChatterjeeBultinckZaletel2020}
Shubhayu Chatterjee, Nick Bultinck, and Michael~P. Zaletel.
\newblock Symmetry breaking and skyrmionic transport in twisted bilayer
  graphene.
\newblock {\em Phys. Rev. B}, 101(16), 2020.

\bibitem{DasLuHerzog-Arbeitman2021}
Ipsita Das, Xiaobo Lu, Jonah Herzog-Arbeitman, Zhi-Da Song, Kenji Watanabe,
  Takashi Taniguchi, B.~Andrei Bernevig, and Dmitri~K. Efetov.
\newblock Symmetry-broken chern insulators and rashba-like landau-level
  crossings in magic-angle bilayer graphene.
\newblock {\em Nat. Phys.}, 17(6):710--714, 2021.

\bibitem{FaulstichStubbsZhuEtAl2023}
Fabian~M Faulstich, Kevin~D Stubbs, Qinyi Zhu, Tomohiro Soejima, Rohit Dilip,
  Huanchen Zhai, Raehyun Kim, Michael~P Zaletel, Garnet Kin-Lic Chan, and Lin
  Lin.
\newblock Interacting models for twisted bilayer graphene: A quantum chemistry
  approach.
\newblock {\em Phys. Rev. B}, 107(23):235123, 2023.

\bibitem{FultonHarris2004}
William Fulton and Joe Harris.
\newblock {\em Representation {{Theory}}}, volume 129 of {\em Graduate
  {{Texts}} in {{Mathematics}}}.
\newblock Springer New York, 2004.

\bibitem{Hall2015}
Brian~C. Hall.
\newblock {\em Lie {{Groups}}, {{Lie Algebras}}, and {{Representations}}: {{An
  Elementary Introduction}}}, volume 222 of {\em Graduate {{Texts}} in
  {{Mathematics}}}.
\newblock Springer International Publishing, 2015.

\bibitem{HuertasHernandoGuineaBrataas2006}
D.~Huertas-Hernando, F.~Guinea, and A.~Brataas.
\newblock Spin-orbit coupling in curved graphene, fullerenes, nanotubes, and
  nanotube caps.
\newblock {\em Phys. Rev. B}, 74(15):155426, 2006.

\bibitem{JiangLaiWatanabe2019}
Yuhang Jiang, Xinyuan Lai, Kenji Watanabe, Takashi Taniguchi, Kristjan Haule,
  Jinhai Mao, and Eva~Y. Andrei.
\newblock Charge order and broken rotational symmetry in magic-angle twisted
  bilayer graphene.
\newblock {\em Nature}, 573(7772):91--95, 2019.

\bibitem{KangVafek2019}
Jian Kang and Oskar Vafek.
\newblock Strong {{Coupling Phases}} of {{Partially Filled Twisted Bilayer
  Graphene Narrow Bands}}.
\newblock {\em Phys. Rev. Lett.}, 122(24):246401, 2019.

\bibitem{KwanWagnerBultinck2021}
Yves~H Kwan, Glenn Wagner, Nick Bultinck, Steven~H Simon, and SA~Parameswaran.
\newblock Skyrmions in twisted bilayer graphene: stability, pairing, and
  crystallization.
\newblock {\em Phys. Rev. X}, 12(3):031020, 2022.

\bibitem{LedwithTarnpolskyKhalaf2020}
Patrick~J Ledwith, Grigory Tarnopolsky, Eslam Khalaf, and Ashvin Vishwanath.
\newblock Fractional chern insulator states in twisted bilayer graphene: An
  analytical approach.
\newblock {\em Phys. Rev. Res.}, 2(2):023237, 2020.

\bibitem{LianSongRegnaultEtAl2021}
Biao Lian, Zhi-Da Song, Nicolas Regnault, Dmitri~K. Efetov, Ali Yazdani, and
  B.~Andrei Bernevig.
\newblock {Twisted bilayer graphene. IV. Exact insulator ground states and
  phase diagram}.
\newblock {\em Phys. Rev. B}, 103(20):205414, 2021.

\bibitem{LiuKhalafLee2021}
Shang Liu, Eslam Khalaf, Jong~Yeon Lee, and Ashvin Vishwanath.
\newblock Nematic topological semimetal and insulator in magic-angle bilayer
  graphene at charge neutrality.
\newblock {\em Phys. Rev. Res.}, 3(1), 2021.

\bibitem{LiuKhalafLeeEtAl2021}
Shang Liu, Eslam Khalaf, Jong~Yeon Lee, and Ashvin Vishwanath.
\newblock Nematic topological semimetal and insulator in magic-angle bilayer
  graphene at charge neutrality.
\newblock {\em Phys. Rev. Res.}, 3(1):013033, 2021.

\bibitem{LuStepanovYang2019}
Xiaobo Lu, Petr Stepanov, Wei Yang, Ming Xie, Mohammed~Ali Aamir, Ipsita Das,
  Carles Urgell, Kenji Watanabe, Takashi Taniguchi, Guangyu Zhang, Adrian
  Bachtold, Allan~H. MacDonald, and Dmitri~K. Efetov.
\newblock Superconductors, orbital magnets and correlated states in magic-angle
  bilayer graphene.
\newblock {\em Nature}, 574(7780):653--657, 2019.

\bibitem{MeraOzawa2024}
Bruno Mera and Tomoki Ozawa.
\newblock Uniqueness of {{Landau}} levels and their analogs with higher
  {{Chern}} numbers.
\newblock {\em Phys. Rev. Research}, 6(3):033238, 2024.

\bibitem{MumfordMusili2007}
David Mumford and C.~Musili.
\newblock {\em Tata Lectures on {{Theta}}}.
\newblock Modern {{Birkh\"auser}} Classics. {Birkh\"auser}, 2007.

\bibitem{PotaszXieMacDonald2021}
Pawel Potasz, Ming Xie, and Allan~H. MacDonald.
\newblock Exact {{Diagonalization}} for {{Magic-Angle Twisted Bilayer
  Graphene}}.
\newblock {\em Phys. Rev. Lett.}, 127(14):147203, 2021.

\bibitem{Ragone2024}
Michael Ragone.
\newblock {{SO}}(n) {{AKLT Chains}} as {{Symmetry Protected Topological Quantum
  Ground States}}, 2024.

\bibitem{SaitoGeRademaker2021}
Yu~Saito, Jingyuan Ge, Louk Rademaker, Kenji Watanabe, Takashi Taniguchi,
  Dmitry~A. Abanin, and Andrea~F. Young.
\newblock Hofstadter subband ferromagnetism and symmetry-broken chern
  insulators in twisted bilayer graphene.
\newblock {\em Nat. Phys.}, 17(4):478--481, 2021.

\bibitem{Serlin}
M.~{Serlin}, C.~L. {Tschirhart}, H.~{Polshyn}, Y.~{Zhang}, J.~{Zhu},
  K.~{Watanabe}, T.~{Taniguchi}, L.~{Balents}, and A.~F. {Young}.
\newblock {Intrinsic quantized anomalous Hall effect in a moir{\'e}
  heterostructure}.
\newblock {\em Science}, 367(6480):900--903, 2020.

\bibitem{SoejimaParkerBultinckEtAl2020}
Tomohiro Soejima, Daniel~E. Parker, Nick Bultinck, Johannes Hauschild, and
  Michael~P. Zaletel.
\newblock Efficient simulation of moire materials using the density matrix
  renormalization group.
\newblock {\em Phys. Rev. B}, 102(20):205111, 2020.

\bibitem{StubbsBeckerLin2024}
Kevin~D. Stubbs, Simon Becker, and Lin Lin.
\newblock On the {{Hartree-Fock}} ground state manifold in magic angle twisted
  graphene systems.
\newblock {\em arXiv:2403.19890}, 2024.

\bibitem{TarnopolskyKruchkovVishwanath2019}
Grigory Tarnopolsky, Alex~Jura Kruchkov, and Ashvin Vishwanath.
\newblock Origin of magic angles in twisted bilayer graphene.
\newblock {\em Phys. Rev. Let.}, 122(10):106405, 2019.

\bibitem{WatsonLuskin2021}
Alexander~B Watson and Mitchell Luskin.
\newblock Existence of the first magic angle for the chiral model of bilayer
  graphene.
\newblock {\em J. Math. Phys.}, 62(9):091502, 2021.

\bibitem{WuSarma2020}
Fengcheng Wu and Sankar~Das Sarma.
\newblock Collective excitations of quantum anomalous hall ferromagnets in
  twisted bilayer graphene.
\newblock {\em Phys. Rev. Lett.}, 124(4), 2020.

\bibitem{XieMacDonald2020}
Ming Xie and A.~H. MacDonald.
\newblock Nature of the {{Correlated Insulator States}} in {{Twisted Bilayer
  Graphene}}.
\newblock {\em Phys. Rev. Lett.}, 124(9):097601, 2020.

\bibitem{YankowitzChenPolshyn2019}
Matthew Yankowitz, Shaowen Chen, Hryhoriy Polshyn, Yuxuan Zhang, K.~Watanabe,
  T.~Taniguchi, David Graf, Andrea~F. Young, and Cory~R. Dean.
\newblock Tuning superconductivity in twisted bilayer graphene.
\newblock {\em Science}, 363(6431):1059--1064, 2019.

\bibitem{Zirnbauer2021}
Martin~R. Zirnbauer.
\newblock Particle-{{Hole Symmetries}} in {{Condensed Matter}}.
\newblock {\em J. Math. Phys.}, 62(2):021101, 2021.

\end{thebibliography}

\appendix

\section{Proof of~\cref{lem:linear-indep}}
\label{sec:lemma-linear-indep}

To prove the lemma, we first notice that for all $\bvec{G}' \in \Gamma^{*}$
\begin{equation}
  \label{eq:full-rank-appendix-1}
  \begin{split}
   & \sum_{\bvec{k} \in \mathcal{K}}^{} a_{\bvec{k}}(\bvec{G}) c_{\bvec{k}} = 0 \qquad \forall \bvec{G} \in \Gamma^{*} \\
   & \Longleftrightarrow \sum_{\bvec{G}' \in \Gamma^{*}}^{} e^{-i \bvec{G}' \cdot \bvec{r}} \sum_{\bvec{k} \in \mathcal{K}}^{} a_{\bvec{k}}(\bvec{G}) c_{\bvec{k}} = 0
  \end{split}
\end{equation}
Since
\begin{equation}
  a_{\bvec{k}}(\bvec{G}) = [\Lambda_{\bvec{k}}(\bvec{G})]_{1,1} = \int_{\RR^{2}} e^{i \bvec{G} \cdot \bvec{r}} \braket{u_{1, \bvec{k}}(\bvec{r}), u_{1, \bvec{k}}(\bvec{r})}_{\CC^{4}} \dd{\bvec{r}} = \int_{\RR^{2}} e^{i \bvec{G} \cdot \bvec{r}} \| u_{1, \bvec{k}}(\bvec{r}) \|^{2} \dd{\bvec{r}}.
\end{equation}
Since $u_{\bvec{k}}(\bvec{r})$ is $C^{\infty}$ as a function of $\bvec{r}$~\cite{BeckerHumbertZworski2022a}, by the Fourier inversion theorem, we conclude that~\cref{eq:full-rank-appendix-1} holds only if
\begin{equation}
  \sum_{\bvec{k} \in \mathcal{K}}^{} \| u_{1,\bvec{k}}(\bvec{r}) \|^{2} c_{\bvec{k}} = 0 \qquad \forall \bvec{r} \in \Omega.
\end{equation}
Since we have an explicit formula for $u_{1, \bvec{k}}(\bvec{r})$ in terms of Jacobi theta functions in \cref{eq:jacobi-def}, we can prove the lemma by an explicit calculation.
Since $u_{1, \bvec{k}}(\bvec{r})$ is the only flat band eigenfunction in this calculation, we will suppress the band index and simply write $u_{\bvec{k}}(\bvec{r})$.

\subsection{Notational Setup}
We follow the conventions of \cite{BeckerHumbertZworski2022a} and use the complexified notation for variables in $\RR^2$.
Specifically, we identify the vectors $\bvec{r} = (x, y)$ and $\bvec{k} = (k_{x}, k_{y})$ with the complex numbers $z := x + iy$ and $k = k_{x} + i k_{y}$.
In this notation, bold face variables are used to denote vectors in $\RR^{2}$, and non-bold face variables denote complex scalars.

Defining $\omega = e^{2 \pi i / 3}$ we can write the moir{\'e} lattice, $\Gamma$, and the moir{\'e} reciprocal lattice, $\Gamma^{*}$
\begin{equation}
  \Gamma := \omega \ZZ \oplus \ZZ \qquad \Gamma^{*} := \frac{4 \pi i}{\sqrt{3}} \Gamma
\end{equation}
We also recall the definition of the (shifted) Jacobi theta functions which are defined for any $\tau \in \CC$ with $\Im{(\tau)} > 0$ 
\begin{equation}
  \label{eq:jacobi-def}
  \begin{split}
    \theta(z | \tau)
    & = - \sum_{n \in \ZZ}^{}  e^{i \pi \tau (n + \frac{1}{2})^{2}} e^{2 \pi i (n  + \frac{1}{2}) (z  +\frac{1}{2})} \\
    & = - \sum_{n \in \ZZ}^{} q^{(n + \frac{1}{2})^{2}} e^{2 \pi i (n  + \frac{1}{2}) (z  +\frac{1}{2})} \\
  \end{split}
\end{equation}
where we define $q := e^{i \pi \tau}$ to compress notation.
Note that since $\Im{(\tau)} > 0$, the sum in~\cref{eq:jacobi-def} is absolutely convergent.

The following lemma is an immediate consequence of the Vandermonde determinant formula, and will be useful in the final steps of the proof:
\begin{lemma}
  \label{lem:vandermonde}
  Let $\{ \alpha_{n} \}_{n=1}^{N} \subseteq [0, 1)$ be a collection distinct scalars and let $d \in \CC$ be a constant.
  Then the $N \times N$ matrix $[A]_{mn} = e^{2 \pi i m ( \alpha_{n} - d)}$ where $m, n \in \{ 1, \cdots, N \}$ is full rank.
\end{lemma}

\subsection{Proof of~\cref{lem:linear-indep}}
Since $\mathcal{K} \subseteq \CC / \Gamma^{*}$, we know that each $k \in \mathcal{K}$ may be written as $k = \frac{4 \pi i}{\sqrt{3}} (\omega \alpha + \beta)$ for some $\alpha, \beta \in [0, 1)$.
We begin by collecting the unique $\alpha$ coordinates and $\beta$ coordinates in the sets $\mathcal{K}_{\alpha}$ and $\mathcal{K}_{\beta}$ respectively.
\begin{equation}
  \begin{split}
    \mathcal{K}_{\alpha} & := \{ \alpha : k = \frac{4 \pi i}{\sqrt{3}} (\omega \alpha + \beta) \text{ for some } k \in \mathcal{K} \} \\
    \mathcal{K}_{\beta} & := \{ \beta : k = \frac{4 \pi i}{\sqrt{3}} (\omega \alpha + \beta) \text{ for some } k \in \mathcal{K} \}
  \end{split}
\end{equation}
Note that both sets are finite since $\mathcal{K}$ is finite.

We define the enlarged set of momenta
\begin{equation}
  \label{eq:k-tilde-def}
  \tilde{\mathcal{K}} := \left\{ \frac{4 \pi i}{\sqrt{3}} (\alpha \omega + \beta) : \alpha \in \mathcal{K}_{\alpha}, \beta \in \mathcal{K}_{\beta} \right\}.
\end{equation}
Note that this enlarged set now satisfies the property that 
\begin{equation}
  \sum_{k \in \tilde{\mathcal{K}}}^{} f(k) = \sum_{\alpha \in \mathcal{K}_{\alpha}}^{} \sum_{\beta \in \mathcal{K}_{\beta}}^{} f\left( \frac{4 \pi i}{\sqrt{3}} (\omega \alpha + \beta)\right),
\end{equation}
which will allow us to separate the sum over $k$ in a later stage in our calculations.
Since $\mathcal{K} \subseteq \tilde{\mathcal{K}}$ and we are interested in whether a set of functions are linearly independent, we may assume without loss of generality that $\mathcal{K} = \tilde{\mathcal{K}}$.

The flat band eigenfunction, $u_{\bvec{k}}(\bvec{r}) \equiv u_{k}(z)$, can be written in terms of the Jacobi theta functions as follows
\begin{equation}
  u_{k}(z) = e^{\Im{(z)} k} \frac{\theta(z + z(k) | \omega)}{\theta(z | \omega)} u_{0}(z) \quad \text{where} \quad z(k) := \frac{\sqrt{3}}{4 \pi i} k.
\end{equation}
Now towards a contradiction, suppose the functions $\{ \| u_{k}(z) \|^{2} : k \in \mathcal{K} \}$ were linearly dependent, then there exists a family of constants $\{ c_{k} \in \CC : k \in \mathcal{K} \}$ not all zero so that for all $z \in \Omega$
\begin{equation}
  \begin{split}
    0 = & \sum_{k \in \mathcal{K}}^{} c_{k} \| u_{k}(z) \|^{2}  \\
    = & \sum_{k \in \mathcal{K}}^{} c_{k} e^{2 \Im{(z)} \Re{(k)}} \frac{|\theta(z + z(k) | \omega)|^{2}}{|\theta(z | \omega)|^{2}} \|u_{0}(z)\|^{2} \\
    = & \frac{\|u_{0}(z) \|^{2}}{|\theta(z | \omega)|^{2}} \sum_{k \in \mathcal{K}}^{} c_{k} e^{2 \Im{(z)} \Re{(k)}} |\theta(z + z(k) | \omega)|^{2}.
  \end{split}
\end{equation}
Since $u_{0}(z)$ and $\theta(z)$ vanish at $z = 0$, which is the only node in $\Omega$, we have $\frac{\| u_{0}(z) \|^{2}}{|\theta(z)|^{2}} \neq 0$ for all $z \in \Omega$~\cite{BeckerHumbertZworski2022a}.
Therefore, the functions $\{ \| u_{k}(z) \|^{2} : k \in \mathcal{K} \}$ are linearly dependent only if there exist $\{ c_{k} \}$ not all zero so that
\begin{equation}
  \sum_{k \in \mathcal{K}}^{} c_{k} e^{2 \Im{(z)} \Re{(k)}} |\theta(z + z(k) | \omega)|^{2} = 0 \qquad \forall z \in \CC
\end{equation}
A straightforward calculation (see~\cref{sec:proof-theta-squared}) shows that for all $z \in \CC$ and $\omega$ so that $\Im{(\omega)} > 0$, we can write the norm of the Jacobi theta function as
\begin{equation}
  \label{eq:theta-squared}
  | \theta(z | \omega) |^{2}  
  = - \sum_{m \in \ZZ}^{} \theta\Big( 2 i \Im{(z)} + \omega m - \frac{1}{2} | 2 i \Im{(\omega)} \Big) q^{m^{2}}  e^{2 \pi i m (z + \frac{1}{2})} 
\end{equation}
where $q = e^{i \pi \omega}$.
Therefore,
\begin{equation}
  \begin{split}
    0 &= \sum_{k \in \mathcal{K}}^{} c_{k} e^{2 \Im{(z)} \Re{(k)}} |\theta(z + z(k) | \omega)|^{2} \\
      & = \sum_{k \in \mathcal{K}}^{} c_{k} e^{2 \Im{(z)} \Re{(k)}} \left( \sum_{m \in \ZZ}^{} \theta\Big( 2 i \Im{(z + z(k))} + \omega m - \frac{1}{2} | 2 i \Im{(\omega)} \Big) q^{m^{2}}  e^{2 \pi i m (z + z(k) + \frac{1}{2})} \right) \\
      & = \sum_{m \in \ZZ}^{} \left( \sum_{k \in \mathcal{K}}^{} c_{k} e^{2 \Im{(z)} \Re{(k)}}  \theta\Big( 2 i \Im{(z + z(k))} + \omega m - \frac{1}{2} | 2 i \Im{(\omega)} \Big) e^{2 \pi i m z(k)} \right) q^{m^{2}}  e^{2 \pi i m (z + \frac{1}{2})}
  \end{split}
\end{equation}
Separating $z$ into its real and imaginary parts $z = x + i y$, we notice that the terms inside the sum over $k$ are independent of $x = \Re{(z)}$ and the only instance of $x$ appears in the complex exponential $e^{2 \pi i m x}$.
Therefore, multiplying both sides by $e^{-2 \pi i \ell x}$ and integrating $x$ from $0$ to $1$ lets us conclude for all $\ell \in \ZZ$ it must be that
\begin{equation}
  \label{eq:sum-k-1}
  0 = \left( \sum_{k \in \mathcal{K}}^{} c_{k} e^{2 \Im{(z)} \Re{(k)}}  \theta\Big( 2 i \Im{(z + z(k))} + \omega \ell- \frac{1}{2} | 2 i \Im{(\omega)} \Big) e^{2 \pi i \ell z(k)} \right) q^{\ell^{2}}  e^{2 \pi i \ell(iy + \frac{1}{2})}.
\end{equation}
Note that the terms outside the sum are non-zero and so the above holds if and only if for all $\ell \in \ZZ$
\begin{equation}
  0 = \sum_{k \in \mathcal{K}}^{} c_{k} e^{2 \Im{(z)} \Re{(k)}}  \theta\Big( 2 i \Im{(z + z(k))} + \omega \ell- \frac{1}{2} | 2 i \Im{(\omega)} \Big) e^{2 \pi i \ell z(k)}.
\end{equation}
We now replace the sum over $k \in \mathcal{K}$ to a sum over $\alpha \in \mathcal{K}_{\alpha}$ and $\beta \in \mathcal{K}_{\beta}$ by setting $k = \frac{4\pi i}{\sqrt{3}} (\omega \alpha + \beta)$ and making the substitutions
\begin{equation}
  \begin{split}
    \Re{(k)} & = -2 \pi \alpha, \\
    \Re{( z(k) )} & = \Re{(\omega)} \alpha + \beta, \\
    \Im{( z(k) )} & = \Im{(\omega)} \alpha
  \end{split}
\end{equation}
With these substitutions~\cref{eq:sum-k-1} becomes
\begin{equation}
  \begin{split}
    0 & = \sum_{\alpha \in \mathcal{K}_{\alpha}} \sum_{\beta \in \mathcal{K}_{\beta}}^{}  c_{\alpha,\beta} e^{-4 \pi \alpha \Im{(z)} }  \theta\Big( 2 i \Im{(z)} + 2 i \Im{(\omega)} \alpha + \omega \ell- \frac{1}{2} | 2 i \Im{(\omega)} \Big) e^{2 \pi i \omega \alpha} e^{2 \pi i \ell \beta} \\
      & = \sum_{\alpha \in \mathcal{K}_{\alpha}}^{} e^{-4 \pi \alpha \Im{(z)} } \theta\Big( 2 i \Im{(z)} + 2 i \Im{(\omega)} \alpha + \omega \ell- \frac{1}{2} | 2 i \Im{(\omega)} \Big) e^{2 \pi i \omega \alpha}  \left(  \sum_{\beta \in \mathcal{K}_{\beta}} e^{2 \pi i \ell \beta} c_{\alpha,\beta} \right).
  \end{split}
\end{equation}
Now for any $\ell \in \ZZ$ and $\alpha \in \mathcal{K}_{\alpha}$ define
\begin{equation}
  \hat{c}_{\alpha, \ell} := e^{2 \pi i \omega \alpha}  \sum_{\beta \in \mathcal{K}_{\beta}} e^{2 \pi i \ell \beta} c_{\alpha,\beta}.
\end{equation}
Since the matrix $A_{\ell, \beta} = e^{2 \pi i \ell \beta}$ (where $\ell \in \{ 1, \cdots, \# |\mathcal{K}_{\beta} | \}$) is of full rank, by~\cref{lem:vandermonde}, the mapping $c_{\alpha, \beta} \mapsto \hat{c}_{\alpha, \ell}$ is invertible.

Hence, to prove the proposition it suffices to show that for all $\ell \in \{ 1, \cdots, \# |\mathcal{K}_{\beta}| \}$
\begin{equation}
  \label{eq:sum-k-2}
  0 = \sum_{\alpha \in \mathcal{K}_{\alpha}}^{} e^{-4 \pi \alpha \Im{(z)} } \theta\Big( 2 i \Im{(z)} + 2 i \Im{(\omega)} \alpha + \omega \ell- \frac{1}{2} | 2 i \Im{(\omega)} \Big) \hat{c}_{\alpha,\ell} \quad \iff \quad \hat{c}_{\alpha, \ell} = 0. 
\end{equation}
Fix some $\ell_{*} \in \{ 1, \cdots, \# | \mathcal{K}_{\beta} | \}$ and suppose that $\hat{c}_{\alpha, \ell_{*}}$ is non-zero for at least one $\alpha$.
Consider the function
\begin{equation}
  z \mapsto \sum_{\alpha \in \mathcal{K}_{\alpha}}^{} e^{4 \pi i \alpha z} \theta\Big( 2z  + 2 i \Im{(\omega)} \alpha + \omega \ell_{*}- \frac{1}{2} | 2 i \Im{(\omega)} \Big) \hat{c}_{\alpha,\ell_{*}}.
\end{equation}
Since the set $\mathcal{K}_{\alpha}$ is finite, the above function is holomorphic on $\CC$ and agrees with the function in~\cref{eq:sum-k-2} on the set $\{ z \in \CC : \Re{(z)} = 0 \}$ (i.e. the imaginary axis).
Since the imaginary axis contains an accumulation point, if the left hand side of~\cref{eq:sum-k-2} holds, then by analytic continuation it must be that for all $z \in \CC$
\begin{equation}
  0 =  \sum_{\alpha \in \mathcal{K}_{\alpha}}^{} e^{4 \pi i \alpha z} \theta\Big( 2z  + 2 i \Im{(\omega)} \alpha + \omega \ell_{*}- \frac{1}{2} | 2 i \Im{(\omega)} \Big) \hat{c}_{\alpha,\ell_{*}}.
\end{equation}
However, as the following lemma shows, this implies that $\hat{c}_{\alpha, \ell_{*}} = 0$ for all $\alpha \in \mathcal{K}_{\alpha}$ completing the proof. 
\begin{lemma}
  Let $\mathcal{K}_{\alpha}$ be any finite set of points so that $\mathcal{K}_{\alpha} \subseteq [0, 1)$, let $\tau, d \in \CC$, so that $\Im{(\tau)} > 0$.
  Then the functions
  \begin{equation}
    \left\{ e^{2 \pi i \alpha z }\theta( z + \tau \alpha + d | \tau ) : \alpha \in \mathcal{K}_{\alpha} \right\}
  \end{equation}
  are linearly independent.
\end{lemma}
\begin{proof}
  Standard results on the Jacobi theta function $\theta(z | \tau)$ show that the zeros only occur at the discrete points $z = m \tau + n$ where $m, n \in \ZZ$ \cite{MumfordMusili2007}.
  Therefore, since $\alpha \in [0, 1)$, we may choose $z_{*}$ so that $\theta(z_{*} + \tau \alpha + d)$ does not vanish for all $\alpha$.
  
  Towards a contradiction, suppose that there existed constants $\{ c_{\alpha} \}_{\alpha \in \mathcal{K}_{\alpha}}$ not all zero so that for all $z \in \CC$:
  \begin{equation}
    \sum_{\alpha \in \mathcal{K}_{\alpha}}^{} c_{\alpha} e^{2 \pi i \alpha z}\theta( z + \tau \alpha + d | \tau ) = 0.
  \end{equation}
  But recall that $\theta(z + 1 | \tau) = (-1) \theta(z | \tau)$ so it must also be that for all $z$
  \begin{equation}
    \begin{split}
      0 & = \sum_{\alpha \in \mathcal{K}_{\alpha}}^{} c_{\alpha} e^{2 \pi i \alpha (z + 1)} \theta( z + 1 + \tau \alpha + d | \tau ) \\
        & = \sum_{\alpha \in \mathcal{K}_{\alpha}}^{} c_{\alpha} e^{2 \pi i (\alpha + \frac{1}{2}) } e^{2 \pi i \alpha z} \theta( z + \tau \alpha + d | \tau ). 
    \end{split}
  \end{equation}
  Repeating this shift $m$ times, we see that it must be for all $m \in \ZZ$
  \begin{equation}
    \label{eq:shift-jacobi}
    0 = \sum_{\alpha \in \mathcal{K}_{\alpha}}^{} c_{\alpha} e^{2 \pi i m (\alpha + \frac{1}{2}) } e^{2 \pi i \alpha z} \theta( z + \tau \alpha + d | \tau ).
  \end{equation}
  Plugging in $z = z_{*}$ to the above equalities, we can rewrite collect constraints from \cref{eq:shift-jacobi} for $m \in \{ 1, \cdots, \# |\mathcal{K}_{\alpha}| \}$ into a matrix equation as $A D c = 0$ where $A$ is the Vandermonde matrix $[A]_{m\alpha} = e^{2 \pi i m (\alpha + \frac{1}{2}) }$, $D$ is the diagonal matrix $[D]_{\alpha\beta} = \delta_{\alpha\beta}  e^{2 \pi i \alpha z_{*}} \theta( z_{*} + \tau \alpha + d | \tau )$, and $c$ is a vectorization of $c_{\alpha}$.
  Since $A$ is full rank by~\cref{lem:vandermonde} and the diagonal entries of $D$ are never zero by the choice of $z_{*}$, we conclude that $c = 0$ and hence the lemma is proved.
\end{proof}

\subsection{Proof of~\Cref{eq:theta-squared}}
\label{sec:proof-theta-squared}
We begin by using the definition of the Jacobi theta function~\cref{eq:jacobi-def} recalling that $\overline{e^{z}} = e^{\overline{z}}$
\begin{equation}
  \begin{split}
    |\theta(z | \omega)|
    & = \overline{\theta(z | \omega)} \theta(z | \omega) \\
    & = \left( \sum_{n}^{} \overline{q}^{(n + \frac{1}{2})^{2}} e^{-2 \pi i (n + \frac{1}{2}) (\overline{z} + \frac{1}{2})}  \right) \left( \sum_{m}^{} q^{(m + \frac{1}{2})^{2}} e^{2 \pi i (m + \frac{1}{2}) (z + \frac{1}{2})} \right) \\[1ex]
    & = \sum_{n,m}^{} \overline{q}^{(n + \frac{1}{2})^{2}} q^{(m + \frac{1}{2})^{2}} \exp\left(2 \pi i\left[  (m + \frac{1}{2}) (z + \frac{1}{2}) - (n + \frac{1}{2}) (\overline{z} + \frac{1}{2}) \right] \right).
  \end{split}
\end{equation}
Splitting $z$ and $\overline{z}$ into their real and imaginary parts, we see that
\begin{equation}
  \begin{split}
    (m & + \frac{1}{2}) (z + \frac{1}{2}) - (n  + \frac{1}{2}) (\overline{z} + \frac{1}{2}) \\[1ex]
       & = (m - n) (\Re{(z)} + \frac{1}{2}) + i (m + n + 1) \Im{(z)} \\[1ex]
       & = (m - n) (z + \frac{1}{2}) + i ( 2 n + 1) \Im{(z)}. 
  \end{split}
\end{equation}
So we conclude that
\begin{equation}
  \begin{split}
    | \theta(z | \omega) |^{2} 
    & = \sum_{n, m \in \ZZ}^{} \overline{q}^{(n + \frac{1}{2})^{2}} q^{(m + \frac{1}{2})^{2}} \exp\left(2 \pi i [ (m - n) (z + \frac{1}{2}) + i (2 n + 1) \Im{(z)} ] \right) \\
    & = \sum_{n, m \in \ZZ}^{} \overline{q}^{(n + \frac{1}{2})^{2}} q^{(m + n + \frac{1}{2})^{2}} \exp\left(2 \pi i \left[ m (z + \frac{1}{2}) + 2 i (n + \frac{1}{2}) \Im{(z)} \right] \right) \\
  \end{split}
\end{equation}
where we have performed a change of variables $m \mapsto m + n$.
Now we can combine the two $q$ terms to get
\begin{equation}
  \begin{split}
    \overline{q}^{(n + \frac{1}{2})^{2}} q^{(m + n + \frac{1}{2})^{2}}
    & = \overline{q}^{(n + \frac{1}{2})^{2}} q^{(n + \frac{1}{2})^{2} + 2 (n + \frac{1}{2}) m + m^{2}} \\
    & = |q|^{2 (n + \frac{1}{2})^{2}}  q^{2 (n + \frac{1}{2}) m + m^{2}}.
  \end{split}
\end{equation}
Hence
\begin{equation}
  \label{eq:jacobi-theta-simplify-1}
  | \theta(z | \omega) |^{2} 
  = \sum_{n, m \in \ZZ}^{} |q|^{2 (n + \frac{1}{2})^{2}}  q^{2 (n + \frac{1}{2}) m + m^{2}} \exp\left( 2 \pi i \left[ m ( z + \frac{1}{2} ) + 2 i ( n + \frac{1}{2}) \Im{(z)} \right]  \right).
\end{equation}
Collecting all the terms depending on $n$ together gives
\begin{equation}
  | \theta(z | \omega) |^{2}  
  = \sum_{m \in \ZZ}^{} \left( \sum_{n \in \ZZ}^{}  |q|^{2 (n + \frac{1}{2})^{2}}  q^{2 (n + \frac{1}{2}) m} e^{2 \pi i [ 2 i (n + \frac{1}{2}) \Im{(z)} ] } \right) q^{m^{2}} e^{2 \pi i [ m (z + \frac{1}{2}) ]}.
\end{equation}
Recalling that $q := e^{i \pi \omega}$, we can rewrite the sum over $n$ in terms of a Jacobi theta function
\begin{equation}
  \begin{split}
    \sum_{n \in \ZZ}^{}  |q|^{2 (n + \frac{1}{2})^{2}}  q^{2 (n + \frac{1}{2}) m} e^{2 \pi i [ 2i (n + \frac{1}{2}) \Im{(z)}] }
    & = \sum_{n \in \ZZ}^{} |q|^{2 (n + \frac{1}{2})^{2}} e^{2 \pi i (n + \frac{1}{2}) [ 2 i \Im{(z)} + \omega m]} \\[1ex]
    & = \theta\Big( 2 i \Im{(z)} + \omega m - \frac{1}{2} | 2 i \Im{(\omega)} \Big) 
  \end{split}
\end{equation}
where we have used that $|q| = |e^{i \pi \omega}| = e^{i \pi (i \Im{(\omega)})}$. Hence
\begin{equation}
  \label{eq:theta-squared-1}
  | \theta(z | \omega) |^{2}  
  = \sum_{m \in \ZZ}^{} \theta\Big( 2 i \Im{(z)} + \omega m - \frac{1}{2} | 2 i \Im{(\omega)} \Big) q^{m^{2}}  e^{2 \pi i m (z + \frac{1}{2})}.
\end{equation}

\section{Calculation of Double Commutator}
\label{sec:proof-gs-prop-2a}
We begin this section by recalling the definitions of \(\hat{\mathcal{N}}_{\bvec{k}}\) for \(\bvec{k} \in \mathcal{K}\) and \(\widehat{\rho}(\bvec{q}')\) for \(\bvec{q}' \in (\mathcal{K} + \Gamma^{*}) \setminus \Gamma^{*}\):
\begin{equation}
  \begin{split}
    \hat{\mathcal{N}}_{\bvec{k}} & := \hat{n}_{+, \bvec{k}} + \hat{n}_{-, (-\bvec{k})} \\
    \widehat{\rho}(\bvec{q}') & := \sum_{\bvec{k}}^{} a_{\bvec{k}}(\bvec{q}') \hat{C}_{+, \bvec{k}, (\bvec{k} + \bvec{q}')} + \overline{a_{\bvec{k}}(\bvec{q}')} \hat{C}_{-, \bvec{k}, (\bvec{k} + \bvec{q}')}
  \end{split}
\end{equation}
Using the CAR, one can easily check that
\begin{equation}
  \Big[ \hat{n}_{\pm,\bvec{k}},~ \hat{C}_{\pm, \bvec{k}', (\bvec{k}' + \bvec{q}')} \Big] =
  \begin{cases}
    \hat{C}_{\pm, \bvec{k}, (\bvec{k} + \bvec{q}')} & \bvec{k} = \bvec{k}' \\
    - \hat{C}_{\pm, (\bvec{k} - \bvec{q}'), \bvec{k}} & \bvec{k} = \bvec{k}' + \bvec{q}' \\
    0 & \text{otherwise}
  \end{cases}
\end{equation}
and that \([\hat{n}_{\pm, \bvec{k}}, \hat{C}_{\mp, \bvec{k}', (\bvec{k}' + \bvec{q}')}] = 0 \).
Using this calculation, the commutator between \(\hat{\mathcal{N}}_{\bvec{k}}\) and \(\widehat{\rho}(\bvec{q}')\) can be calculated to be
\begin{equation}
  \begin{split}
    \Big[ \hat{\mathcal{N}}_{\bvec{k}},~ \widehat{\rho}(\bvec{q}') \Big] & = a_{\bvec{k}}(\bvec{q}') \hat{C}_{+, \bvec{k}, (\bvec{k} + \bvec{q}')} - a_{\bvec{k} - \bvec{q}'}(\bvec{q}') \hat{C}_{+, (\bvec{k} - \bvec{q}'), \bvec{k}} \\[.5ex]
                                                                   & \hspace{1em} +  \overline{a_{-\bvec{k}}(\bvec{q}')} \hat{C}_{-, (-\bvec{k}), (-\bvec{k} + \bvec{q}')} - \overline{a_{-\bvec{k} - \bvec{q}'}(\bvec{q}')} \hat{C}_{-, (-\bvec{k} - \bvec{q}'), (-\bvec{k})}.
  \end{split}
\end{equation}
Now notice that \(N_{\bvec{k} + \bvec{q}'}\) only involves momenta \(\bvec{k} + \bvec{q}'\) and \(- (\bvec{k} + \bvec{q}')\).
Since these momenta only appear in the first and last term above, commuting with \(\hat{\mathcal{N}}_{\bvec{k} + \bvec{q}'}\) reduces the number of terms from \(4\) to \(2\):
\begin{equation}
  \Bigg[  \hat{\mathcal{N}}_{\bvec{k} + \bvec{q}'},~ \Big[ \hat{\mathcal{N}}_{\bvec{k}},~ \widehat{\rho}(\bvec{q}') \Big]  \Bigg] 
  = a_{\bvec{k}}(\bvec{q}') \hat{C}_{+, \bvec{k}, (\bvec{k} + \bvec{q}')} - \overline{a_{-\bvec{k} - \bvec{q}'}(\bvec{q}')} \hat{C}_{-, (-\bvec{k} - \bvec{q}'), (-\bvec{k})}.
\end{equation}

\section{Additional background materials on representation theory}
\label{sec:tools-from-repr-appendix}

The following is based on \cite[Chapter 15]{FultonHarris2004}, and another careful presentation with thoroughly worked examples may be found in \cite[Chapter 6]{Hall2015}. For a quantum spin systems perspective on these constructions see \cite[Chapter 3]{Ragone2024}.

Let $\mathfrak{g}$ be a semisimple Lie algebra and let $V$ be a finite dimensional complex vector space. A representation $(\pi,V)$ is a Lie algebra homomorphism $\pi:\mathfrak{g}\to \mathfrak{gl}(V)$, where $\mathfrak{gl}(V)$ denotes the Lie algebra of linear operators acting on $V$ equipped with the usual matrix commutator $[X,Y] = XY-YX$.
Throughout the following one should have in mind $\mathfrak{g} = \mathfrak{su}(d)$, the $d\times d$ skew-adjoint complex matrices. 
Of particular importance is the adjoint representation, the Lie algebra homomorphism $\text{ad}:\mathfrak{g}\to \mathfrak{gl}(\mathfrak{g})$ given for each $Y\in \mathfrak{g}$ by 
\[
    \text{ad}_Y(X) = [Y,X], \qquad \text{for all }X\in \mathfrak{g}.
\]

Since we only care about complex representations, we may freely pass between the real Lie algebra $\mathfrak{g}$ and its complexification $\mathfrak{g}_\CC$. A standard fact \cite[Proposition 3.39]{Hall2015} ensures that every real Lie algebra representation of $\mathfrak{g}$ on $V$ extends uniquely to a complex Lie algebra representation of $\mathfrak{g}_\CC$ on $V$. The complexification of $\mathfrak{g}=\mathfrak{su}(d)$ is $\mathfrak{g}= \mathfrak{sl}(d, \CC)$, the $d\times d$ traceless complex matrices.

\begin{definition}
Let $\mathfrak{g}$ be a complex semisimple Lie algebra. A \emph{Cartan subalgebra} $\mathfrak{h}$ is a maximal abelian subalgebra in $\mathfrak{g}$ such that for every $H\in \mathfrak{h}$, the adjoint representative $\text{ad}_H:\mathfrak{g}\to \mathfrak{g}$ is diagonalizable.
\end{definition}

To find a Cartan subalgebra of $\mathfrak{sl}(d,\CC)$, start with the standard basis of $\CC^d$, denoted $\{e_j\}_{j=1}^d$ with bra-ket notation $\ket{j}:= e_j$. Define $H_j := \ket{j}\bra{j}$ to be the rank-1 projectors onto these basis vectors. Then a Cartan subalgebra is given by the following subspace of diagonal traceless matrices:
\begin{equation} \label{eqn:cartan subalgebra}
    \mathfrak{h} := \left\{ a_1 H_1 + \dots + a_d H_d: \sum_{j=1}^d a_j=0 \right\} \subseteq \mathfrak{sl}(d,\CC).
\end{equation} We will need to work on the dual space $\mathfrak{h}^*$. There are several natural choices for a basis of $\mathfrak{h}^*$. We will simply take the dual basis: define for $i=1,\dots, d$ linear functionals $L_i:\mathfrak{h}\to \CC$ by $L_i(H_j) = \delta_{ij}$ and note that 
\begin{equation} \label{eqn:weights of cartan}
    \mathfrak{h}^* = \text{span}(L_1,\dots, L_d) / \{ L_1+\dots+L_d=0 \}.
\end{equation} We will often write $L_i$ to denote the equivalence class in $\mathfrak{h}^*$ containing $L_i$. Observe that $L_i(H)$ extracts the $i^{th}$ diagonal entry of $H\in \mathfrak{h}$. 

\begin{definition}
    A linear functional $L\in\mathfrak{h}^*$ is called a \emph{weight} of $(\pi,V)$ if there exists a nonzero vector $v\in V$ such that 
    \[
    \pi(H)v = L(H) v \qquad \text{for all } H\in \mathfrak{h}.
    \] The \emph{weight space} corresponding to $L$ is the set $\{v\in V: \pi(H) v = L(H) v \text{ for all } H\in \mathfrak{h}\}$, and the \emph{multiplicity} of $L$ is the dimension of the corresponding weight space. 
\end{definition} Weights are a tool to track the eigenvalues of the simultaneously diagonalized representatives of $\mathfrak{h}$. Indeed, in the setting of~\cref{thm:highest-wt-vector}, we are working on the tensor representation $(\pi,(\CC^d)^{\otimes N})$ and have distinguished a basis for $\mathfrak{h}$. Then, we have that for any fixed $\ket{i_1}\ket{i_2}\dots\ket{i_N} \in (\CC^d)^{\otimes N}$,
\[
    \pi(H_j) \ket{i_1}\ket{i_2}\dots\ket{i_N} = (L_{i_1}(H_j) + L_{i_2}(H_j) + \dots + L_{i_N}(H_j)) \ket{i_1}\ket{i_2}\dots\ket{i_N},
\] so $\ket{i_1}\ket{i_2}\dots\ket{i_N}$ is a vector of weight $L_{i_1}+ L_{i_2} + \dots + L_{i_N}$. In particular, if we define 
\[
    \mu_j:= (L_{i_1}+ L_{i_2} + \dots + L_{i_N})(H_j),
\] then the tuple of eigenvalues $(\mu_j)_{j=1}^d$ is the same data as the weight $L_{i_1}+ L_{i_2} + \dots + L_{i_N}$. In general one uses weights because, unlike eigenvalues, they do not depend on the choice of basis, but in our context we may simply fix a basis at the outset.

The following definition generalizes the notion of raising and lowering operators, allowing us to ``hop'' between simultaneous eigenvectors of the image of the Cartan subalgebra $\pi(\mathfrak{h})$ in the representation. 

\begin{definition}
    A nonzero linear functional $\alpha\in \mathfrak{h}^*$ is a \emph{root} if there exists a nonzero $X\in\mathfrak{g}$ with 
    \[
        [H,X] = \alpha(H) X \qquad \text{for all } H\in \mathfrak{h}. 
    \] In other words, $\alpha$ is a nonzero weight of the adjoint representation $\text{ad}:\mathfrak{g}\to \mathfrak{gl}(\mathfrak{g})$. We call the weight space corresponding to a root $\alpha$ the \emph{root space} of $\alpha$. 
\end{definition}

The finite set of roots, which thought of as vectors in $\mathfrak{h}^*$ with a suitable inner product, is called a root system. The structure of a semisimple Lie algebra is characterized by its root system, which then dictates the structure of representations, thanks to the following  computation: if $(\pi,V)$ is a representation, $\alpha$ is a root with $X_\alpha$ in the associated root space, and $v\in V$ a weight vector of weight $\beta$, then
\begin{align*}
    \pi(H)\pi(X_\alpha)v &= \pi(X_\alpha)\pi(H) v + [\pi(H),\pi(X_\alpha)] v \\
    &= \pi(X_\alpha)\beta(H) v + \alpha(H) \pi(X_{\alpha}) v\\
    &= (\beta(H)+\alpha(H))\pi(X_\alpha) v.
\end{align*} In other words, given a simultaneous eigenvector $v$ of the image of the Cartan subalgebra $\pi(\mathfrak{h})$, then $\pi(X_\alpha)v$ is also a simultaneous eigenvector of $\pi(\mathfrak{h})$. In this sense we may think of roots (or more precisely root vectors) as generalizing the spin ``raising and lowering'' operators from the $\mathfrak{sl}(2,\CC)$ case. 

We may obtain a compressed description of a root system by making a few choices. Roots always come in such positive and negative pairs, so by choosing a hyperplane in $\mathfrak{h}^*$ we may distinguish a set of positive roots. This is akin to paying attention only to the raising operator for $\mathfrak{sl}(2,\CC)$. Further, we may restrict to a set of simple roots, which is a minimal set of roots which generate the rest of the root system by integer linear combinations. The resultant set of positive simple roots $\Delta$ stores a complete description of $\mathfrak{g}$ and will allow us to state the theorem of highest weight. 
For $\mathfrak{sl}(d,\CC)$, such a set is given by
\begin{equation}
    \Delta:= \{L_1 - L_2, L_2-L_3, \dots , L_{d-1}-L_{d}\}.
\end{equation} Given a root $\alpha$, it is not too hard to find an explicit matrix $X_{\alpha}\in \mathfrak{g}$ in the associated root space. In our specific case, the root $L_i-L_j$, $1\leq i<j\leq d$, has the matrix $E_{ij} := \ket{i}\bra{j}$ in its root space, since
\begin{equation}\begin{split}
    [a_1 H_1+ \dots + a_d H_d, E_{ij}] &= (a_i-a_j)E_{ij}\\
    &= (L_i-L_j)(a_1 H_1+ \dots + a_d H_d) E_{ij} . 
\end{split}\end{equation} We now arrive at a key definition, which generalizes the notion of ``highest spin'' of an irreducible representation of $\mathfrak{sl}(2,\CC)$.

\begin{definition}
    Let $\Delta$ be a set of positive simple roots for a semisimple Lie algebra $\mathfrak{g}$, and let $(\pi,V)$ be a representation of $\mathfrak{g}$. A \emph{highest weight vector} $v\in V$ is a vector annihilated by all of the positive simple root vectors $X_\alpha$, i.e. 
    \[
        \pi(X_\alpha) v =0 \qquad \text{for all } \alpha\in \Delta. 
    \] The \emph{highest weight} of a representation is the weight of its highest weight vector. 
\end{definition}

We finally arrive at a simplified statement of the theorem of highest weight.

\begin{theorem}(Theorem of highest weight)
\begin{enumerate}
    \item Every irreducible representation of $\mathfrak{g}$ has a highest weight whose weight space is one dimensional.
    \item Two irreducible representations of $\mathfrak{g}$ are isomorphic if and only if they have the same highest weight.
    \item If $v$ is a highest weight vector of an irreducible representation $(\pi,V)$, then $V$ is cyclically generated by $v$, i.e.
\[
    V = \{ \pi(X) v : X\in \mathfrak{g}\}.
\]
\end{enumerate}
\end{theorem}

As a special case, we obtain the theorem which appeared in the main text. We mentioned already that the (finite-dimensional) complex representations of $\mathfrak{su}(d)$ and its complexification $\mathfrak{sl}(d,\CC)$ are in one-to-one correspondence. We then leverage that since $\mathfrak{su}(d)$ is the Lie algebra of the simply connected Lie group $\mathsf{SU}(d)$, their representations too are in one-to-one correspondence~\cite[Theorem 5.6]{Hall2015}. Of particular note is that a representation of $\mathsf{SU}(d)$ is irreducible if and only if the representation it induces on $\mathfrak{sl}(d,\CC)$ is irreducible.
\begin{theorem} (Theorem of highest weight for $\mathsf{SU}(d)$)
    Every irreducible representation $(\pi, V)$ of $\mathfrak{sl}(d,\CC)$ has a unique 1 dimensional subspace spanned by $v\in V$ such that 
    \[
        v\in \bigcap_{i=1}^{d-1} \ker(\pi(E_{i,i+1})) , 
    \] where $E_{i,i+1} = \ket{i}\bra{i+1} \in \mathfrak{sl}(d,\CC)$. 
    
    Suppose $v$ has weight $\lambda\in \mathfrak{h}^*$. If $(\tilde{\pi},\tilde{V})$ is another irreducible representation with highest weight $\nu\in \mathfrak{h}^*$, then $V\cong \tilde{V}$ as representations of $\mathfrak{g}$. 
\end{theorem}

\subsection{The 2-site Tensor Case for $\mathsf{SU}(2)$}
\label{subsec:2-site tensor decomp sl2}
As an example, we decompose the two-site tensor representation $\CC^2\otimes \CC^2$ into irreps of $\mathsf{SU}(2)$.

First, consider the defining representation $(\Pi,\CC^2)$ of $\mathsf{SU}(2)$, which is given by the map $\Pi(U)=U$ for all $U\in \mathsf{SU}(2)$. This descends to the defining representation of the Lie algebra $\mathfrak{su}(2)$, which uniquely extends to the defining representation $(\pi,\CC^2)$ of $\fraksl(2,\CC)$, which is given by $\pi(X)=X$ for all $X\in\fraksl(2,\CC)$. This has a highest weight vector $\ket{1}$:
\[
    \pi(E_{1,2}) \ket{1} = E_{1,2} \ket{1} = 0.
\] In fact, the highest weight space is 1-dimensional since $E_{1,2}\ket{2}\neq 0$. Thus the defining representation is an irrep of highest weight $L_1$ since 
\[
    \pi(H_1 - H_2)\ket{1} = (H_1-H_2)\ket{1} = \ket{1} = L_1(H_1-H_2)\ket{1}.
\] Let us denote this irrep by $\mathbb{S}_{(1)}(\CC^2)$, or diagrammatically by the box
\[
    \mathbb{S}_{(1)}(\CC^2) = \ydiagram{1}.
\]

We now wish to understand the 2-site tensor representation. At the group level, this is the representation $\Pi:\mathsf{SU}(2)\to GL(\CC^2\otimes \CC^2)$ given by $\Pi(U)=U\otimes U$. This passes to the Lie algebra representation $\pi:\fraksl(2,\CC)\to \mathfrak{gl}(\CC^2\otimes \CC^2)$ given by $\pi(X) = X\otimes \idty + \idty \otimes X$. To decompose this into irreps, we look for highest weight vectors by computing the kernel of $\pi(E_{1,2}) = E_{1,2}\otimes \idty + \idty\otimes E_{1,2}$. We see that this kernel is two dimensional and spanned by the vectors $\ket{1} \ket{1}$ and $\ket{1} \ket{2} - \ket{2} \ket{1}$. We may quickly calculate the weights of these two highest weight vectors:
\begin{equation}\begin{split}
    \pi(H_1-H_2)\ket{1} \ket{1} &= 2L_1(H_1-H_2)\ket{1} \ket{1} \\ 
    \pi(H_1-H_2) \Big( \ket{1} \ket{2} - \ket{2} \ket{1} ) &= (L_1+L_2)(H_1-H_2) \Big( \ket{1} \ket{2} - \ket{2} \ket{1} \Big) = 0.
\end{split}\end{equation} So, $\ket{1} \ket{1}$ is a vector of highest weight $2L_1$, while $\ket{1} \ket{2} - \ket{2} \ket{1}$ is a vector of highest weight $0$, where we mention that $L_1+L_2=0$ by virtue of the quotient defining $\mathfrak{h}^*$ in \cref{eqn:weights of cartan}. Thus, they both cyclically generate irreps, which we will notationally call
\begin{equation}
    \mathbb{S}_{(2)}(\CC^2) := \gen_{\fraksl(2,\CC)} \{ \ket{1} \ket{1} \}, \qquad \mathbb{S}_{(1^2)}(\CC^2) := \gen_{\fraksl(2,\CC)} \{ \ket{1} \ket{2} - \ket{2} \ket{1} \}.
\end{equation}
For the sake of being explicit, we may write bases for these spaces by e.g. repeatedly applying the ``lowering'' operator $\pi(E_{2,1})$:
\begin{equation}
    \mathbb{S}_{(2)}(\CC^2) = \text{span}\{\ket{1} \ket{1}, \ket{1} \ket{2}+\ket{2} \ket{1}, \ket{2} \ket{2}\}, \qquad  \mathbb{S}_{(1^2)}(\CC^2) = \text{span}\{\ket{1} \ket{2} - \ket{2} \ket{1}\}.
\end{equation} Indeed, $\mathbb{S}_{(2)}(\CC^2)$ consists of vectors invariant under the permutation which swaps two sites, while $\mathbb{S}_{(1^2)}(\CC^2)$ consists of vectors which are antisymmetric with respect to this permutation.

In particular note that we have completely decomposed $\CC^2\otimes \CC^2$ into a direct sum of irreps:\footnote{To connect to the language of spin in physics, the irrep $\mathbb{S}_{(1^2)}(\CC^2)$ is often called a singlet, while the irrep $\mathbb{S}_{(2)}(\CC^2)$ is often called a triplet.}
\[
    \CC^2 \otimes \CC^2 \cong \mathbb{S}_{(2)}(\CC^2) \oplus \mathbb{S}_{(1^2)}(\CC^2) . 
\] We will conclude with the suggestive diagrammatic version of the above isomorphism:
\begin{equation} \label{eqn:sl2 2 tensor ydiagram}
    \ydiagram{1}\otimes \ydiagram{1} \cong \ydiagram{2} \oplus \ydiagram{1,1}.
\end{equation}

\subsection{The 2-site Tensor Case for $\mathsf{SU}(4)$} 
\label{subsec:2-site tensor decomp sl4}

The treatment for $\mathsf{SU}(4)$ is similar to that of $\mathsf{SU}(2)$, but now the Cartan subalgebra $\mathfrak{h}$ of $\mathfrak{sl}(4,\CC)$ defined by~\cref{eqn:cartan subalgebra} is 3-dimensional.

First, the defining representation $(\pi,\CC^4)$ given by $\pi(X)=X$ has the highest weight vector $\ket{1}$:
\[
    \ket{1}\in \ker \pi(E_{1,2})\cap \ker \pi(E_{2,3}) \cap \ker \pi(E_{3,4}). 
\] Evidently this is the only vector (up to rescaling) of highest weight. Thus the defining representation is an irrep of highest weight $L_1$ since
\[
    \pi(H_i - H_{i+1}) \ket{1} = L_1(H_{i}-H_{i+1}) \ket{1}, \qquad i=1,2,3.
\] Let us denote this irrep by $\mathbb{S}_{(1)}(\CC^4)$, or diagrammatically by the box
\[
     \mathbb{S}_{(1)}(\CC^4) = \ydiagram{1}.
\]

On to the 2-site tensor representation $\pi:\fraksl(4,\CC) \to \mathfrak{gl}(\CC^4\otimes \CC^4)$, given by $\pi(X) = X\otimes \idty +\idty \otimes X$. We once again compute that the space of highest weight vectors is two-dimensional:
\[
    \text{span}\{\ket{1} \ket{1}, \ket{1} \ket{2} - \ket{2} \ket{1}\} = \ker \pi(E_{1,2})\cap \ker \pi(E_{2,3}) \cap \ker \pi(E_{3,4}) . 
\] The weights of these vectors are \begin{equation}\begin{split}
  \pi(H_1-H_2)\ket{1} \ket{1} &= 2L_1(H_1-H_2)\ket{1} \ket{1} \\ 
    \pi(H_1-H_2)\Big( \ket{1} \ket{2} - \ket{2} \ket{1} \Big) &= (L_1+L_2)(H_1-H_2) \Big( \ket{1} \ket{2} - \ket{2} \ket{1} \Big) = 0.
\end{split}\end{equation} So, $\ket{1} \ket{1}$ is a vector of highest weight $2L_1$, while $\ket{1} \ket{2} - \ket{1} \ket{2}$ is a vector of highest weight $L_1+L_2$. Hence, they both cyclically generate irreps, which we will notationally call
\begin{equation}
    \mathbb{S}_{(2)}(\CC^4) := \gen_{\fraksl(4,\CC)} \{ \ket{1} \ket{1} \}, \qquad \mathbb{S}_{(1^2)}(\CC^4) := \gen_{\fraksl(4,\CC)} \{\ket{1} \ket{2} - \ket{2} \ket{1}  \},
\end{equation} again yielding the decomposition $\CC^4\otimes \CC^4 \cong \mathbb{S}_{(2)}(\CC^4) \oplus \mathbb{S}_{(1^1)}(\CC^4)$ with an identical diagrammatic representation to~\cref{{eqn:sl2 2 tensor ydiagram}}. In this case, by counting symmetric and antisymmetric vectors. we obtain that $\dim(\mathbb{S}_{(2)}(\CC^4)) = 4(4+1)/2 = 10$ and $\dim(\mathbb{S}_{(1^2)}(\CC^4)) = 4(4-1)/2 = 6$.

\end{document}